\documentclass[letterpaper,11pt]{article}
\usepackage{fullpage}
\usepackage{graphicx}
\usepackage{comment}
\usepackage{mathtools}
\usepackage[linesnumbered,ruled,vlined]{algorithm2e}
\usepackage[utf8]{inputenc}
\usepackage{amssymb}
\usepackage{tikz}
\usepackage{amsthm}
\usepackage{subcaption}
\usepackage{amsfonts}
\usepackage{thmtools}
\usepackage{thm-restate}
\usepackage[colorlinks=true]{hyperref}
\usepackage[capitalize]{cleveref}
\usepackage[framemethod=TikZ]{mdframed}
\usepackage[normalem]{ulem}
\usepackage{filecontents}
\usepackage{float}

\usepackage[draft]{fixme}
\fxsetup{mode=multiuser,theme=color, layout=inline}
\FXRegisterAuthor{ss}{ass}{\color{red} \bf Saeed}
\FXRegisterAuthor{tk}{atk}{\color{blue} \bf  Tomasz}

\renewcommand{\bar}[1]{\mkern 1mu\overline{\mkern-1mu#1\mkern-1mu}\mkern 1mu}

\newtheorem{theorem}{Theorem}
\newtheorem{lemma}[theorem]{Lemma}
\newtheorem{corollary}[theorem]{Corollary}
\newtheorem{claim}[theorem]{Claim}
\newtheorem{observation}[theorem]{Observation}
\newtheorem{fact}[theorem]{Fact}
\newtheorem{definition}[theorem]{Definition}

\newtheorem{proposition}[theorem]{Proposition}

\crefname{observation}{Observation}{Observations}

\newtheorem{innercustomthm}{Theorem}
\newenvironment{customthm}[1]
{\renewcommand\theinnercustomthm{#1}\innercustomthm}
{\endinnercustomthm}

\newtheorem{innercustomlem}{Lemma}
\newenvironment{customlem}[1]
{\renewcommand\theinnercustomlem{#1}\innercustomlem}
{\endinnercustomlem}

\newcounter{theo}[section] 

\newcommand{\ogood}{\tilde O_{\epsilon,\kappa}}
\newcommand{\opt}{\mathsf{opt}}

\renewcommand{\tilde}{\widetilde}
\newcommand{\Ohtilde}{\tilde{O}}
\newcommand{\R}{\mathbb{R}}
\newcommand{\dd}{\mathinner{.\,.}}
\newcommand{\dist}{\mathrm{dist}}
\newcommand{\sub}{\subseteq}
\newcommand{\sm}{\setminus}
\newcommand{\AG}{\mathsf{AG}}
\newcommand{\SAG}{\mathsf{SAG}}
\newcommand{\Oh}{O}

\newcommand{\LIS}{\mathsf{LIS}}

\newcommand{\tpi}{\tilde{\pi}}
\newcommand{\tx}{\tilde{x}}
\newcommand{\bx}{\bar{x}}
\newcommand{\by}{\bar{y}}

\newcommand{\bp}{\bar{p}}
\newcommand{\ty}{\tilde{y}}
\newcommand{\tS}{\tilde{S}}
\newcommand{\tG}{\tilde{G}}

\newcommand{\tsigma}{\tilde{\sigma}}
\newcommand{\row}{\mathsf{xi}}
\newcommand{\col}{\mathsf{yi}}
\newcommand{\Lft}{\mathcal{L}}
\newcommand{\Rgt}{\mathcal{R}}

\newcommand{\Zz}{\mathbb{Z}_{\ge 0}}

\newcommand{\dx}{\mathsf{dx}_{\bar{x}}}

\newcommand{\dl}{\mathsf{d}_{X,Y}}
\newcommand{\rank}{\mathsf{rk}}

\newcommand{\Gr}{\mathcal{G}}
\newcommand{\Diag}{\mathcal{D}}
\usepackage[nocompress]{cite}
\begin{document}
	
	\title{Improved Dynamic Algorithms for Longest Increasing Subsequence}
	\author{Tomasz Kociumaka\\
		University of California, Berkeley\\
		\texttt{kociumaka@berkeley.edu}
		\and
		Saeed Seddighin\footnote{Supported in part by an Adobe research award and a Google research gift.}\\
		Toyota Technological Institute at Chicago\\
		\texttt{saeedreza.seddighin@gmail.com}
	}
	\date{}
	\definecolor{printable}{RGB}{43,70,120}
	\maketitle

	\begin{abstract}
		\setcounter{page}{0}
We study dynamic algorithms for the longest increasing subsequence (\textsf{LIS}) problem.
A dynamic \textsf{LIS} algorithm maintains a sequence subject to operations of the following form arriving one by one:
(i) insert an element, (ii) delete an element, or (iii) substitute an element for another. 
After performing each operation, the algorithm must report the length of the longest increasing subsequence
of the current sequence.

Our main contribution is the first exact dynamic \textsf{LIS} algorithm with sublinear update time. More precisely, we present a randomized algorithm that performs each operation in time $\tilde O(n^{2/3})$ and after each update, reports the answer to the \textsf{LIS} problem correctly with high probability.
We use several novel techniques and observations for this algorithm that may find their applications in future work.

In the second part of the paper, we study approximate dynamic \textsf{LIS} algorithms, which are allowed to underestimate the solution size within a bounded multiplicative factor.
In this setting, we give a deterministic algorithm with update time $O(n^{o(1)})$ and approximation factor $1-o(1)$. This result substantially improves upon the previous work of Mitzenmacher and Seddighin (STOC'20) that presents an $\Omega(\epsilon ^{O(1/\epsilon)})$-approximation algorithm with update time $\tilde O(n^\epsilon)$ for any constant $\epsilon > 0$.

\thispagestyle{empty}
\newpage
	\end{abstract}
	
	\section{Introduction}\label{sec:intro}
Longest increasing subsequence (\textsf{LIS}) is a very old and classic problem in computer science. In this problem, a sequence $a = \langle a_1, a_2, \ldots, a_n\rangle$ of size $n$ is given as input and the \textsf{LIS} of the sequence is defined as the largest subset of the elements whose values are strictly increasing in the order of their indices.  \textsf{LIS} can also be thought of as a special case of the longest common subsequence (\textsf{LCS}) problem where the two inputs are permutations. Both problems date back to the 1950s and have been subject to a plethora of research~\cite{fredman1975computing,ramanan1997tight,DBLP:conf/soda/GopalanJKK07,DBLP:conf/focs/GalG07,DBLP:conf/stoc/ErgunKKRV98, DBLP:conf/random/DodisGLRRS99,DBLP:journals/eatcs/Fischer01,DBLP:conf/approx/AilonCCL04} especially in recent years~\cite{hajiaghayi2019approximating,saeedfocs19,our-stoc-paper,our-soda-paper}.

In this work, we focus on exact and approximation algorithms in the dynamic setting, where at each step, the sequence can be updated by inserting, deleting, or substituting an element.  The goal is to maintain the size of the longest increasing subsequence.
Many problems have been studied in dynamic settings; see e.g.~\cite{DBLP:conf/stoc/HenzingerKNS15,DBLP:conf/stoc/NanongkaiS17,gawrychowski2018optimal,DBLP:conf/stoc/AssadiOSS18,DBLP:conf/soda/AssadiOSS19,chen2013dynamic,DBLP:conf/stoc/LackiOPSZ15,DBLP:journals/corr/abs-1909-03478,DBLP:conf/focs/NanongkaiSW17}.
In general, in a dynamic setting, the goal is to develop an algorithm which updates the solution efficiently given incremental changes to the input.
In the context of graph algorithms~\cite{DBLP:conf/stoc/NanongkaiS17,DBLP:conf/focs/NanongkaiSW17,DBLP:conf/stoc/LackiOPSZ15,DBLP:conf/stoc/AssadiOSS18,DBLP:conf/soda/AssadiOSS19,DBLP:journals/corr/abs-1909-03478}, such changes are usually modeled by edge insertion or deletion. For string problems, changes are typically modeled with character insertion, deletion, and substitution~\cite{chen2013dynamic,CGP20,CKM20}, as we consider here. 

Our work is closely related to two previous works on dynamic \textsf{LIS}. (In what follows, when we refer to a solution, we typically refer to the size of the \textsf{LIS}, but also the corresponding increasing subsequence can be found in time proportional to its size.) Mitzenmacher and Seddighin~\cite{our-stoc-paper} give an $\Omega(\epsilon ^{O(1/\epsilon)})$-approximation algorithm for dynamic \textsf{LIS} whose update time is bounded by $\tilde O(n^\epsilon)$ for any constant $\epsilon > 0$. 
They also present a $(1+\epsilon)$-approximation algorithm for the dynamic variant of distance to monotonicity which is the dual of \textsf{LIS}. Their solution for distance to monotonicity requires polylogarithmic update time. 
Chen, Chu, and Pinkser~\cite{chen2013dynamic} also study dynamic \textsf{LIS} and give an exact solution whose update time depends on the size of the solution. More precisely, if we denote the solution size by $\opt$, then their algorithm requires update time $O(\opt \log \frac{n}{\opt})$. Notice that $\opt$ can be as large as $\Omega(n)$ and therefore their update time is $\Theta(n)$ in the worst case.

\definecolor{LightCyan}{rgb}{0.88,1,1}
\begin{table}[b!]
    \renewcommand{\arraystretch}{1.2}
	\centering
	\begin{tabular}{|l|c|c|c|}
		\hline
		Approximation factor & Update time & Reference\\
		\hline
		$1-\epsilon$ & $\tilde O(\sqrt{n})$ & \cite{our-stoc-paper}\\
		\hline	
        $\Omega(\epsilon^{O(1/\epsilon)})$ & $\tilde O(n^{\epsilon})$ & \cite{our-stoc-paper}\\
        \hline 
        exact & $O(\opt \log \frac{n}{\opt})$ & \cite{chen2013dynamic}\\
        \hline 
        exact & $O(n^{2/3} \log^4 n)$& \cref{thm:exact2}\\
        \hline
        $1-o(1)$ & $O(n^{o(1)})$ & \cref{theorem:third}\\
        \hline
	\end{tabular}
\caption{The results of this paper along with previous work on dynamic \textsf{LIS}.
Here, $\opt$ denotes the size of the longest increasing subsequence.
In the algorithm of \cref{theorem:third}, if we denote the approximation factor by $1-\epsilon$, the update time will be bounded by $O((\log n/\epsilon)^{O((\log n)^{2/3}/\epsilon)})$.}
\end{table}

We provide the first exact algorithm for \textsf{LIS} with sublinear update time regardless of the solution size. In other words, we give an algorithm that reports the correct value of the size of the longest increasing subsequence with high probability (at least $1-n^{-10}$) and, after each operation, updates the solution in time $O(n^{2/3} \log^4 n)$.
In addition to this, we also present a $(1-o(1))$-approximation algorithm for dynamic \textsf{LIS} that updates the solution in time $O(n^{o(1)})$.
Our method substantially advances the techniques of~\cite{our-stoc-paper} and significantly improves their main results. 

\subsection{Parallel and Independent Work}
Parallel to and independent of this work, Gawrychowski and Janczewski~\cite{gawrychowski2020fully} present a different $1-o(1)$ approximation algorithm for dynamic \textsf{LIS} that updates the solution in subpolynomial time. While the approximation factors of both algorithms are $1-o(1)$, their algorithm is faster: it achieves update time $O(\epsilon^{-5}\log^{11} n)$ while maintaining a $(1-\epsilon)$-approximation of \textsf{LIS}.
To derive a solution, both approximation algorithms generalize the dynamic \textsf{LIS} problem; our algorithm is able to answer more general queries, though.
We also remark that~\cite{gawrychowski2020fully} does not claim any progress towards obtaining an exact sublinear-time solution for dynamic \textsf{LIS} (which is our main result).

\subsection{Related Work}
In addition to previous dynamic algorithms for \textsf{LIS}~\cite{our-stoc-paper,chen2013dynamic}, this problem has received significant attention in other areas such as property testing~\cite{DBLP:conf/stoc/ErgunKKRV98, DBLP:conf/random/DodisGLRRS99,DBLP:journals/eatcs/Fischer01,DBLP:conf/approx/AilonCCL04}, streaming~\cite{DBLP:conf/soda/GopalanJKK07,DBLP:conf/focs/GalG07},
and massively parallel computation (MPC)~\cite{DBLP:conf/stoc/ImMS17}, as well as in the standard algorithmic setting~\cite{fredman1975computing,ramanan1997tight,saeedfocs19,saks2010estimating,our-soda-paper}. 

The classic \textit{patience sorting} solution for \textsf{LIS}  utilizes dynamic programming and binary search to solve \textsf{LIS}
exactly in time $O(n \log n)$.  
Matching lower bounds ($\Omega(n \log n)$) are known for comparison-based algorithms~\cite{fredman1975computing} and solutions based on algebraic decision trees~\cite{ramanan1997tight}.
For approximation algorithms, for any $\epsilon > 0$, a multiplicative $\Omega(n^{-\epsilon})$ approximate solution can be determined in truly sublinear time\footnote{Truly sublinear stands for $O(n^{1-\Omega(1)})$.} via random sampling\footnote{For an $\Omega(n^{-\epsilon})$-approximation algorithm, one can sample $O(n^{1-\epsilon})$ elements from the sequence and report the \textsf{LIS} of those samples.}. 
Surprisingly, not much is known that improves upon this algorithm generally, although when $n/\mathsf{LIS}(a)$ is subpolynomial in $n$, we can obtain better approximation guarantees  for \textsf{LIS}~\cite{saks2010estimating,saeedfocs19,our-soda-paper}.

From a complexity point of view, unconditional lower bounds apply to sublinear-time algorithms for \textsf{LIS}. For instance, any algorithm that obtains a $1/f(n)$-approximate solution for \textsf{LIS} has to make at least $n/(f(n)+1)$ value queries\footnote{A value query provides an $i$ as input and asks for the value of $a_i$.} to the elements of $a$ to distinguish the case that $a$ is decreasing from the case that $a$ has an increasing subsequence of length at least $f(n)+1$. Thus a subpolynomial-factor approximation algorithm for \textsf{LIS} in truly sublinear time is not possible in general. 
In contrast, positive results are given in previous work for a special case in which the solution size is at least $\lambda n$ for large enough $0 < \lambda \leq 1$. (Known query complexity lower bounds do not apply to this setting.) Saks and Seshadhri~\cite{saks2010estimating} present a $(1-\epsilon)$-approximation algorithm  for \textsf{LIS} in this case. The runtime of their algorithm is sublinear as long as $\lambda > \log \log n / \log n$. Moreover,  Rubinstein,
Seddighin, Song, and Sun~\cite{saeedfocs19} give an $\Omega(\lambda^3)$-approximation algorithm for this case in time $\tilde O(\sqrt{n}/\lambda^7)$. Very recently, Mitzenmacher and Seddighin~\cite{our-soda-paper}  improve the approximation factor to $\Omega(\lambda^\epsilon)$ for any constant $\epsilon > 0$ while keeping the runtime truly sublinear.

From a technical standpoint, our approximation algorithm is related to the recent results of~\cite{our-stoc-paper} and~\cite{our-soda-paper} that approximate \textsf{LIS} in the dynamic and standard settings, respectively.
Both these works use the grid packing technique for \textsf{LIS} to design their algorithms. We generalize this notion and prove that the generalized grid packing gives improved dynamic algorithms for \textsf{LIS}.
	\subsection{Preliminaries}
\textsf{LIS} is defined on a sequence of numbers. We assume for simplicity that all of the numbers are distinct and positive integers. In this problem, the goal is to find the length of the largest subsequence of elements such that their values increase according to their indices. We denote the size of the sequence by $n$ and use $a_1, \ldots, a_n$ to denote the sequence elements. We also give an alternative definition for the problem which represents the input as $n$ points on the 2D plane. In this representation, we have $n$ points on the plane with distinct coordinates. Similarly, we assume for simplicity that all the coordinates are positive integers. For any subset of points, its \textsf{LIS} is defined as the largest number of points such that if we sort them based on their $x$ or $y$ coordinates, we obtain the same ordering. 

We adapt the setting of~\cite{our-stoc-paper} for the dynamic \textsf{LIS} problem. Initially, the input sequence is empty ($|a| = 0$). At each step, an element can either be inserted at an arbitrary position of the sequence or removed from an arbitrary position of the sequence. (Element substitution can also be implemented with the previous two operations, so we consider only insertions and deletions.) 

We now define the operations more formally. Each insertion operation is of the form ``\textsf{insert $(i,x)$}" where $i$ is an integer between $1$ and the length of the current sequence plus one. The index $i$ specifies the position of element $x$. After this operation, all the elements whose previous index was at least $i$ will be shifted to the right. Similarly, an operation ``\textsf{delete $(i)$}" removes the $i$'th element. Likewise, all the elements whose previous index was at least $i$ will be shifted to the left. As discussed in previous work~\cite{our-stoc-paper}, one can use a balanced tree data structure that provides access to any element of the sequence in time $O(\log n)$. Thus, in the analysis of our time bounds for dynamic \textsf{LIS}, we consider an additional multiplicative $O(\log n)$ overhead for using this data structure and assume that random access is provided to any element.

To simplify the explanations, we often use the notation $\tilde O$ that hides the logarithmic factors. When other parameters such as $\epsilon$ or $\kappa$ are involved, we may use $\tilde O_{\epsilon}$ or $\tilde O_{\kappa}$ notations that also hide factors that only depend on $\epsilon$ or $\kappa$. Similarly, $\ogood$ hides all the factors that depend on $\epsilon$, $\kappa$, or (at most polynomially) on $\log n$.

	\section{Our Results and Techniques}
In this section, we present our results and techniques.
Our main contribution is a dynamic algorithm for \textsf{LIS} that maintains an exact solution with update time $\tilde O(n^{2/3})$.
We explain the high-level ideas behind this algorithm in Section~\ref{sec:results-exact}.
We then proceed by bringing the ideas behind our $(1-o(1))$-approximation algorithm in Section~\ref{sec:results-approx}. 
For both our dynamic algorithms, we use the notion of block-based algorithms~\cite{our-stoc-paper} to simplify the exposition.
This enables us to include preprocessing steps which may violate the worst-case update times.
Yet, a block-based algorithm can be turned into a dynamic algorithm whose worst-case update time is asymptotically equal to the block-based algorithm's amortized update time.
More precisely, a block-based algorithm starts with an array $a$ of size $n$. 
It is allowed to preprocess the array in time $f(n)$.
After the preprocessing step, the algorithm is required to execute $g(n)$ operations, each in worst-case time $h(n)$. 
After $g(n)$ operations, the block-based algorithm terminates.
It follows from~\cite{our-stoc-paper} that such an algorithm can be used to design a dynamic algorithm for \textsf{LIS} with worst-case update time $O(f(n)/g(n) + h(n))$.

\subsection{An Exact Algorithm with Sublinear Update Time}\label{sec:results-exact}
Our main result is an exact dynamic algorithm for \textsf{LIS} with sublinear update time. For this algorithm, we use several combinatorial techniques which we explain in the following. 
The first idea is a randomized coloring argument which lets us decompose the problem into smaller subproblems.
Recall that the algorithm of Chen, Chu, and Pinsker~\cite{chen2013dynamic} provides the exact \textsf{LIS} value with update time $O(\opt \log \frac{n}{\opt})$, where $\opt$ is the solution size. Thus, as long as $\opt$ is sublinear in $n$, their algorithm updates the solution in sublinear time. Intuitively, this signals that the real difficulty of the problem is for the case where the \textsf{LIS} is very large, namely, of size $\tilde \Omega(n)$.

We define $b_i$ as the size of the longest increasing subsequence ending at element $a_i$ and divide the sequence into disjoint layers. More precisely, let each layer $L_i = \{a_j \mid b_j = i\}$ be the set of elements whose corresponding $b_j$ is equal to $i$.
Indeed, the number of distinct layers is equal to the size of the solution and therefore, when the solution size is large, we expect that the size of the layers is small on average.
As an example, if the \textsf{LIS} of the sequence is of size $\Omega(n)$, then we expect the average size of the layers to be $O(1)$. The following observation enables us to decompose the \textsf{LIS} problem into smaller subproblems that can be updated independently: Let $\opt$ be the number of layers for a sequence. If, for some integer $t$, we color $t$  layers of the sequence uniformly at random and perform $\opt / (10t)$ arbitrary operations on the sequence, with probability at least $1/2$ there exists a longest increasing subsequence in the new sequence that has exactly $t$ colored elements.

A formal proof for the above claim is given in Section~\ref{sec:exact}. Since the element values are decreasing in each layer of the original sequence, no layer can contribute more than one element to any increasing subsequence. Moreover, since the solution size for the original sequence is $\opt$, after $\opt / (10t)$ operations, the size of the \textsf{LIS} changes by at most an additive $\opt / (10t)$ term. Therefore, if we fix a longest increasing subsequence in the new array (after all the operations are performed) and denote its size by $\opt'$, at least $\opt - \opt / (10t)$ layers of the original sequence contribute to $\opt'$.
Since we color $t$ layers uniformly at random, this proves that, with probability at least $1/2$, all of the colored layers contribute to $\opt'$.

Let us go back to our previous discussion. 
In a block-based algorithm, we can in time $O(n \log n)$ construct the layers and sample $t$ layers uniformly at random. 
For the next $g(n) = \opt/(10 t)$ operations, we can be sure that with probability at least $1/2$ we have a longest increasing subsequence that goes through all sampled layers. 
For now, we ignore the bad event and assume for simplicity that this property always holds. Thus, we only need to keep partial solutions between consecutive sampled layers. 
Since our aim is to design a dynamic algorithm for the case that the size of the longest increasing subsequence is large, we expect that the layer sizes are small.
Hence, assume that for every pair of elements $a_i$ and $a_j$ that belong to two consecutive sampled layers we are given the size of the longest increasing subsequence starting from $a_i$ and ending at $a_j$, and for every element of the first sampled layer we have the size of the longest increasing subsequence ending at that element. 
Similarly, assume that for each element of the last sampled layer, the size of the longest increasing subsequence starting from that element is available.
It follows that based on this information, we can recover the \textsf{LIS} of the whole sequence (ignoring the bad event). In the extreme case that the solution size $\opt$ is $\Omega(n)$, we expect the layer sizes to be $O(1)$ on average, which makes the total size of the information to be stored small.

The benefit of the above approach is obvious for the dynamic setting: If we only care about local solutions between consecutive sampled layers, whenever an operation is performed, we only need to update the local solutions. For this purpose, we only consider elements located between the two consecutive sampled layers. We remark that the positions of such elements in the sequence does not necessarily form an interval. On average, the expected number of such elements is $O(n/t)$. We bring an example to clarify the advantage of this approach: Assume that the \textsf{LIS} of the original sequence is equal to $\Omega(n)$ and the size of each layer is bounded by $O(1)$. If we set $t = \sqrt{n}$ and sample $t$ layers uniformly at random, we expect to have a solution that goes through all the sampled layers with constant probability for up to $\sqrt{n}/10$ steps. Thus, we set $g(n) = \sqrt{n}/10$ and assume that our block-based algorithm is only responsible for performing $g(n)$ operations. For simplicity, we ignore the bad event that a solution may not go through all the sampled layers. 
Since the size of each layer is $O(1)$, every time an operation arrives, we need to update the solution for at most $O(1)$ pairs of elements, and this can be done in time $\tilde O(\sqrt{n})$ since with high probability there are at most $\tilde O(\sqrt{n})$ layers between two consecutive sampled layers (each having $O(1)$ elements).
Thus, the update time is $\tilde O(\sqrt{n})$ and once the solutions between consecutive sampled layers are provided, we can find the longest increasing subsequence of the entire sequence in time $\tilde O(\sqrt{n})$.
Therefore, with preprocessing time $f(n) = \tilde O(n)$ and $g(n) = \Omega(\sqrt{n})$, we can update the solution in time $h(n) = \tilde O(\sqrt{n})$, which leads to a dynamic algorithm with update time $\tilde O(\sqrt{n})$.

\begin{figure}[ht]

\centering

\tikzset{every picture/.style={line width=0.75pt}} 

\begin{tikzpicture}[x=0.75pt,y=0.75pt,yscale=-0.6,xscale=0.6]

\draw   (310.5,47) -- (339.5,47) -- (339.5,334) -- (310.5,334) -- cycle ;
\draw   (270.5,171) -- (299.5,171) -- (299.5,334) -- (270.5,334) -- cycle ;
\draw   (230.5,171) -- (259.5,171) -- (259.5,334) -- (230.5,334) -- cycle ;
\draw   (190.5,247) -- (219.5,247) -- (219.5,334) -- (190.5,334) -- cycle ;
\draw   (150.5,248) -- (179.5,248) -- (179.5,334) -- (150.5,334) -- cycle ;
\draw   (110.5,247) -- (139.5,247) -- (139.5,334) -- (110.5,334) -- cycle ;
\draw   (70.5,247) -- (99.5,247) -- (99.5,334) -- (70.5,334) -- cycle ;
\draw   (390.5,171) -- (419.5,171) -- (419.5,334) -- (390.5,334) -- cycle ;
\draw   (350.5,171) -- (379.5,171) -- (379.5,334) -- (350.5,334) -- cycle ;
\draw   (550.5,246) -- (579.5,246) -- (579.5,334) -- (550.5,334) -- cycle ;
\draw   (510.5,247) -- (539.5,247) -- (539.5,334) -- (510.5,334) -- cycle ;
\draw   (470.5,246) -- (499.5,246) -- (499.5,334) -- (470.5,334) -- cycle ;
\draw   (430.5,246) -- (459.5,246) -- (459.5,334) -- (430.5,334) -- cycle ;

\draw (296,6.4) node [anchor=north west][inner sep=0.75pt]    {$\frac{n}{2 \log n}$};
\draw (238,130.4) node [anchor=north west][inner sep=0.75pt]    {$\frac{n}{4 \log n}$};
\draw (113,210.4) node [anchor=north west][inner sep=0.75pt]    {$\frac{n}{8 \log n}$};
\draw (-9,307.4) node [anchor=north west][inner sep=0.75pt]  [font=\Huge]  {$\dotsc $};
\draw (593,307.4) node [anchor=north west][inner sep=0.75pt]  [font=\Huge]  {$\dotsc $};
\draw (358,130.4) node [anchor=north west][inner sep=0.75pt]    {$\frac{n}{4 \log n}$};
\draw (473,210.4) node [anchor=north west][inner sep=0.75pt]    {$\frac{n}{8 \log n}$};

\end{tikzpicture}

\caption{The sizes of the layers in an example are shown in the figure. The middle layer has the largest size and the sizes decrease exponentially as the layers move to the sides.} \label{fig:bad-layers}
\end{figure}
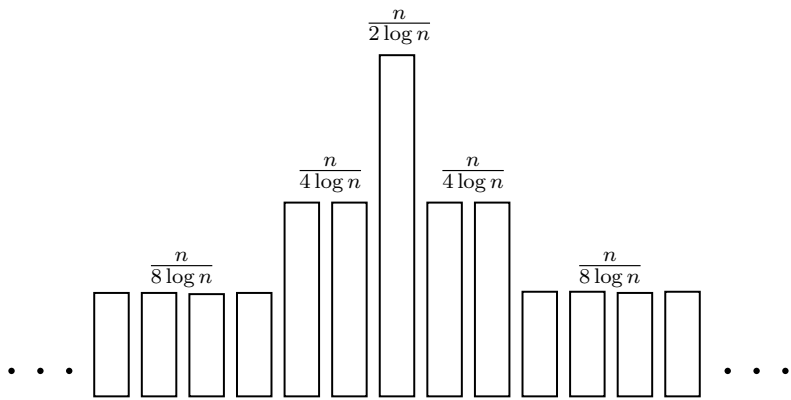

Roughly speaking, if we only aim to obtain a dynamic algorithm with sublinear update time, the assumption that the solution size is $\tilde\Omega(n)$ does not overly simplify the problem since otherwise we can use the algorithm of Chen, Chu, and Pinsker~\cite{chen2013dynamic} whose update time is sublinear in that case.
The probability $1/2$ that a solution may fail to go through all the sampled layers (failure probability) can be conveniently  reduced to $1-n^{-\Omega(1)}$ via standard techniques at the price of an $O(\log n)$ multiplicative overhead to the runtimes. 
However, one of the assumptions above does, in fact, oversimplify the problem and may give rise to fundamental issues in general. 
Even if we assume the solution size is $\tilde \Omega(n)$, we can only guarantee that the size of the layers is small on average and not in the worst case. To further clarify this issue, consider the example given in Figure~\ref{fig:bad-layers}.

The example of Figure~\ref{fig:bad-layers} shows why the above idea alone does not provide a sublinear-time algorithm. Although the size of the \textsf{LIS} is $\tilde \Omega(n)$ and thus the average layer size is $\tilde O(1)$, the two sampled layers that sandwich the middle layer do not yield an easier subproblem. Either the size of one of the sampled layers is large or the  number of layers included between them is large. More precisely, the size of the larger sampled layer multiplied by the number of layers between them is $\tilde \Omega(n)$. Therefore, regardless of whether we naively use the  patience sorting algorithm to update the local solutions or we use the more sophisticated algorithm of Chen, Chu, and Pinsker~\cite{chen2013dynamic}, the time required to update the local solutions for such sampled layers is $\tilde \Omega(n)$.

\input{figs/baskets2}

To resolve this issue, we devise a heavy-light decomposition technique to deal with different subproblems. In our algorithm, we make a basket out of consecutive layers that end with a sampled layer. Thus, we refer to each of the sampled layers as a \textit{boundary layer}. For simplicity, we omit some of the details of our algorithm here and only state the overall ideas. (For instance, in our algorithm we do not sample the layers completely at random to make the boundary layers.) Each basket then becomes one subproblem.
We define a basket to be \textit{light} if the total number of elements included in it as well as the size its boundary layer are both small. In other words, for each of the light baskets, we can store and update the local solutions in sublinear time.
A basket is \textit{heavy} if either its boundary size is large or it contains a large number of elements. In the example of Figure~\ref{fig:bad-layers}, the basket that contains the middle layer is a heavy basket. For such baskets, we do not store local solutions but, instead, we maintain global information: for each element of a heavy basket, we store the size of the longest increasing subsequence of the entire sequence ending that element.

Because the solutions for heavy baskets are not local, once an operation arrives, after updating the local solution for the corresponding basket (if it is light), we need to propagate the effect of the change. That is, if the local solution for a light basket is affected by a modification, we need to update our solution for all the heavy baskets. To keep the update time sublinear, we make one more observation: It is possible to update the solution for a heavy basket in time proportional to the number of layers included in it. For this purpose, we use the ideas of the work by Chen, Chu, and Pinsker~\cite{chen2013dynamic} that design balanced trees to obtain a dynamic solution for \textsf{LIS} with update time proportional to the solution size.
In the interest of space, here we omit several details of our algorithm; we prove in Section~\ref{sec:exact} that the combination of these ideas gives us an exact dynamic algorithm for \textsf{LIS} with update time $\tilde O(n^{4/5})$.

\begin{customthm}{\ref{theorem:exact}}[restated informally]
	There exists a randomized algorithm for dynamic \textsf{LIS} that has update time $\tilde O(n^{4/5})$ and maintains the value of \textsf{LIS} correctly with probability at least $1-n^{-10}$ at each step.
\end{customthm}

One technical difficulty that arises in the algorithm of Theorem~\ref{theorem:exact} is modifying the global information stored within heavy baskets. Although we prove that, with a desirable probability, the longest increasing subsequence of the whole array contains an element from each of the sampled layers, this does not hold for \textbf{all} increasing subsequences. Thus, it is likely that, for some element $a_i$ in a heavy basket, our algorithm maintains an incorrect value for the longest increasing subsequence ending at $a_i$. We discuss this in Section~\ref{sec:exact} and explain how to overcome the issue.

In \cref{sec:improved}, we further improve the update time of our dynamic algorithm down to $\tilde O(n^{2/3})$ using advanced methods based on efficient algorithms for handling Monge and unit-Monge matrices. The high-level structure of the algorithm is similar to what is explained before, but the light baskets are processed more efficiently.
\begin{customthm}{\ref{thm:exact2}}[restated informally]
	There exists a randomized algorithm for dynamic \textsf{LIS} that has update time $\tilde O(n^{2/3})$ and maintains the value of \textsf{LIS} correctly with probability at least $1-n^{-10}$ at each step.
\end{customthm}

\subsection{$(1-o(1))$-Approximation Algorithm with $O(n^{o(1)})$ Update Time}\label{sec:results-approx}
Our $(1-o(1))$-approximation algorithm for dynamic \textsf{LIS} is based on the notion of grid packing introduced by Mitzenmacher and Seddighin~\cite{our-stoc-paper}. However, as we discuss later in this section, a constant factor loss in the original grid packing technique is inevitable. We address this issue by introducing the extended variant of grid packing, which is the basis of our $(1-o(1))$-approximation algorithm. We explain this in Section~\ref{sec:results-previous}. 
The generalization is natural and inspired by previous work on longest increasing subsequence~\cite{DBLP:conf/stoc/ImMS17}. 
However, the more novel and technically challenging component of our algorithm is the application of extended grid packing to dynamic \textsf{LIS}. 
Since previous applications of grid packing are based on a bound that cannot be guaranteed for a $(1-o(1))$-approximate solution of extended grid packing, we design a completely different approach for applying extended grid packing to dynamic \textsf{LIS}. We elaborate more on this in Section~\ref{sec:results-approach}. In what follows, we denote the approximation factor of our algorithm by $1-\epsilon$, but we allow subconstant $\epsilon$.

\subsubsection{Background: Grid Packing}\label{sec:results-previous}
Grid packing is related to the notion of window-compatible solutions proposed by Boroujeni, Ehsani, Ghodsi, HajiAghayi, and Seddighin~\cite{boroujeni2018approximating} for approximating edit distance within a constant factor. The problem can be thought of as a game between us and an adversary. In this problem, we have a table of size $m \times m$. Our goal is to introduce a number of segments on the table. 
Each segment covers a consecutive set of cells either in a row or in a column.
A segment $A$ \textit{precedes} a segment $B$ if \textbf{every} cell of $A$ is strictly higher than every cell of $B$ and also \textbf{every} cell of $A$ is strictly to the right of every cell of $B$.
Two segments are \textit{non-conflicting} if one of them precedes the other one. 
Otherwise, we call them \textit{conflicting}.
The segments we introduce can overlap, and there is no restriction on the number of segments or the length of each segment.
However, we would like to minimize the maximum number of segments that cover each cell. 

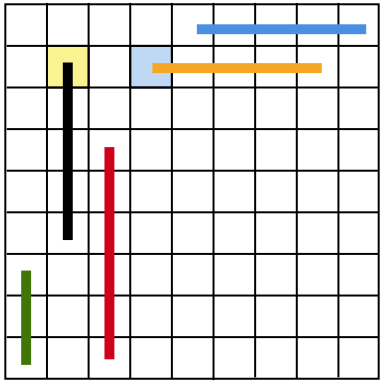
\begin{figure}[ht]

\centering

\tikzset{every picture/.style={line width=0.75pt}} 

\begin{tikzpicture}[x=0.75pt,y=0.75pt,yscale=-0.7,xscale=0.7]

\draw   (181,11) -- (450,11) -- (450,281) -- (181,281) -- cycle ;
\draw    (301,10) -- (301,281) ;

\draw    (331,10) -- (331,282) ;

\draw    (361,10) -- (361,280) ;

\draw    (211,10) -- (211,281) ;

\draw    (241,10) -- (241,281) ;

\draw    (271,10) -- (271,281) ;

\draw    (391,10) -- (391,281) ;

\draw    (421,10) -- (421,280) ;

\draw    (182,41) -- (450,41) ;

\draw    (182,71) -- (450,71) ;

\draw    (182,101) -- (450,101) ;

\draw    (182,131) -- (450,131) ;

\draw    (182,161) -- (450,161) ;

\draw    (182,191) -- (450,191) ;

\draw    (182,221) -- (450,221) ;

\draw    (182,251) -- (450,251) ;

\draw [color={rgb, 255:red, 208; green, 2; blue, 27 }  ,draw opacity=1 ][line width=3.75]    (256,114) -- (256,267) ;

\draw [color={rgb, 255:red, 74; green, 144; blue, 226 }  ,draw opacity=1 ][line width=3.75]    (319,29) -- (441,29) ;

\draw [color={rgb, 255:red, 65; green, 117; blue, 5 }  ,draw opacity=1 ][line width=3.75]    (196,203) -- (196,271) ;

\draw  [color={rgb, 255:red, 0; green, 0; blue, 0 }  ,draw opacity=1 ][fill={rgb, 255:red, 248; green, 231; blue, 28 }  ,fill opacity=0.48 ] (211,41) -- (241,41) -- (241,71) -- (211,71) -- cycle ;
\draw [line width=3.75]    (226,53) -- (226,181) ;

\draw  [color={rgb, 255:red, 0; green, 0; blue, 0 }  ,draw opacity=1 ][fill={rgb, 255:red, 74; green, 144; blue, 226 }  ,fill opacity=0.35 ] (271,41) -- (301,41) -- (301,71) -- (271,71) -- cycle ;
\draw [color={rgb, 255:red, 245; green, 166; blue, 35 }  ,draw opacity=1 ][line width=3.75]    (287,57) -- (409,57) ;

\end{tikzpicture}
\caption{Segments are shown on the grid. The pair (black, orange) is conflicting since the yellow cell (covered by the black segment) is on the same row as the blue cell (covered by the orange segment). The following pairs are non-conflicting: (green, black), (green, orange), (green, blue), (red, orange), (red, blue), (black, blue).} \label{fig:crossing}
\end{figure}

After we choose the segments, an adversary puts a non-negative number on each cell of the table. The score of a subset of cells of the table would be the sum of their values and the overall score of the table is the maximum score of a path of length $2m-1$ from the bottom-left corner to the top-right corner. In such a path, we always either move up or to the right.

The score of a segment is the sum of the numbers on the cells it covers. We obtain the maximum sum of the scores of a non-conflicting set of segments.  The score of the table is an upper bound on the score of any set of non-conflicting segments. We would like to choose segments so that the ratio of the score of the table and our score is bounded by a constant, no matter how the adversary puts the numbers on the table. More precisely, we call a solution $(\alpha,\beta)$-approximate, if at most $\alpha$ segments cover each cell and it guarantees a $1/\beta$ fraction of the score of the table for us for any assignment of numbers to the table cells.

Mitzenmacher and Seddighin~\cite{our-stoc-paper} prove the following theorem: For any $m \times m$ table and any $0 < \kappa < 1$, there exists a grid packing solution with  guarantee $(O_{\kappa}(m^\kappa \log m),O(1/\kappa))$. That is, each cell is covered by at most $O_{\kappa}(m^\kappa \log m)$ segments and the ratio of the table's score over our score is bounded by $O(1/\kappa)$ in the worst case. 

\begin{theorem}[from~\cite{our-stoc-paper}]\label{theorem:grid packing} For any $0 < \kappa < 1$, the grid packing
	problem on an $m \times m$ table admits an $(O_{\kappa}(m^\kappa \log m),O(1/\kappa))$-approximate solution.
\end{theorem}

\subsubsection{Extension: Grid Packing with Multisegments}
The general framework of grid packing remains the same for our extension:
The problem can be thought of as a game played on an $m \times m$ table against an adversary and the goal is to introduce some multisegments (a generalization of segments explained below) such that after the adversary puts their numbers on the table cells, the score we obtain is comparable to table's score.
However, extended grid packing differs from grid packing in two ways: 
First, we introduce a new notion that we call a \textit{multisegment} and we allow the use of multisegments instead of segments. Second, we do not enforce any bound on the number of multisegments that cover each cell. That is, we only have one objective which is maximizing the ratio of our score and the score of the table. Without the bound, algorithmically utilizing extended grid packing for \textsf{LIS} becomes more challenging as previous solutions require a cap on the maximum number of segments covering each cell.
Nevertheless, we show in Section~\ref{sec:example} how to apply extended grid packing in absence of this bound.

We bring an example in Section~\ref{sec:dynamic} to prove that by just using the segments in the grid packing problem, there is no way to obtain more than a $2/3$ fraction of the table's score, even if there is no bound on the number of segments that cover any cell. 
This example motivates our generalization, which we discuss in the following:
For a horizontal/vertical segment, we define its \textit{first} cell as its leftmost/bottommost cell and its \textit{last} cell as the rightmost/topmost cell of the segment. 
A $\Delta$-multisegment is defined as a combination of $\Delta$ segments $s_1, s_2, \ldots, s_{\Delta}$ where, for each $1 \leq i \leq \Delta-1$, the last cell of segment $s_i$ coincides with the first cell of segment $s_{i+1}$. 
(By definition, $1$-multisegments are the same as segments.) 
We say that a multisegment covers a cell if any of its segments covers that cell. 
Moreover, two multisegments $S_1$ and $S_2$ are non-conflicting if, for each segment $x$ of $S_1$ and each segment $y$ of $S_2$, the segments $x$ and $y$ are non-conflicting.
To avoid confusion, we use uppercase letters for multisegments and lowercase letters for segments. Based on this definition, for any $1 \leq i < \Delta$, an $i$-multisegment is also a $\Delta$-multisegment (we may add $\Delta-i$ single cell segments to an $i$-multisegment to make it compatible with the definition of $\Delta$-multisegment without any change in its shape).
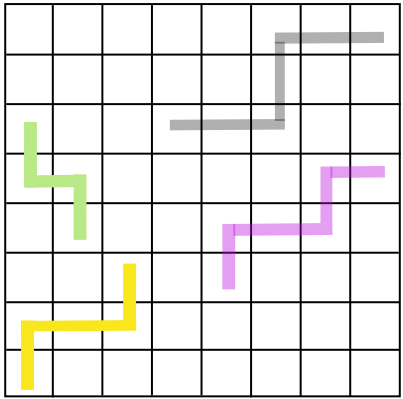
\begin{figure}[ht]

\centering

\tikzset{every picture/.style={line width=0.75pt}} 

\begin{tikzpicture}[x=0.75pt,y=0.75pt,yscale=-0.5,xscale=.5]

\draw   (141.5,15) -- (539.5,15) -- (539.5,411) -- (141.5,411) -- cycle ;
\draw    (189.5,16) -- (189.5,412) ;
\draw    (239.5,16) -- (239.5,410) ;
\draw    (289.5,16) -- (289.5,410) ;
\draw    (339.5,16) -- (339.5,410) ;
\draw    (389.5,16) -- (389.5,412) ;
\draw    (439.5,16) -- (439.5,411) ;
\draw    (489.5,16) -- (489.5,411) ;
\draw    (539.5,66) -- (142.5,66) ;
\draw    (540.5,116) -- (142.5,116) ;
\draw    (540.5,166) -- (142.5,166) ;
\draw    (539.5,216) -- (142.5,216) ;
\draw    (538.5,266) -- (142.5,266) ;
\draw    (539.5,316) -- (142.5,316) ;
\draw    (539.5,364) -- (142.5,364) ;
\draw  [draw opacity=0][fill={rgb, 255:red, 248; green, 231; blue, 28 }  ,fill opacity=1 ] (157.48,404.5) -- (157.52,334.5) -- (170.52,334.5) -- (170.48,404.5) -- cycle ;
\draw  [draw opacity=0][fill={rgb, 255:red, 248; green, 231; blue, 28 }  ,fill opacity=1 ] (157.54,334.81) -- (273.5,334) -- (273.57,344.69) -- (157.62,345.5) -- cycle ;
\draw  [draw opacity=0][fill={rgb, 255:red, 248; green, 231; blue, 28 }  ,fill opacity=1 ] (260.48,343) -- (260.52,276.99) -- (273.5,277) -- (273.45,343.01) -- cycle ;
\draw  [draw opacity=0][fill={rgb, 255:red, 189; green, 16; blue, 224 }  ,fill opacity=0.4 ] (360.48,303) -- (360.52,236.99) -- (373.5,237) -- (373.45,303.01) -- cycle ;
\draw  [draw opacity=0][fill={rgb, 255:red, 189; green, 16; blue, 224 }  ,fill opacity=0.4 ] (371.51,236.94) -- (471.5,236) -- (471.61,248.04) -- (371.62,248.98) -- cycle ;
\draw  [draw opacity=0][fill={rgb, 255:red, 189; green, 16; blue, 224 }  ,fill opacity=0.4 ] (459.48,236) -- (459.52,178.99) -- (471.5,179) -- (471.46,236.01) -- cycle ;
\draw  [draw opacity=0][fill={rgb, 255:red, 189; green, 16; blue, 224 }  ,fill opacity=0.4 ] (469.51,178.94) -- (524.39,178.43) -- (524.5,190) -- (469.62,190.52) -- cycle ;
\draw  [draw opacity=0][fill={rgb, 255:red, 0; green, 0; blue, 0 }  ,fill opacity=0.31 ] (307.54,131.81) -- (423.5,131) -- (423.57,141.69) -- (307.62,142.5) -- cycle ;
\draw  [draw opacity=0][fill={rgb, 255:red, 0; green, 0; blue, 0 }  ,fill opacity=0.31 ] (413.42,43.77) -- (523.49,43) -- (523.57,54.23) -- (413.49,55) -- cycle ;
\draw  [draw opacity=0][fill={rgb, 255:red, 0; green, 0; blue, 0 }  ,fill opacity=0.31 ] (423.5,53) -- (423.51,133.04) -- (413.51,133.04) -- (413.49,53) -- cycle ;
\draw  [draw opacity=0][fill={rgb, 255:red, 184; green, 233; blue, 134 }  ,fill opacity=1 ] (210.48,253) -- (210.52,186.99) -- (223.5,187) -- (223.45,253.01) -- cycle ;
\draw  [draw opacity=0][fill={rgb, 255:red, 184; green, 233; blue, 134 }  ,fill opacity=1 ] (219.1,199.68) -- (161.66,200.43) -- (161.5,188) -- (218.93,187.25) -- cycle ;
\draw  [draw opacity=0][fill={rgb, 255:red, 184; green, 233; blue, 134 }  ,fill opacity=1 ] (160.48,200) -- (160.52,133.99) -- (173.5,134) -- (173.45,200.01) -- cycle ;

\end{tikzpicture}

\caption{All polylines except for the green one are valid multisegments. Yellow and gray multisegments are non-conflicting, while the rest of the multisegment pairs are conflicting.} \label{fig:milti}
\end{figure}

We define an extended version of the grid packing problem as a game between us and an adversary. Similar to grid packing, we first introduce a number of multisegments and then the adversary puts nonnegative numbers on the cells of the table. Then, table's score is formulated as the largest sum the adversary can collect from the values of the cells by moving from the bottom-left corner to the top-right corner of the table. Our score is the largest sum we can collect by non-conflicting multisegments where the value of a multisegment is equal to the total sum of the numbers of the cells it covers.

As we show in Lemma~\ref{lemma:multi}, if we consider all possible $\Delta$-multisegments in our solution, our score is always at least a $\frac{\Delta-1}{\Delta}$ fraction of the table's score. Notice that by introducing all such multisegments, a cell may be covered by $m^{\Theta(\Delta)}$  multisegments.

\begin{customlem}{\ref{lemma:multi}}[restated]
	Let $\Delta$ be a positive integer.
	If we introduce all $\Delta$-multisegments in the extended grid packing problem, our score will be least a $\frac{\Delta-1}{\Delta}$ fraction of the table's score regardless of the values of the table cells.
\end{customlem}

Lemma~\ref{lemma:multi} alone does not suffice to improve the dynamic \textsf{LIS} algorithm of Mitzenmacher and Seddighin~\cite{our-stoc-paper} since the algorithm relies on a bound on the number of segments that cover each cell. We remedy this issue by presenting a more clever algorithm that does not require this bound. 

\subsubsection{Application of Extended Grid Packing}\label{sec:results-approach}
We refer the reader to previous work~\cite{our-stoc-paper,our-soda-paper} for discussions on how to use grid packing for approximating \textsf{LIS}. Roughly speaking, they consider the point-based representation of the problem and construct an $m\times m$ grid whose rows and columns evenly divide the points.
Next, they use the grid packing technique and, for each of the selected segments, they construct a partial solution that maintains an approximation to the \textsf{LIS} of the points covered by that segment.
Thus, every time a change is made, their algorithm has to update the solution for all segments that cover the modified point.
Therefore, previous techniques require a bound on the number of segments that cover each cell of the grid to make sure the update time is sublinear. In order to obtain a score arbitrarily close to the score of the table, we need to include a lot of multisegments in our solution for extended grid packing, many of which cover the same cells of the table.
This renders the previous approach incompatible with the new construction. 
To address this issue, we introduce a new method that allows using all the $\Delta$-multisegments (for a specific value $\Delta$).

At a high-level, the advantage of our new algorithm over the previous technique is that we adaptively decide which multisegments to use in the construction of a global solution. Previous applications are non-adaptive in this sense: 
They consider all segments of the grid packing solution and, for each segment, they maintain a partial solution for the points covered by that segment.
Our solution for extended grid packing uses $m^{O(\Delta)}$ multisegments (for a $(1-\epsilon)$-approximate solution, we require $\Delta = \Omega(1/\epsilon)$)  which is too many to even loop over.
Thus, we need to determine which multisegments have the potential to contribute to our overall solution before combining the partial solutions. We make such decisions adaptively as we query the subproblems in order to verify if a multisegment can be used in our solution. Thus, as we modify the sequence, the multisegments that may contribute to the global solution are subject to change.

In order to apply extended grid packing, we generalize the dynamic \textsf{LIS} problem. 
Instead of asking the size of the longest increasing subsequence after each operation, we define a \textit{query} to our algorithm in the following way:
a rectangle whose sides are parallel to the axis lines is given to us, and our algorithm should output an estimation to the size of the longest increasing subsequence of the points in the rectangle.
This obviously generalizes the problem since if the rectangle contains all of the points, then the answer is the \textsf{LIS} of the entire sequence. 
This generalization has two benefits: (i) Instead of defining each subproblem as the points covered by each segment (as Mitzenmacher and Seddighin~\cite{our-stoc-paper} do in previous work), we can define each subproblem as the points covered by a row or a column.
When the \textsf{LIS} of the points covered by a segment is desired, we can simply query the corresponding part of the subproblem which is covered by the segment. 
(This does not hold for multisegments in general, but we show how to approximate the \textsf{LIS} of a multisegment in Section~\ref{sec:dynamic}.)
This way, every operation changes at most two subproblems (one row and one column), and thus there is no need to have a bound on the number of segments that cover a point. 
(ii) The second and more important benefit of this approach is that we can distinguish between the query time and update time. More precisely, in previous work we make no distinction between query time and update time since we only look for the \textsf{LIS} of the whole sequence. 
Therefore, after each operation, the only question that we ask is for the \textsf{LIS} of the entire sequence.
With our generalization, we may make multiple queries after an operation, and therefore answering a query may require a different runtime. 
One of the key points of our algorithm is that our query time is much smaller than our update time, and this allows us to recursively run multiple queries in each of the subproblems without incurring too much cost in the running time.

Generalized queries, of course, make the problem substantially more complicated. 
Even in the stating setting (when there are no operations to be performed on the sequence), answering this type of queries is not easy.
If we only seek to find the \textsf{LIS} of the entire sequence, patience sorting can solve the problem in nearly linear time. 
However, if we are allowed to preprocess the sequence and then have to handle queries for the \textsf{LIS} of rectangles, the problem becomes more challenging. 
The authors are not aware of any linear-time (or even quadratic-time) preprocessing algorithm supporting exact queries in polylogarithmic time. 
As another application of extended grid packing, we show in Section~\ref{sec:dynamic} how to answer the queries in polylogarithmic time with nearly linear-time preprocessing by losing a $1-\epsilon$ factor in the approximation.

A key ingredient of our algorithm is a discretization technique that significantly improves the running time.
Roughly speaking, the number of $\Delta$-multisegments grows as $m^{\Theta(\Delta)}$ for an $m\times m$ grid as we increase $\Delta$. Obviously, we cannot afford to consider all such multisegments in our solution. Thus, we need to adaptively decide which multisegments to use in a solution for \textsf{LIS}. To this end, we use a discretization technique that narrows down the space of search fromall $m^{\Theta(\Delta)}$ multisegments to an $m^{O(1)}(\log_{1+\epsilon} n)^{O(\Delta)}$-sized subset at the expense of losing a $1-\epsilon$ factor in the approximation. This method is technically involved, but enables us to estimate the solution of a query in polylogarithmic time. We explain the details of our algorithm in Section~\ref{sec:dynamic}.

\begin{customthm}{\ref{theorem:third}}[restated informally]
	There exists an algorithm for dynamic \textsf{LIS} that approximates the solution within a $1-o(1)$ multiplicative factor and updates the sequence in time $O(n^{o(1)})$.
\end{customthm}

	\newpage
	\section{Exact Algorithm for Dynamic \textsf{LIS}}\label{sec:exact}
We present an exact dynamic algorithm for \textsf{LIS} with sublinear update time. Our method is based on a heavy-light decomposition of regions of the array combined with combinatorial analysis of increasing subsequences.


We use the notion of block-based algorithms~\cite{our-stoc-paper} to simplify the explanation. This lets us include preprocessing steps which may violate the worst-case update times. Yet, it follows from~\cite{our-stoc-paper} that a block-based algorithm can be turned into a dynamic algorithm whose worst-case update time is equal to the block-based algorithm's amortized update time. More precisely, a block-based algorithm starts with a sequence $a$ of size $n$. It is allowed to preprocess the sequence in time $f(n)$. After the preprocessing step, the algorithm is required to execute $g(n)$ operations, each in worst-case time $h(n)$. After $g(n)$ operations, the block-based algorithm terminates. It follows from previous work~\cite{our-stoc-paper} that such an algorithm can be transformed into a dynamic algorithm for \textsf{LIS} with worst-case update time $O(f(n)/g(n) + h(n))$.

\input{figs/baskets}

We construct a block-based algorithm in the following way: In the preprocessing step, we spend time $O(n \log n)$ and compute the size of the longest increasing subsequence that ends at any element $a_j$. Let $b_j$ denote this value for element $a_j$. Let $L_i = \{a_j \mid b_j = i\}$ be the set of elements whose corresponding solution has size $i$. We refer to each $L_i$ as a \emph{layer}. In the preprocessing step, we construct \emph{baskets} each of which consists of the elements of several consecutive layers. Keep in mind that the elements of a basket are not necessarily consecutive in terms of their position in the sequence.
However, we maintain as an invariant that every increasing subsequence visits the baskets in the increasing order.

In order to construct the baskets, we define a parameter $w$ that we set later. All baskets, except for the first and the last baskets contain $w$ consecutive layers. We define the \textit{boundary} of a basket as the last layer which is included in that basket. We make the baskets in a way that the total size of the boundaries is bounded by $O(n/w)$. Let us be more precise about this. Since each element of the sequence is included in exactly one of the layers, then we have $\sum |L_i| = n$. Therefore, if we choose an integer value $r$ in range $[0,w-1]$ uniformly at random, then \[\mathbb{E} \left[\sum_{i \bmod w = r} |L_i|\right] = \frac{n}{w}.\] 
This implies that $\sum_{i \bmod w = r} |L_i| \leq 2n/w$ holds for at least $\lceil w/2 \rceil$ choices of $r$. Let $R$ be the set of all such choices for $r$. In our algorithm, we choose a value $r$ from $R$ uniformly at random and set $L_r, L_{w+r},\ldots$ to be the boundaries. If $r \neq 0$, the first basket contains all layers $L_1, L_2, \ldots, L_r$, otherwise it contains the first $w$ layers. The second basket contains the next $w$ layers and the same holds for the rest of the baskets. If the last basket contains fewer than $w$ layers, we create an additional dummy layer that only contains an element with position $n+1$ and value $\infty$ and put this in the last basket as the boundary. Otherwise, we create a new basket that has a single element with position $n+1$ and value infinity. This element is also the boundary of the last basket. We denote by $d_i$ the size of basket $i$ (the number of elements included in it) and by $e_i$ the size of the boundary of basket $i$.

We emphasize that the layers, baskets, and their boundaries are essentially kept intact for the lifetime of the algorithm (which is at most $g(n)$ updates) even though the optimal increasing subsequences may change.
The only exception is that each element inserted to the sequence $a$ is inserted to one of the existing baskets, and each element removed from the sequence $a$ is removed from its basket.
The newly inserted elements are never included in the boundary, and the algorithm declares a failure following 
any attempt to remove a boundary element (which means that the algorithm ignores subsequent updates and keeps reporting $0$ as the $\textsf{LIS}$ length).
Since the updates that the algorithm encounters are independent of the random choice of $r$,
any single update leads to a failure with probability at most $1/|R|<2/w$,
which makes the overall failure probability limited to $2g(n)/w$.

Instead of maintaining the $\textsf{LIS}$, our algorithm actually maximizes the length of an increasing subsequence that includes one element from each boundary (except for the last one consisting of $a_{n+1}=\infty$).
Before we describe how this is achieved, let us argue why this is meaningful.

\begin{lemma}\label{lem:wrong}
	Let $U$ be an arbitrary longest increasing subsequence after at most $g(n)$ updates.
	Then, with probability at least $1-4g(n)/w$, subsequence $U$ contains an element form each boundary layer
	(except for the last one consisting of $a_{n+1}=\infty$).
\end{lemma}
\begin{proof}
	Note that $|U| \ge \opt - g(n)$, where $\opt$ denotes the \textsf{LIS} length 
	during the initialization. 
	Moreover, at most $g(n)$ elements of $U$ may have been added after
	the initialization, which means that at least $|U|-2g(n)$ elements were present in the original sequence.
	Any increasing subsequence contains at most one element from each layer,
	thus $U$ may miss up to $2g(n)$ out of the $\opt$ initial layers constructed during initialization.
	Any layer $L_i$ is a boundary layer if $i \bmod w = r$, which happens with probability $1/|R|\le 2/w$ and thus the probability that $U$ misses a boundary  layer is at most $4g(n)/w$.
\end{proof}
Since the failure probability is $2g(n)/w$, Lemma~\ref{lem:wrong} implies
that each answer reported by the algorithm is correct with probability at least $1-6g(n)/w$.
We set $g(n)=w/12$ to make sure that this is at least $1/2$. 
(To further boost the success probability, our final algorithm will maintain $O(\log n)$ instances
of the algorithm presented here; see the proof of Theorem~\ref{theorem:exact} for details.)

Let us proceed with the details of our algorithm.
It follows from our construction that each basket contains at most $w$ consecutive layers (the equality holds with the exception of the first and the last baskets). Moreover, the total size of the boundaries is bounded by $2n/w+1$ (the additional $+1$ term is due to the dummy layer). Let $w \le s \le n$ be a parameter that we set later. 
We use parameter $s$ to define a bound on the size of the baskets that we call \textit{light}. We categorize the baskets based on the number of their elements and their boundary size:
\begin{itemize}
	\item We call a basket \textit{light}, if its size is bounded by $s$, the size of its boundary (the number of elements in its last layer) is bounded by $s/w$, and the size of the boundary of its previous basket (if any) is bounded by $s/w$.
	\item If a basket is not light, then we call it \textit{heavy}. That is, a heavy basket either has a size more than $s$ or a boundary size more than $s/w$ or its previous basket has a boundary of size more than $s/w$.
\end{itemize} 

Since the total number of elements is bounded by $n$ and the total size of all boundaries is bounded by $2n/w+1$, then the number of heavy baskets is bounded by $5n/s+1$.

Throughout our algorithm, we maintain a local data structure for each light basket $i$ that stores the following information: If $i=1$, then for each element $a_y$ of basket $1$, we store the size of the longest increasing subsequence that ends at $a_y$.
If $i>1$, then for each boundary element $a_x$ of basket $i-1$ and each element $a_y$ of basket $i$, we store the size of the longest increasing subsequence that starts at $a_x$ and ends at $a_y$. 
 Except for the initial element $a_x$, any element that may contribute to such a subsequence is certainly inside basket $i$. 
 
 In the preprocessing step, we initialize all these data structures. We can initialize the local data structure for basket $i>1$ in time $O(e_{i-1} d_i \log n)$ by running patience sorting for each boundary element of basket $i-1$ separately. 
 The size of each light basket is bounded by $s$ and the total boundary size (across all baskets) is at most $2n/w+1$, therefore the total time for initializing all light baskets $i>1$ is $O(ns \log n/w)$.
 Initializing basket $1$ (if it is light) costs $O(d_1 \log n) = O(s \log n)$ time, and this does not change the overall preprocessing time asymptotically.

We also keep some information for each heavy basket, but that information is not local. In other words, it depends on the elements of the previous baskets as well. In contrast, the data structure that we keep for each light basket is completely local. Let $a_i$ be an element in basket $j$. 
We call an increasing subsequence ending at element $a_i$ \textit{basket-compatible} if it includes one element from the boundary of each basket $1, 2, \ldots, j-1$. For every element $a_i$, we denote by $b'_i$ the size of the longest basket-compatible increasing subsequence that ends at $a_i$. Initially, we have $b_i = b'_i$ for all elements but as we make modifications to the sequence, the values $b'_i$ may diverge from $b_i$. 

For each heavy basket, we partition its elements into disjoint sets. Each set contains elements whose $b'_i$'s are equal.
Initially, each of these sets is one of the layers included in the basket. 
As we modify the sequence, the values $b'_i$ change and thus these sets may no longer coincide with the layers. 
Nevertheless, there is no extra preprocessing cost for heavy baskets since these sets are initially equal to the layers.

When an operation arrives, we first locate the basket to which it relates. We only consider element addition and element removal as element substitution can be simulated by the first two operations. Element removal corresponds to the basket that contains the element. Element addition corresponds to the basket with lowest index such that none of its boundary elements is both to the left of the added element and has a smaller value. In other words, when a new element is added, we find the basket with lowest index such that none of its boundary elements can be used to update the solution for the added element. 
(This way, we maintain the invariant that every increasing subsequence visits the baskets in the increasing order.)
If the corresponding basket is light, then we update the local data structure in time $O(s (s/w) \log n) = O(s^2\log n/w)$ (recall that the size of each light basket is bounded by $s$ and the previous basket has its boundary size bounded by $s/w$). If the operation corresponds to a heavy basket, we do not make any local changes.

 After local changes, we compute a global solution for all baskets in the following way: starting from basket 1, for each boundary element $a_i$, we compute $b'_i$. After we do this for basket 1, we move on to basket 2 and proceed to the last basket. The computed value for the dummy element minus 1 is the value that we report to the output. 

Recall that $e_i$ denotes the size of the boundary of basket $i$. For a light basket, we can use the local data structure and update the solution for its boundary elements in time $O(e_{i-1}e_i)$. Since basket $i$ is light, then we have $e_i \leq s/w$ and since the total size of the boundaries is bounded by $2n/w+1$ (i.e., $\sum e_j \leq 2n/w+1$), then the total runtime for light baskets is bounded by $O(ns/w^2)$. For each heavy basket, we use the algorithm of Chen, Chu, and Pinsker~\cite{chen2013dynamic} to update the solution in time $O(w \log n)$ (more details is given in the proof of Theorem~\ref{theorem:exact}). Thus, the total update time for heavy baskets is $O(nw \log n/s)$ since there are at most $5n/s+1$ heavy baskets. By setting $w = n^{0.4}$ and $s = n^{0.6}$ we obtain a block-based algorithm with $f(n) = O(ns \log n/w) = O(n^{1.2} \log n)$, $g(n) = w/12 = n^{0.4}/12$, and $h(n) = O( s^2\log n/w + ns/w^2 + nw \log n/s) =  O(n^{0.8} \log n)$. This leads to a dynamic algorithm for \textsf{LIS} with update time $O(n^{0.8} \log n)$ that after each update reports the solution correctly with probability at least $1/2$. By adding an additional $\log n$ multiplicative factor to the update time, we can improve the accuracy of the algorithm to $1-n^{-10}$. There is one more $O(\log n)$ factor in the update time due to the data structure that we use for accessing the elements of the sequence.

\begin{theorem}\label{theorem:exact}
	There exists a randomized algorithm for the dynamic \textsf{LIS} problem that has update time $O(n^{0.8} \log^3 n)$ and maintains the value of \textsf{LIS} correctly with probability at least $1-n^{-10}$.
\end{theorem} 
\begin{proof}
	As discussed earlier we design a block-based algorithm with preprocessing time $f(n) = O(n^{1.2} \log n)$ which is responsible for updating the solution for up to $g(n) = n^{0.4}/12$ operations and updates the solution in worst-case time $h(n)= O(n^{0.8} \log n)$. To this end, we set $s = n^{0.6}$, $w = n^{0.4}$, and divide the elements into different baskets. By the discussion following Lemma~\ref{lem:wrong}, after each update, we report the correct value of \textsf{LIS} with probability at least $1/2$.
	
	To bring the failure probability down to $n^{-10}$, we repeat the same procedure $20 \log n$ times. That is, we choose $20 \log n$ different values $r$ from $R$ and each time we make the baskets according to different boundaries. Every time we output the maximum solution that we obtain from all the $20 \log n$ algorithms. Also, since we use a balanced tree to access the elements of the sequence, another $O(\log n)$ factor is also involved in the runtime which makes the overall update time $O(n^{0.8} \log^3 n)$.
	
	Another thing to note here is the algorithm we use for the heavy baskets. It has been shown by Chen, Chu, and Pinsker~\cite{chen2013dynamic} that we can solve dynamic \textsf{LIS} with update time $O(\opt \log n)$, where $\opt$ is the size of the solution. In their algorithm, they divide the elements into subsets $L'_1, L'_2, \ldots, L'_{\opt}$ such that for all the elements in subset $L'_i$ the longest increasing subsequence ending at them has size~$i$. We use $L'$ to denote these subsets since we use $L_i$ for the layers of our algorithm. The stark difference between $L'_i$ and $L_i$ is that in our algorithm, $L_i$'s remain intact as operations arrive but in their algorithm each $L'_i$ gets updated after changing the sequence.
	
	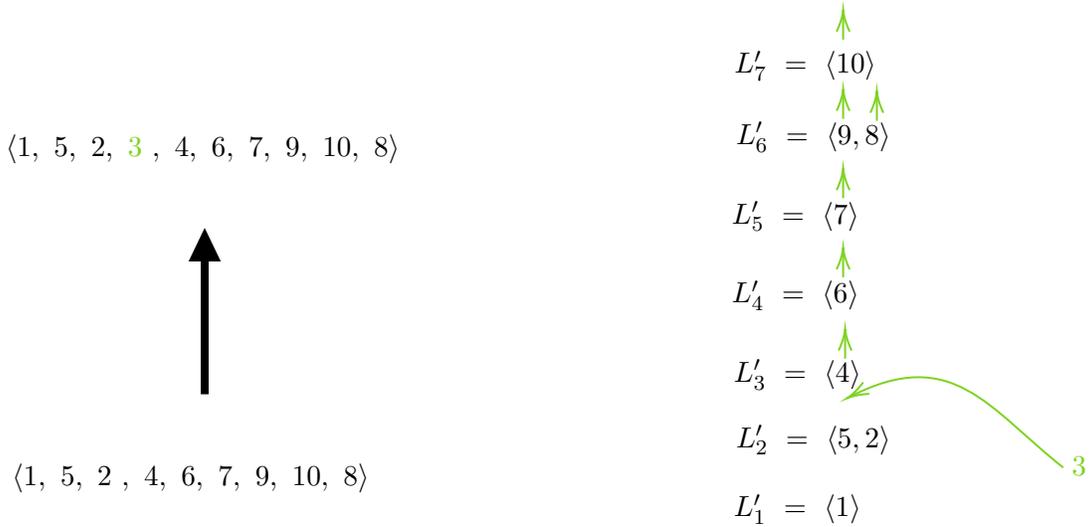
\begin{figure}[ht]

\centering

\tikzset{every picture/.style={line width=0.75pt}} 

\begin{tikzpicture}[x=0.75pt,y=0.75pt,yscale=-1,xscale=1]

\draw [line width=3]    (163,205) -- (163,126) ;
\draw [shift={(163,121)}, rotate = 450] [fill={rgb, 255:red, 0; green, 0; blue, 0 }  ][line width=3]  [draw opacity=0] (16.97,-8.15) -- (0,0) -- (16.97,8.15) -- cycle    ;

\draw [color={rgb, 255:red, 126; green, 211; blue, 33 }  ,draw opacity=1 ]   (596,242) .. controls (558.38,212.3) and (540.36,179.66) .. (488.58,206.18) ;
\draw [shift={(487,207)}, rotate = 332.15] [color={rgb, 255:red, 126; green, 211; blue, 33 }  ,draw opacity=1 ][line width=0.75]    (10.93,-3.29) .. controls (6.95,-1.4) and (3.31,-0.3) .. (0,0) .. controls (3.31,0.3) and (6.95,1.4) .. (10.93,3.29)   ;

\draw [color={rgb, 255:red, 126; green, 211; blue, 33 }  ,draw opacity=1 ][fill={rgb, 255:red, 126; green, 211; blue, 33 }  ,fill opacity=1 ]   (486,187) -- (486,174) ;
\draw [shift={(486,172)}, rotate = 450] [color={rgb, 255:red, 126; green, 211; blue, 33 }  ,draw opacity=1 ][line width=0.75]    (10.93,-3.29) .. controls (6.95,-1.4) and (3.31,-0.3) .. (0,0) .. controls (3.31,0.3) and (6.95,1.4) .. (10.93,3.29)   ;

\draw [color={rgb, 255:red, 126; green, 211; blue, 33 }  ,draw opacity=1 ][fill={rgb, 255:red, 126; green, 211; blue, 33 }  ,fill opacity=1 ]   (485,146) -- (485,133) ;
\draw [shift={(485,131)}, rotate = 450] [color={rgb, 255:red, 126; green, 211; blue, 33 }  ,draw opacity=1 ][line width=0.75]    (10.93,-3.29) .. controls (6.95,-1.4) and (3.31,-0.3) .. (0,0) .. controls (3.31,0.3) and (6.95,1.4) .. (10.93,3.29)   ;

\draw [color={rgb, 255:red, 126; green, 211; blue, 33 }  ,draw opacity=1 ][fill={rgb, 255:red, 126; green, 211; blue, 33 }  ,fill opacity=1 ]   (485,106) -- (485,93) ;
\draw [shift={(485,91)}, rotate = 450] [color={rgb, 255:red, 126; green, 211; blue, 33 }  ,draw opacity=1 ][line width=0.75]    (10.93,-3.29) .. controls (6.95,-1.4) and (3.31,-0.3) .. (0,0) .. controls (3.31,0.3) and (6.95,1.4) .. (10.93,3.29)   ;

\draw [color={rgb, 255:red, 126; green, 211; blue, 33 }  ,draw opacity=1 ][fill={rgb, 255:red, 126; green, 211; blue, 33 }  ,fill opacity=1 ]   (485,66) -- (485,53) ;
\draw [shift={(485,51)}, rotate = 450] [color={rgb, 255:red, 126; green, 211; blue, 33 }  ,draw opacity=1 ][line width=0.75]    (10.93,-3.29) .. controls (6.95,-1.4) and (3.31,-0.3) .. (0,0) .. controls (3.31,0.3) and (6.95,1.4) .. (10.93,3.29)   ;

\draw [color={rgb, 255:red, 126; green, 211; blue, 33 }  ,draw opacity=1 ][fill={rgb, 255:red, 126; green, 211; blue, 33 }  ,fill opacity=1 ]   (502,67) -- (502,54) ;
\draw [shift={(502,52)}, rotate = 450] [color={rgb, 255:red, 126; green, 211; blue, 33 }  ,draw opacity=1 ][line width=0.75]    (10.93,-3.29) .. controls (6.95,-1.4) and (3.31,-0.3) .. (0,0) .. controls (3.31,0.3) and (6.95,1.4) .. (10.93,3.29)   ;

\draw [color={rgb, 255:red, 126; green, 211; blue, 33 }  ,draw opacity=1 ][fill={rgb, 255:red, 126; green, 211; blue, 33 }  ,fill opacity=1 ]   (485,26) -- (485,13) ;
\draw [shift={(485,11)}, rotate = 450] [color={rgb, 255:red, 126; green, 211; blue, 33 }  ,draw opacity=1 ][line width=0.75]    (10.93,-3.29) .. controls (6.95,-1.4) and (3.31,-0.3) .. (0,0) .. controls (3.31,0.3) and (6.95,1.4) .. (10.93,3.29)   ;

\draw (462,263) node   {$L'_{1} \ =\ \langle 1\rangle$};
\draw (470,228) node   {$L'_{2} \ =\ \langle 5,2\rangle$};
\draw (462,195) node   {$L'_{3} \ =\ \langle 4\rangle$};
\draw (162,81) node   {$\langle 1,\ 5,\ 2,\ \color{rgb, 255:red, 126; green, 211; blue, 33 }3\color{black}\ ,\ 4,\ 6,\ 7,\ 9,\ 10,\ 8\rangle$};
\draw (156,247) node   {$\langle 1,\ 5,\ 2\ ,\ 4,\ 6,\ 7,\ 9,\ 10,\ 8\rangle$};
\draw (461,155) node   {$L'_{4} \ =\ \langle 6\rangle$};
\draw (461,115) node   {$L'_{5} \ =\ \langle 7\rangle$};
\draw (470,75) node   {$L'_{6} \ =\ \langle 9,8\rangle$};
\draw (466,39) node   {$L'_{7} \ =\ \langle 10\rangle$};
\draw (604,241) node [color={rgb, 255:red, 126; green, 211; blue, 33 }  ,opacity=1 ]  {$3$};

\end{tikzpicture}

\caption{This example illustrates how an element addition is handled in the algorithm of Chen, Chu, and Pinsker~\cite{chen2013dynamic}. Inserting element $3$ to the array changes the levels of the elements. Upward arrows show that the level of the corresponding element increases after we add $3$ to the array.}\label{fig:chan}
\end{figure}
	
	For completeness, we explain the algorithm of Chen, Chu, and Pinsker~\cite{chen2013dynamic} in Appendix~\ref{sec:chen-ap}. Their algorithm is based on the following observation: when an operation is performed to the sequence, each element may only move between consecutive subsets and, moreover, the elements that move form an interval. In other words, if three elements $a_i > a_j > a_k$ belong to a subset, it is impossible for elements $a_i$ and $a_k$ to move to another subset while $a_j$ remains in the same subset after the update. They show that based on these two properties, after each operation, we can update the solution in time  $O(\opt \log(n/\opt))$. 
	
	We use the same idea for heavy baskets. More precisely, for each heavy basket, we divide the elements into subsets $L'_{\alpha}, L'_{\alpha+1},\ldots$ where for an $a_j$ in element in a subset $L'_i$ the size of the longest basket-compatible increasing subsequence ending at $a_j$ is equal to $i$. We emphasize that these subsets are different from $L_1, L_2, \ldots$ since they do not change as we perform operations to the sequence. It is obvious that after each operation, elements can only move between consecutive subsets of $L'$ in a heavy basket. Moreover, if in a subset $L'_i$ of a heavy basket some elements move, the moving elements form an interval. To see this, we show a reduction from the longest basket-compatible increasing subsequence to the longest increasing subsequence. We know that this property holds for the longest increasing subsequence. Now, from our sequence, we make another sequence $a'$, where $a'$ is the same as $a$ except that we copy each boundary element $n$ times and put the copies next to each other. To make the values distinct, we add $i \epsilon$ to the $i$'th copy of each boundary element. It follows that for each element $a'_i$ of the new sequence, if we find the longest increasing subsequence ending at $a'_i$ and remove the copied elements from the solution, we obtain the longest basket-compatible increasing subsequence that ends at its corresponding element in sequence $a$. This means that since performing a single operation in $a'$ preserves the interval property of moving elements, the same also holds for $a$ when we are concerned with basket-compatible increasing subsequences.
	
	Initially, the values $b'_i$ across elements in any single basket span $w$ consecutive integers.
	Each operation changes the value $b'_i$ by at most one, therefore at any time the values $b'_i$ 
	across elements in any single basket span at most $w + 2g(n)\le \frac{14}{12}w$ consecutive integers.
	Hence, the cost of updating any heavy basket is $O(w \log n)$, which yields $O(nw \log n / s)$ across all the heavy baskets.

	Thus, as discussed earlier, we obtain a block-based algorithm with $f(n) = O(ns \log n/w) = O(n^{1.2} \log n)$, $g(n) = w/12 = n^{0.4}/12$, $h(n) = O( s^2\log n/w + ns/w^2 + nw \log n/s) =  O(n^{0.8} \log n)$. This leads to a dynamic algorithm for \textsf{LIS} with update time $O(n^{0.8} \log n)$ that after each update reports the solution correctly with probability at least $1/2$. Two multiplicative $O(\log n)$ factors are added to the update time due to the data structure we use to access the elements of the sequence and $20 \log n$ different algorithms that we run in parallel to improve the failure probability down to $n^{-10}$.
\end{proof} 

\subsection{An Improved Dynamic Algorithm with Update Time $\tilde O(n^{2/3})$}\label{sec:improved}
We further improve the update time of our dynamic algorithm by exploiting structural insights expressed in terms of \emph{Monge} matrices, i.e., matrices $M\in \R^{\ell\times m}$ such that $M[x,y]+M[x+1,y+1]\le M[x+1,y]+M[x,y+1]$ holds for $x\in [1\dd \ell)$ and $y\in [1\dd m)$.
Monge matrices arise in our algorithm due to the following observation:
For every basket $i>1$, the maximum lengths of increasing subsequences starting
at the boundary of basket $i-1$ and ending at the boundary of basket $i$
can be embedded in an \emph{anti-Monge} matrix
(obtained by negating the entries of a Monge matrix).
Recall that the algorithm of \cref{theorem:exact} maintains these lengths for all light baskets $i>1$.
We develop the following two components to handle such baskets more efficiently:
\begin{itemize}
	\item In $\Ohtilde(e_{i-1} + d_i)$ time, we can construct an oracle that, given elements $a_x$ and $a_y$ at the boundaries of baskets $i-1$ and $i$, respectively, computes in $\Ohtilde(1)$ time the size of the longest increasing subsequence starting at $a_x$ and ending at $a_y$. This is proven in \cref{cor:1},	which also states that these values can be embedded in an anti-Monge matrix $M_i$; the latter requires carefully setting the values corresponding to \emph{invalid queries} for which $x> y$ or $a_x > a_y$.
	\item Given the sizes $b'_x$ of the longest basket-compatible increasing subsequences ending at the boundary elements $a_x$ of basket $i-1$ and assuming random access to the entries~$M_i$, we can in $\tilde O(e_{i-1}+e_i)$ time compute the sizes $b'_y$ of the longest basket-compatible increasing subsequences ending at the boundary elements $a_y$ of basket $i$. This is proven in \cref{lem:smawk}, where we rely on an efficient algorithm for computing the $(\min,+)$-product of a Monge matrix with a vector~\cite{AggarwalKMSW87} (equivalent to the $(\max,+)$-product for an anti-Monge matrix). 
\end{itemize}

For every basket $i$, let $(a_{f_{i,j}})_{j=1}^{e_i}$ be the subsequence of $(a_i)_{i=1}^n$
consisting of the boundary elements of basket $i$. 
Recall that each boundary was initialized as a single layer and that boundary elements are never modified.
Consequently, $(a_{f_{i,j}})_{j=1}^{e_i}$ forms a decreasing subsequence.

\begin{lemma}\label{cor:1}
	For each basket $i>1$, there exists an anti-Monge matrix $M_i\in \R^{e_{i-1}\times e_i}$ such that,
	for every $j\in [1\dd e_{i-1}]$ and $k\in [1\dd e_i]$, we have:
	\begin{align*}
		M_i[j,k] &= \text{length of the \textsf{LIS} from $a_{f_{i-1,j}}$ to $a_{f_{i,k}}$} & \text{if }f_{i-1,j}<f_{i,k}\text{ and }a_{f_{i-1,j}} < a_{f_{i,k}},\\
		M_i[j,k] & \le -n & \text{otherwise.}
	\end{align*}
	Moreover, any entry of $M_i$ can be computed in $\Oh(\log n/\log \log n)$ time
	after $O((e_{i-1}+d_i)\log^2 n)$-time preprocessing of basket $i$ and the boundary elements of basket $i-1$.
\end{lemma}

We derive \cref{cor:1} as an immediate corollary of the following theorem,
stated in terms of two arbitrary non-increasing subsequences
rather than just the subsequent boundary layers arising in our algorithm.
\cref{thm:oracle} is a generalization of a result of Tiskin~\cite{Tiskin13}, who considered the problem of computing so-called \emph{semi-local} \textsf{LIS} values.
The key algorithmic ingredient of Tiskin's procedure is an efficient algorithm for computing the $(\min,+)$-product of two \emph{simple unit-Monge matrices}~\cite{Tiskin15} (see \cref{sec:monge} for further discussion and the proof of~\cref{thm:oracle}).
\begin{restatable}{theorem}{thmoracle}\label{thm:oracle}
	Let $(a_i)_{i=0}^{n-1}$ be a real-valued sequence of length $n$, let $(a_{p_i})_{i=0}^{k-1}$ and $(a_{q_j})_{j=0}^{\ell-1}$ be non-increasing subsequences of $a$,
	and let $N$ be a positive integer.
	There exists an anti-Monge matrix $M\in \R^{k\times \ell}$ such that,
	for every $i\in [0\dd k)$ and $j\in [0\dd \ell)$, we have 
	\begin{align*}
		M[i,j] &= \text{length of the longest increasing subsequence from $a_{p_i}$ to $a_{q_j}$} & \text{if }p_i<q_j\text{ and }a_{p_i} < a_{q_j},\\
		M[i,j] & \le -N & \text{otherwise.}
	\end{align*}
	Moreover, any entry of $M$ can be computed in $\Oh(\log n/ \log \log n)$ time
	after $\Oh(n\log^2n)$-time preprocessing.
\end{restatable}

Recall that, for each $i\in [1\dd n]$, we denote by $b'_i$ the length of the longest basket-compatible increasing subsequence ending at element $a_i$.

\begin{lemma}\label{lem:smawk}
	Consider a basket $i>1$ augmented with oracle access to the matrix $M_i$ of \cref{cor:1}.
	Given the values $(b'_{f_{i-1,j}})_{j=1}^{e_{i-1}}$, the values $(b'_{f_{i,k}})_{k=1}^{e_{i}}$
	can be computed in $\Oh((e_{i-1}+e_i)\log n / \log \log n)$ time.
\end{lemma}
\begin{proof}
	We shall argue that $b'_{f_{i,k}} = \max_{j=1}^{e_{i-1}}(b'_{f_{i-1,j}}+M_i[j,k]-1)$ holds for each $k\in [1\dd e_i]$.
	For a proof of the `$\le$' inequality, consider a longest basket-compatible increasing subsequence ending at~$a_{f_{i,k}}$.
	Since $a_{f_{i,k}}$ belongs to basket $i$, this subsequence must include an element
	from the boundary of basket $i-1$, i.e., $a_{f_{i-1,j}}$ for some $j\in [1\dd e_{i-1}]$.
	Observe that every prefix of a basket-compatible subsequence is basket-compatible.
	Hence, the prefix ending at $a_{f_{i-1,j}}$ is of length at most $b'_{f_{i-1,j}}$.
	Furthermore, since $f_{i-1,j}\ne f_{i,k}$, the suffix from $a_{f_{i-1,j}}$ to $a_{f_{i,k}}$ is of length at most $M_i[j,k]$.
	The total length of the subsequence is therefore at most
	$b'_{f_{i-1,j}}+M_i[j,k]-1$ (note that the prefix and the suffix share $a_{f_{i-1,j}}$).
	Consequently, $b'_{f_{i,k}} \le \max_{j=1}^{e_{i-1}}(b'_{f_{i-1,j}}+M_i[j,k]-1)$.

	For a proof of the converse inequality, let us fix $j\in [1\dd e_{i-1}]$.
	Note that if $M_{i}[j,k]\le -n$, then \[b'_{f_{i-1,j}}+M_i[j,k]-1 \le f_{i-1,j}-n-1 < 0 < i \le b'_{f_{i,k}}.\]
	Hence, we may assume that $f_{i-1,j}<f_{i,k}$ and $a_{f_{i-1,j}} < a_{f_{i,k}}$.
	In particular, there exists an increasing subsequence from $a_{f_{i-1,j}}$ to $a_{f_{i,k}}$
	of length $M_i[j,k] \ge 2$.
	Combined with a basket-compatible increasing subsequence of length $b'_{f_{i-1,j}}$ ending at $a_{f_{i-1,j}}$,
	this yields a basket-compatible increasing subsequence of length $b'_{f_{i-1,j}}+M_i[j,k]-1$ ending at $a_{f_{i,k}}$.
	Thus, $ b'_{f_{i,k}} \ge b'_{f_{i-1,j}}+M_i[j,k]-1$.

	We conclude that the sought values $(b'_{f_{i,k}})_{k=1}^{e_{i}}$ 
	can be expressed as the result of the $(\max,+)$-product of matrix $M_i$ with vector $(b'_{f_{i-1,j}}-1)_{j=1}^{e_{i-1}}$.
	Since $M_i$ is an anti-Monge matrix, the SMAWK algorithm~\cite{AggarwalKMSW87} can be applied to compute such a product in $\Oh(e_{i-1}+e_{i})$ time using $\Oh(e_{i-1}+e_{i})$ queries accessing the elements $M_i$;
	see~\cite[Theorem 3.10]{Tiskin13}. The oracle of \cref{cor:1} provides $\Oh(\log n / \log \log n)$-time access to $M_i$,
	so the total running time is $\Oh((e_{i-1}+e_i)\log n / \log \log n)$.
\end{proof}

Our final algorithm for dynamic \textsf{LIS} applies \cref{cor:1,lem:smawk} for all light baskets.

\begin{theorem}\label{thm:exact2}
		There exists a randomized algorithm for the dynamic \textsf{LIS} problem that has update time $O(n^{2/3}\log^4 n)$ and maintains the value of \textsf{LIS} correctly with probability at least $1-n^{-10}$.
\end{theorem}
\begin{proof}
	Our algorithm is based on the same framework as that of \cref{theorem:exact}, and thus its correctness follows from \cref{theorem:exact}. 
	Here, we discuss the runtime.
	One thing to keep in mind is that the matrices $M_i$ that we construct for light baskets are accessed via oracle calls rather than being explicitly stored in memory.
	For each light basket $i>0$, we build the oracle of \cref{cor:1} when the basket is created and then every time it is modified.
	
	Similar to Theorem~\ref{theorem:exact}, we design a block-based algorithm that starts with an array of length~$n$, performs the operations for $g(n)$ steps, and spends (worst-case) time $h(n)$ for each operation. Assuming we have random access to the elements of the sequence, we manage to make our algorithm work for $f(n) =  O(n \log^3 n)$, $g(n) = n^{1/3}/12$, and $h(n) = O(n^{2/3} \log^3 n)$, which leads to a dynamic algorithm with update time $O(n^{2/3} \log^3 n)$.
	An additional $O(\log n)$ multiplicative factor appears in the runtime due to the data structure that we use for accessing the sequence elements, and thus the overall update time is bounded by $O(n^{2/3} \log^4 n)$.
	
	 Using Corollary~\ref{cor:1}, we achieve the preprocessing time of $O(n\log^3 n)$ since, for each basket $i$ (regardless of whether it is light or heavy), the preprocessing step can be implemented in time $O((e_{i-1}+d_i) \log^2 n)$ and therefore the total preprocessing time is $O(n \log^3 n)$ (recall that we run $O(\log n)$ parallel instances of our algorithm with different choices of boundary layers to keep the error rate small). 
	 Compared to the algorithm of \cref{theorem:exact}, the update time is improved in two ways: whenever an update modifies a light basket, we recompute the data structure of~\cref{cor:1} from scratch;
	 this costs  $ O((e_{i-1}+d_i) \log^2 n) =  O(s \log^2 n)$ time per local update.
	 Similar to before, after local updates, we iteratively update the size of \textsf{LIS} in an iterative manner. More precisely, starting from basket $1$, we determine the size of the longest (basket-compatible) increasing subsequence that ends at each boundary element of the basket.  
	 \cref{lem:smawk} implies that the cost of processing a light basket is bounded by $O((e_{i-1} + e_i)\log n)$. Since the total size of the boundary layers is bounded by $O(n/w)$, this amounts to a total cost of $O(n/w \log n)$ per instance.
	 Similar to before, the cost of processing a heavy basket is $O(w \log n)$ and, since we have at most $O(n/s)$ such baskets in every instance, this amounts to a total cost of $O(nw/s \log n)$ per instance. 
	 All the update times are multiplied by a factor $O(\log n)$ since we solve the problem for $O(\log n)$ simultaneous instances. Thus, the overall update time of our algorithm is bounded by $O(s \log^3 n+ n/w \log^2 n + nw/s \log^2 n)$. By setting $s = n^{2/3}$ and $w = n^{1/3}$, we obtain a block-based algorithm with preprocessing time $f(n) = O(n\log^3 n)$, $g(n) = w/12 = \Omega(n^{1/3})$, and $h(n) =  O(n^{2/3} \log^3 n)$, which in turn gives a dynamic algorithm with worst-case update time $O(n^{2/3} \log^3 n)$.
\end{proof}
	\subsection{\textsf{LIS} Oracle via Unit-Monge Matrices (Proof of \cref{thm:oracle})}\label{sec:monge}
Observe that \cref{thm:oracle} yields a data structure that, after $\Oh(n\log^2 n)$-time preprocessing of a sequence $(a_i)_{i=0}^{n-1}$ and its two non-increasing subsequences $(a_{p_i})_{i=0}^{k-1}$ and $(a_{q_j})_{j=0}^{\ell-1}$, supports the following queries in $\Oh(\log n / \log \log n)$ time:
given $i\in [0\dd k)$ and $j\in [0\dd \ell)$, compute the length of the longest increasing subsequence of $a$ starting at $a_{p_i}$ and ending at $a_{q_j}$.
Tiskin~\cite[Section 8.1]{Tiskin13} considered closely related semi-local \textsf{LIS} values
and, in particular, proved the following result:
\begin{theorem}[Tiskin~\cite{Tiskin13}]\label{thm:localLIS}
	A sequence $(a_i)_{i=0}^{n-1}$ can be preprocessed in $\Oh(n\log^2 n)$ to support the following queries in $\Oh(\log n / \log \log n)$ time:
	given indices $i,j\in [0\dd n)$, compute the length of the longest increasing subsequence of 
	the sequence $a_{i},a_{i+1},\ldots,a_j$.
\end{theorem}
Note that the increasing subsequences in \cref{thm:localLIS} do not have to contain $a_i$ or $a_j$, whereas the increasing subsequences in \cref{thm:oracle} must start at $a_{p_i}$ and end at $a_{q_j}$.
To see that the latter setting is more general, consider a sequence $(b_i)_{i=0}^{3n-1}$
whose entries are defined as follows for $i\in [0\dd 3n)$:
\[b_i = \begin{cases}
	-\infty & \text{if }i\bmod 3 = 0,\\
	a_{(i-1)/3} & \text{if }i\bmod 3 = 1,\\
	+\infty & \text{if }i\bmod 3 = 2.
\end{cases}\]
Observe that any increasing subsequence of $a_{i},a_{i+1},\ldots,a_j$
yields an increasing subsequence of $b_{3i+1},b_{3i+2},\ldots, b_{3j+1}$,
which can be extended with $b_{3i}=-\infty$ and $b_{3j+2}=+\infty$.
Conversely, we obtain an increasing subsequence of $a_{i},a_{i+1},\ldots, a_{j}$ by trimming the endpoints of any increasing subsequence of $b$ starting at $b_{3i}=-\infty$ and ending at $b_{3j+2}=+\infty$.
Thus, \cref{thm:localLIS} can be derived from \cref{thm:oracle} applied to the sequence $(b_i)_{i=0}^{3n-1}$ and its two subsequences $(b_{3i})_{i=0}^{n-1}$ and $(b_{3j+2})_{j=0}^{n-1}$.

In order to prove \cref{thm:oracle}, we carefully adapt the techniques of~\cite{Tiskin13}.
Unfortunately, this requires generalizing the definitions of most combinatorial objects and, consequently, repeating most of the proofs. The only component that we managed to use in a black-box fashion is an efficient procedure for computing the min-plus product of two simple unit-Monge matrices~\cite{Tiskin13,Tiskin15}.

\subsubsection{Preliminaries}
We start by introducing the terminology behind the definition of unit-Monge matrices and their space-efficient representation. We mostly follow~\cite[Chapters 2 and 3]{Tiskin13}, except that we always use integers to index matrix rows and columns.

The $(\min,+)$ product of matrices $A\in \R^{n\times p}$ and $B\in \R^{p\times m}$
is a matrix $A\odot B\in \R^{n\times m}$ with entries defined as follows for $i\in [0\dd n)$ and $j\in [0\dd m)$:
\[(A\odot B)[i,j] = \min_{k\in [0\dd p)}\left(A[i,k]+B[k,j]\right).\]

For a matrix $A\in \R^{n\times m}$, the \emph{distribution matrix} $A^\Sigma \in \R^{(n+1)\times (m+1)}$ has its entries defined as follows for $i\in [0\dd n]$ and $j\in [0\dd m]$:
\[A^\Sigma [i,j] = \sum_{i'\in [i\dd n),\; j'\in [0\dd j)} A[i',j'].\]
For a matrix $A\in \R^{(n+1)\times (m+1)}$, the \emph{density matrix} $A^\square \in \R^{n\times m}$ has its entries defined as follows for $i\in [0\dd n)$ and $j\in [0\dd m)$:
\[A^\square [i,j] = A[i+1,j]+ A[i,j+1] - A[i,j]- A[i+1,j+1].\]
The \emph{seaweed product} of matrices $A\in \R^{n\times p}$ and $B\in \R^{p\times m}$
is a matrix $A \boxdot B = (A^\Sigma \odot B^\Sigma)^\square \in \R^{n\times m}$.

Note that every matrix $A\in \R^{n\times m}$ satisfies $A = (A^\Sigma)^\square$.
A matrix $A\in \R^{(n+1)\times (m+1)}$ is called \emph{simple} if $A = (A^\square)^\Sigma$ or, equivalently, 
if $A[n,j]=0$ for $j\in [0\dd m]$ and $A[i,0]=0$ for $i\in [0\dd n]$.

A matrix $P\in \{0,1\}^{n\times n}$ is a \emph{permutation matrix} if each row and each column contains exactly one entry equal to $1$.
Note that an $n\times n$ permutation matrix can be represented with a permutation $\sigma:[0\dd n)\to [0\dd n)$ 
such that $P[i,j]=1$ if and only if $j = \sigma(i)$.
\begin{theorem}[{Tiskin~\cite{Tiskin13,Tiskin15}}]\label{thm:prod}
	For any two $n\times n$ permutation matrices $P_A,P_B$, the seaweed product $P_C := P_A \boxdot P_B$
	is an $n\times n$ permutation matrix.
	Moreover, given the permutations representing $P_A$ and $P_B$, the permutation representing $P_C$ can be constructed
	in $\Oh(n\log n)$ time.
\end{theorem}

A matrix $M\in \R^{n\times m}$ is a \emph{Monge} matrix if $M^\square$ has non-negative entries,
an \emph{anti-Monge} matrix if $M^\square$ has non-positive entries,
and a \emph{unit-Monge} matrix if $M^\square$ is a permutation matrix.
Note that an $n\times n$ simple unit-Monge matrix $M$ satisfies $M=P^\Sigma$ for the permutation matrix
$P=M^\square$ and thus admits an $\Oh(n)$-space representation based on the permutation representing $P$.

\begin{fact}[\!\!{\cite[Theorem 2.15]{Tiskin13}}]\label{fct:oracle}
	Let $P$ be an $n\times n$ permutation matrix.
	There is a data structure of size $\Oh(n)$ that, given indices $i,j\in [0\dd n]$, computes $P^\Sigma [i,j]$ in  $\Oh(\log n / \log \log n)$-time. Moreover, the data structure can be constructed in $\Oh(n\sqrt{\log n})$
	time given the permutation representing~$P$.
\end{fact}
\newcommand{\Pts}{\mathcal{P}}
\begin{proof}
	Consider a set $\Pts:=\{(i,j)\in [0\dd n)^2 : P[i,j]=1\}$ consisting of exactly $n$ points on the plane.
	Observe that, for every $i,j\in [0\dd n]$, we have $P^\Sigma[i,j] = |\Pts\cap ([i\dd n)\times [0\dd j))|$.
	Thus, a query asking for $P^\Sigma[i,j]$ can be interpreted as an orthogonal range counting query on $\Pts$.
	
	As proved by Chan and Pătraşcu~\cite{DBLP:conf/soda/ChanP10}, these queries can be answered in $\Oh(\log n / \log \log n)$ time using a data structure of size $\Oh(n)$ that can be built in $\Oh(n\sqrt{\log n})$ time given $\Pts$. 
	This set can be expressed as $\{(i,\sigma(i)):i\in [0\dd n)\}$ in terms of the permutation $\sigma$ representing $P$.
\end{proof}

\subsubsection{Grids and Alignment Graphs}
In~\cite[Section 4.3]{Tiskin13}, Tiskin defines an alignment dag of two strings $X$ and $Y$, capturing the structure of the dynamic-programming algorithm for computing their longest common subsequence.
In order to derive the semi-local score matrix, encoding, in particular, the \textsf{LCS} values between $X$ and the substrings of $Y$, as well as between $X$ and the substrings of $Y$, he then extends $X$ with wildcard characters and carefully handles pairs of unreachable vertices (intuitively corresponding to substrings of negative length). 
These complications would be very troublesome in our more general setting, but they can be avoided using undirected edges of cost 0 and 1 instead of directed arcs of score 1 and 0, respectively. 
The price that we pay for this more streamlined approach is a more complicated proof that our alignment graph still encodes all local \textsf{LCS} values.
Another (minor) difference is that our construction is parameterized by an arbitrary set $S\sub [0\dd n)\times [0\dd m)$ rather than a set of the form $\{(x,y) : X[x]=Y[y]\}$ defined in terms of two strings $X$ and $Y$.

For $n,m\in \Zz$, we define the \emph{grid} $\Gr^{n,m} = [0\dd n]\times [0\dd m]$
consisting of $(n+1)(m+1)$ \emph{points}.

\begin{definition}
	Given $n,m\in \Zz$ and $S\sub \Gr^{n-1,m-1}$,
	we define an undirected \emph{alignment graph} $\AG^{n,m}(S)$ with vertices $\Gr^{n,m}$
	and weighted edges:
	\begin{itemize}
		\item $(x,y)\stackrel{1}{\longleftrightarrow} (x+1,y)$ for every $(x,y)\in \Gr^{n-1,m}$,
		\item $(x,y)\stackrel{1}{\longleftrightarrow} (x,y+1)$ for every $(x,y)\in \Gr^{n,m-1}$,
		\item $(x,y)\stackrel{0}{\longleftrightarrow} (x+1,y+1)$ for every $(x,y)\in S$.
	\end{itemize}
\end{definition}
Note that an \emph{alignment graph} $G=\AG^{n,m}(S)$ induces a metric $\dist_G : \Gr^{n,m}\times \Gr^{n,m} \to \Zz$. 
In order to characterize this metric, we define a strict partial order $\prec$ on $\mathbb{Z}^2$
with $(x,y)\prec (x',y')$ if and only if $x< x'$ and $y< y'$.
The underlying partial order $\preceq$ satisfies $(x,y)\preceq (x',y')$ if and only if $(x,y)=(x',y')$
or $(x,y)\prec (x',y')$. Note that this is \emph{not} equivalent to $x\le x'$ and $y\le y'$.

A set $C\sub \mathbb{Z}^2$ is called a \emph{chain} if every two distinct elements of $C$ are comparable with~$\prec$.
Observe that a finite set $C$ is a chain if and only if its elements can be arranged in a sequence $(x_i,y_i)_{i=0}^{|C|-1}$ such that both $(x_i)_{i=0}^{|C|-1}$ and $(y_i)_{i=0}^{|C|-1}$ are (strictly) increasing sequences.
Hence, we also refer to such a sequence as a chain.
Similarly, a set $A\sub \Gr^{n,m}$ is called an \emph{antichain} if no two elements of $A$ are comparable with~$\prec$.
Observe that a finite set $A$ is an antichain if and only if its elements can be arranged in a sequence $(x_i,y_i)_{i=0}^{|A|-1}$ such that $(x_i)_{i=0}^{|A|-1}$ is non-decreasing and $(y_i)_{i=0}^{|A|-1}$ is non-increasing.
Hence, we also refer to such a sequence as an antichain.

For $S\sub \mathbb{Z}^2$, we denote by $\LIS(S)$ the maximum size of a chain contained in $S$.
Observe that $A\sub \mathbb{Z}^2$ is a non-empty antichain if and only if $\LIS(A)=1$.

\begin{lemma}\label{lem:dist}
	Let $p=(x,y)$ and $p'=(x',y')$ be vertices of an alignment graph $G=\AG^{n,m}(S)$.
	\begin{itemize}
		\item\label{it:dist:a} If $x\le x'$ and $y\le y'$, then $\dist_G(p,p')=|x'-x| + |y'-y|-2\LIS(S\cap([x\dd x')\times [y\dd y')))$.
		\item\label{it:dist:b} If $x\le x'$ and $y\ge y'$, then $\dist_G(p,p')=|x'-x| + |y'-y|$.
		\item\label{it:dist:c} If $x\ge x'$ and $y\le y'$, then $\dist_G(p,p')=|x'-x| + |y'-y|$.
		\item\label{it:dist:d} If $x\ge x'$ and $y\ge y'$, then $\dist_G(p,p')=|x'-x| + |y'-y|-2\LIS(S\cap([x'\dd x)\times [y'\dd y)))$.
	\end{itemize}
\end{lemma}
\begin{proof}
	Note that the cases in the lemma statement are not disjoint.
	However, if multiple cases are applicable, then $x=x'$ or $y=y'$, and the provided formulae are consistent due to $\LIS(\emptyset)=0$.
	
	Let us first bound $\dist_G(p,p')$ from above by induction on $|x-x'| + |y-y'|$.
	In the base case of $|x-x'| + |y-y'|=0$, we have $p=p'$ and thus $\dist_G(p,p')=0=|x-x'| + |y-y'|$ holds as claimed.
	
	Next, suppose that $|x-x'| + |y-y'|>0$. By symmetry between the coordinates, we may assume $|x-x'|>0$.
	Moreover, since $G$ is undirected and the claimed formulae for $\dist_G$ are symmetric, we may further assume $x' > x$.
	Let us consider a point $p'' = (x+1,y)$ and note that $G$ contains an edge $p \stackrel{1}{\longleftrightarrow} p''$.
	Furthermore, the inductive assumption, yields $\dist_G(p'',p')\le |x'-x-1|+|y'-y| = |x'-x|+|y'-y|-1$.
	Hence, $\dist_G(p,p')\le 1 + \dist_G(p'',p') \le 1 + |x'-x|+|y'-y|-1 = |x'-x|+|y'-y|$.
	This completes the inductive step provided that $y\ge y'$ or $S\cap([x\dd x')\times [y\dd y'))=\emptyset$.
	
	It remains to consider the complementary case when $y'>y$ and $S\cap([x\dd x')\times [y\dd y'))\ne \emptyset$.
	Let $C$ be a maximum chain contained in $S\cap([x\dd x')\times [y\dd y'))$,
	let $\bp=(\bx,\by)$ be the largest element in $C$ (with respect to $\prec$), and let $\bp' = (\bx+1,\by+1)$.
	The inductive assumption yields $\dist_G(p,\bp) \le |x-\bx|+|y-\by|$
	and, since $C \sm \{\bp\}$ is a chain contained in $S\cap ([\bx+1\dd x')\times [\by+1\dd y'))$,
	that $\dist_G(\bp',p') \le |\bx+1-x'| + |\by+1-y'| - 2(|C|-1)$.
	By definition of $G$, there is an edge $\bp \stackrel{0}{\longleftrightarrow} \bp'$.
	Hence, $\dist_G(p,p') \le \dist_G(p,\bp) + 0 + \dist_G(\bp',p')
	\le |x-\bx|+|y-\by| + |\bx+1-x'| + |\by+1-y'| - 2(|C|-1) = |x-x'| + |y-y'| - 2|C|$.
	This completes the inductive proof of the upper bound on $\dist_G$.
	
	In order to bound $\dist_G$ from below, we need to prove a lower bound on the cost of every walk $W$ from $p$ to $p'$. 
	We proceed by induction on the number of edges in $W$.
	If $p=p'$, then the cost of $W$ is at least $0$ because all edge weights in $G$ are non-negative.
	Thus, we may assume $p\ne p'$, which means that the walk $W$ has at least one edge.
	By symmetry between the coordinates and since $G$ is undirected, we may further assume $x' > x$.
	Let the first edge of the walk $W$ be  $p \stackrel{c}{\longleftrightarrow} p''$
	and let $W'$ be the complementary walk from $p'':=(x'',y'')$ to $p'$.
	
	Note that $x'' \le x+1 \le x'$. If $y'' \ge y'$ and $y\ge y'$, then the inductive assumption guarantees that
	the cost of $W'$ is at least $|x'-x''| + |y'-y''|=x'-x''+y''-y'$,
	and it suffices to prove that the cost of $W$ is at least $|x'-x| + |y'-y| = x'-x + y-y'$,
	i.e., that $c \ge (x'-x + y-y')-(x'-x''+y''-y') = x''-x + y-y''$. We consider two possibilities:
	\begin{itemize}
		\item If $c=1$, then $|x''-x| + |y''-y| = 1$,
		so $c = 1= |x''-x| + |y''-y| \ge x''-x + y-y''$.
		\item If $c=0$, then $x-y = x''-y''$,
		so $c = 0 = (x''-y'') - (x-y) =  x''-x + y-y''$.
	\end{itemize}
	
	Thus, it remains to consider the complementary case when $y'' < y'$ or $y < y'$.
	Due to $|y''-y|\le 1$, this yields $y''\le y'$ \emph{and} $y \le y'$.
	Consequently, the inductive assumption guarantees that the cost of $W'$ is at least 
	$x'-x'' + y'-y'' - 2\LIS(S\cap([x''\dd x')\times [y''\dd y')))$,
	and it suffices to prove that the cost of $W$ is at least 
	$x'-x + y'-y-2\LIS(S\cap([x\dd x')\times [y\dd y')))$. We consider three possibilities:
	\begin{itemize}
		\item If $p''=(x+1,y+1)$, then $c = 0$ and $(x,y)\in S$. Every chain in $S\cap([x+1\dd x')\times [y+1\dd y'))$
		can be extended with $(x,y)$, so $\LIS(S\cap([x''\dd x')\times [y''\dd y'))) \le \LIS(S\cap([x\dd x')\times [y\dd y')))-1$. Hence, the cost of $W$ is at least
		$0 + x'-x-1 + y'-y-1 -2(\LIS(S\cap([x\dd x')\times [y\dd y')))-1)=x'-x + y'-y-2\LIS(S\cap([x\dd x')\times [y\dd y')))$, as claimed.
		\item If $p''\in \{(x+1,y),(x,y+1)\}$, then $c=1$ and $\LIS(S\cap([x''\dd x')\times [y''\dd y'))) \le \LIS(S\cap([x\dd x')\times [y\dd y')))$ by monotonicity of $\LIS$.
		Hence, the cost of $W$ is at least $1 + x'-x + y'-y - 1 -2\LIS(S\cap([x\dd x')\times [y\dd y'))) = x'-x + y'-y-2\LIS(S\cap([x\dd x')\times [y\dd y')))$, as claimed.
		\item If $p'' \in \{(x-1,y-1),(x-1,y),(x,y-1)\}$, then $c = 2-(x-x'')- (y-y'')$ and $\LIS(S\cap([x''\dd x')\times [y''\dd y'))) \le \LIS(S\cap([x\dd x')\times [y\dd y')))+1$ because $(\{x-1\}\times [y\dd y'))\cup ([x\dd x')\times \{y-1\})$ is an antichain.
		Hence, the cost of $W$ is at least $c + x'-x''+y'-y'' - 2(\LIS(S\cap([x\dd x')\times [y\dd y')))+1)
		= x'-x + y'-y-2\LIS(S\cap([x\dd x')\times [y\dd y')))$, as claimed.
	\end{itemize}
	This completes the inductive proof of the lower bound on $\dist_G$.
\end{proof}

\subsubsection{Cut-Paths, Grid Slices, and Slice Alignment Graphs}
The semi-local \textsf{LCS} and \textsf{LIS} values, studied by Tiskin in~\cite[Chapters 4 and 8, respectively]{Tiskin13}, are encoded by the distances between the boundary vertices of the underlying alignment graph $\AG^{n,m}(S)$.
The challenge in proving \cref{thm:oracle} is that we need to encode the distances between vertices on two arbitrary antichains in $\Gr^{n,m}$. For this, we maximally extend these antichains and extract a \emph{slice} of the alignment graph lying ``between'' the two antichains. As shown in \cref{lem:hered}, such a restriction does not affect the distances between vertices within the slice.

We partition $\Gr^{n,m}$ into $n+m+1$ \emph{diagonals} $\Diag^{n,m}_d$ defined as follows for $d\in[0\dd n+m]$:
\[\Diag^{n,m}_d := \{(x,y)\in \Gr^{n,m} : x-y = d - m\}.\]

\begin{definition}
	A \emph{cut-path} in $\Gr^{n,m}$ is an antichain $\pi = (\pi_d)_{d=0}^{n+m}$
	such that $\pi_d \in \Diag^{n,m}_d$.
\end{definition}

The following fact characterizes cut-paths in an alignment graph.
\begin{fact}\label{fct:cdist}
	Let $\pi$ be a cut-path in $\Gr^{n,m}$ and let $G=\AG^{n,m}(S)$ for $S\sub \Gr^{n-1,m-1}$.
	Then, $\pi$ is a path in $G$ (traversing weight-$1$ edges only) and,
	for every $a,b\in [0\dd n+m]$, we have $\dist_G(\pi_a,\pi_b) = |a-b|$.
	Moreover, $\pi$ separates $\{(x,y)\in \Gr^{n,m} : (x,y) \prec \pi_{x-y+m}\}$ 
	from $\{(x,y)\in \Gr^{n,m} : (x,y) \succ \pi_{x-y+m}\}$.
\end{fact}
\begin{proof}
	Let $\pi = (x_d,y_d)_{d=0}^{n+m}$. Consider subsequent points $\pi_d$ and $\pi_{d+1}$
	for $d\in [0\dd n+m)$.
	Note that $x_{d+1}\ge x_d$ and $y_{d+1}\le y_d$ since $\pi$ is an antichain.
	Moreover, $x_{d+1}-y_{d+1} = 1+x_d-y_d$ because $\pi_d \in \Diag^{n,m}_d$ and $\pi_{d+1}\in \Diag^{n,m}_{d+1}$.
	Consequently, $x_{d+1} \le x_d+1$ and $y_{d+1} \ge y_d-1$.
	Thus, $(x_{d+1},y_{d+1})\in \{(x_d+1,y_d),(x_d,y_d-1)\}$. In either case, $G$ contains an edge $\pi_{d} \stackrel{1}{\longleftrightarrow} \pi_{d+1}$. 
	Hence, $\pi$ is indeed a path traversing weight-$1$ edges only.
	
	If $a,b\in [0\dd n+m]$ with $a\ge b$, then $x_a \ge x_b$ and $y_a \le y_b$ since $\pi$ is an antichain.
	Consequently, \cref{lem:dist} yields $\dist(\pi_a,\pi_b) = |x_b-x_a| + |y_b-y_a| = x_a - x_b + y_b - y_a = (x_a-y_a + m) - (x_b-y_b+m) = a -b$ because $\pi_a \in \Diag^{n,m}_a$ and $\pi_b \in \Diag^{n,m}_b$.
	
	It remains to prove the final claim that $\pi$ separates $\{(x,y)\in \Gr^{n,m} : (x,y) \prec \pi_{x-y+m}\}$ 
	from $\{(x,y)\in \Gr^{n,m} : (x,y) \succ \pi_{x-y+m}\}$.
	For a proof by contradiction, suppose that there is an edge $(x,y)\longleftrightarrow (x',y')$ in $G$
	such that $(x,y)\in \Diag^{n,m}_d$ satisfies $(x,y)\prec \pi_d$ and $(x',y')\in \Diag^{n,m}_{d'}$ satisfies
	$(x',y')\succ \pi_{d'}$.
	If $d \le d'$, then $x < x_d \le x_{d'} < x'$, so $x'\ge x+2$, which is clearly impossible.
	Similarly, if $d \ge d'$, then $y < y_d \le y_{d'} < y'$, so $y' \ge y+2$, which is also impossible.
\end{proof}

Not only each cut-path is an antichain, but also each maximal antichain in $\Gr^{n,m}$ is a cut-path.
We will need the following constructive version of the latter statement:
\begin{fact}\label{fct:anti}
	There is an $\Oh(n+m)$-time algorithm that, given an antichain $(x_i,y_i)_{i=0}^{k-1}$ in $\Gr^{n,m}$,
	constructs supersequence that forms a cut-path in $\Gr^{n,m}$.
\end{fact}
\begin{proof}
	Without loss of generality, we assume that $(x_0,y_0) = (0,m)$ and $(x_{k-1},y_{k-1})=(n,0)$.
	If this is not the case, we extend the input antichain accordingly.
	The cut-path $\pi$ is defined as follows:
	\[\pi_d = \begin{cases}
	(x_{i-1},x_{i-1}+m-d)& \text{for }d\in [x_{i-1}-y_{i-1}+m\dd x_{i-1}-y_{i}+m]\text{ and }i\in [1\dd k),\\
	(y_{i}+d-m, y_{i}) & \text{for }d\in [x_{i-1} - y_{i}+m\dd x_{i}-y_{i}+m]\text{ and }i\in [1\dd k).
	\end{cases}\]
	Note that the intervals $[x_{i-1}-y_{i-1}+m\dd x_{i-1}-y_{i}+m]$ and $[x_{i-1} - y_{i}+m\dd x_{i}-y_{i}+m]$
	for $i\in [1\dd k)$ are non-empty because the sequence $(x_i)_{i=0}^{k-1}$ is non-decreasing 
	and the sequence $(y_i)_{i=0}^{k-1}$ is non-increasing.
	Moreover, due to $(x_0,y_0)=(0,m)$ and $(x_{k-1},y_{k-1})=(n,0)$, these intervals cover $[0\dd n+m]$.
	Two subsequent intervals intersect only at their boundaries, where the values are set consistently:
	$\pi_d = (x_i,y_i)$ for $d=x_i-y_i+m$ and $i\in [0\dd k)$,
	as well as $\pi_d = (x_{i-1},y_{i})$ for $d=x_{i-1} - y_{i}+m$ for $i\in [1\dd k)$.
	Hence, $\pi$ is a well-defined supersequence of $(x_i,y_i)_{i=0}^{k-1}$.
	Moreover, it is easy to check that $\pi_d \in \Diag^{n,m}_d$ and $\pi$ is an antichain (the first coordinates are non-decreasing and the second coordinates are non-increasing).
\end{proof}

We extend the the partial order on points to a partial order on cut-paths, with $\pi \preceq \pi'$
if and only if $\pi_d \preceq \pi'_d$ holds for each $d\in [0\dd n+m]$.
Given two cut-paths $\pi \preceq \pi'$, we define \emph{grid slices}
\begin{align*}
	\Gr^{n,m}[\pi\dd \pi']&=\{(x,y)\in \Gr^{n,m} : \pi_{x-y+m} \preceq (x,y) \preceq \pi'_{x-y+m}\},\\
	\Gr^{n,m}[\pi\dd \pi')&=\{(x,y)\in \Gr^{n,m} : \pi_{x-y+m} \preceq (x,y) \prec \pi'_{x-y+m}\}.\end{align*}

\begin{definition}\label{def:sag}
	Let $\pi\preceq \pi'$ be cut-paths in $\Gr^{n,m}$ and let $S\sub \Gr^{n-1,m-1}$.
	We define the \emph{slice alignment graph} $\SAG^{n,m}(\pi,\pi',S)$ as the subgraph of $\AG^{n,m}(S)$
	induced by $\Gr^{n,m}[\pi\dd \pi']$.
\end{definition}

\begin{lemma}\label{lem:hered}
	Let $\pi\preceq \pi'$ be cut-paths in $\Gr^{n,m}$ and let $S\sub \Gr^{n-1,m-1}$.
	Moreover, let $G=\AG^{n,m}(S)$ and $G' = \SAG^{n,m}(\pi,\pi',S)$.
	Then, for every $p,q\in \Gr^{n,m}[\pi\dd \pi']$, we have $\dist_G(p,q)=\dist_{G'}(p,q)$.
\end{lemma}
\begin{proof}
	Since $G'$ is a subgraph of $G$, we trivially have $\dist_G(p,q)\le \dist_{G'}(p,q)$.
	For the converse inequality, we proceed by induction on the minimum number of edges on a shortest path
	$\Pi$ from $p$ to $q$ in $G$.
	If $\Pi$ traverses at most one edge, then $\Pi$ is also a path in $G'$ and thus
	$\dist_{G'}(p,q) = \dist_G(p,q)$.
	If $\Pi$ contains an internal vertex $r\in\Gr^{n,m}[\pi\dd \pi']$ then,
	by the inductive assumption applied to the prefix of $\Pi$ from $p$ to $r$ and the suffix of $\Pi$ from $r$ to $q$, we have $\dist_{G'}(p,q) \le \dist_{G'}(p,r)+\dist_{G'}(r,q) = \dist_G(p,r)+\dist_G(r,q)=\dist_G(p,q)$.
	Thus, we may assume that all internal vertices of $\Pi$ are outside $\Gr^{n,m}[\pi\dd \pi']$
	and that $\Pi$ has at least one internal vertex $r$.
	
	Let $\Lft = \{(x,y)\in \Gr^{n,m} : (x,y)\prec \pi_{x-y+m}\}$ and $\Rgt = \{(x,y)\in \Gr^{n,m} (x,y)\succ \pi'_{x-y+m}\}$. Observe that $\Gr^{n,m}$ forms a disjoint union of $\Lft$, $\Gr^{n,m}[\pi\dd \pi']$, and $\Rgt$.
	Moreover, by \cref{fct:cdist}, each edge leaving $\Lft$ has its other endpoint in $\pi$,
	whereas each edge leaving $\Rgt$ has its other endpoint in $\pi'$.
	Consequently, if $r\in \Lft$, then $p=\pi_a$ and $q=\pi_b$ for some $a,b\in [0\dd n+m]$.
	However, \cref{fct:cdist} then implies $\dist_{G'}(p,q) = |b-a| = \dist_G(p,q)=|b-a|$ because the path following $\pi$ is contained in~$G'$.
	Symmetrically, if $r\in \Rgt$,  then $p=\pi'_a$ and $q=\pi'_b$ for some $a,b\in [0\dd n+m]$
	and, by \cref{fct:cdist}, $\dist_{G'}(p,q)=|b-a| = \dist_G(p,q)$  because the path following $\pi'$ is contained in $G'$.
\end{proof}

\subsubsection{Distance and Seaweed Matrices}
In~\cite[Definitions 4.8 and 4.11]{Tiskin13}, Tiskin uses his alignment dag to define the semi-local score matrix and the seaweed matrix. Maximum-score paths in his alignment dag correspond to shortest paths in our (undirected) alignment graph, so we introduce a distance matrix instead of the semi-local score matrix. 
Both approaches lead to the same seaweed matrix (we stick to the original name even though we do not interpret this matrix in terms of seaweed braids).
Consequently,~\cite[Theorem 4.10]{Tiskin13} can be seen as a special case of \cref{lem:monge} below,
restricted to the two cut-paths $\pi$, $\pi'$ following the boundary of $\AG^{n,m}(S)$ so that $\SAG^{n,m}(\pi,\pi',S)=\AG^{n,m}(S)$.

\begin{definition}
Given a slice alignment graph $G=\SAG^{n,m}(\pi,\pi',S)$, we introduce two matrices $D_G,M_G\in \R^{(n+m+1)\times (n+m+1)}$ with entries defined as follows for $a,b\in [0\dd n+m]$:
\[D_G[a,b] = \dist_{G}(\pi_a,\pi'_b)\quad\text{and}\quad M_G[a,b]=\tfrac12(D_G[a,b]-a+b).\]
Additionally, we set $P_G := M_G^\square$ (note that also $P_G = \frac12 D_G^\square$).
The matrices $D_G$ and $P_G$ are called the \emph{distance matrix} and the \emph{seaweed matrix} of $G$, respectively.
\end{definition}

\begin{lemma}\label{lem:monge}
	For every slice alignment graph $G=\SAG^{n,m}(\pi,\pi',S)$, 
	the matrix $D_G$ is a Monge matrix and the matrix $M_G$ is a simple unit-Monge matrix.
\end{lemma}
\begin{proof}
	Note that $G$ is a planar graph and that the (cyclic) sequence of vertices
	on its outer face is $\pi'_0 = \pi_0, \pi_1,\ldots, \pi_{n+m-1}, \pi_{n+m}=\pi'_{n+m},\pi'_{n+m-1},\ldots,\pi'_1$.
	As first explicitly observed by Fakcharoenphol and Rao~\cite[Section~2.3]{FR}, this yields
	that $D_G$ is a Monge matrix.
	Due to $M_G^\square = \frac12 D_G^\square$, the matrix $M_G$ is thus also a Monge matrix.
	
	Next, let us characterize entries $D_G[a,b]$ with $a\in \{0,n+m\}$ or $b\in \{0,n+m\}$.
	In each case, we use \cref{fct:cdist}.
	\begin{itemize}
		\item If $a=0$, then $D_G[a,b] = \dist_G(\pi_0,\pi'_b)=\dist_G(\pi'_0,\pi'_b) = b$.
		\item If $a=n+m$, then $D_G[a,b] = \dist_G(\pi_{n+m},\pi'_b)=\dist_G(\pi'_{n+m},\pi'_b) = n+m-b$.
		\item If $b=0$, then $D_G[a,b] = \dist_G(\pi_a,\pi'_0)=\dist_G(\pi_a,\pi_0) = a$.
		\item If $b=n+m$, then $D_G[a,b] = \dist_G(\pi_a,\pi'_{n+m})=\dist_G(\pi_a,\pi_{n+m}) = n+m-a$.
	\end{itemize}
	This yields the following characterization of the corresponding entries $M_G[a,b]=\frac12(D_G[a,b]-a+b)$:
	\begin{itemize}
		\item If $a=0$, then $M_G[a,b] = \frac12(b-a+b)=b$.
		\item If $a=n+m$, then $M_G[a,b] = \frac12(n+m-b-a+b)=0$.
		\item If $b=0$, then $M_G[a,b] = \frac12(a-a+b)=0$.
		\item If $b=n+m$, then $M_G[a,b] = \frac12(n-m-a-a+b)=n-m-a$.
	\end{itemize}
	In particular, $M_G[a,b]=0$ if $a=n+m$ or $b=0$, and therefore $M_G$ is a simple matrix,
	i.e., $M_G = P_G^\Sigma$, where $P_G = M_G^\square$.
	Consequently, it remains to prove that $P_G$ is a permutation matrix.
	
	Since $M_G$ is a Monge matrix, we note that the entries of $P_G$ are non-negative.
	Moreover, \cref{lem:dist,lem:hered} yield $D_G[a,b]\equiv (a-b) \pmod{2}$, so $M_G$ and $P_G$ are integer matrices.
	Let us compute the sum of entries in each row and column of $P_G$. 
	For a row $a\in [0\dd n+m)$, we have 
	\begin{multline*}\sum_{b\in [0\dd n+m)} P_G[a,b] = \sum_{b\in [0\dd n+m)} (M_G[a+1,b]+ M_G[a,b+1] - M_G[a,b]- M_G[a+1,b+1]) = \\ M_G[a+1,0] + M_G[a,n+m]-M_G[a,0]-M_G[a+1,n+m] = 0 + (n-m-a) - 0 - (n-m-(a+1)) = 1.
	\end{multline*}
	Similarly, for a column $b\in [0\dd n+m)$, we have
	\begin{multline*}\sum_{a\in [0\dd n+m)} P_G[a,b] = \sum_{a\in [0\dd n+m)} (M_G[a+1,b]+ M_G[a,b+1] - M_G[a,b]- M_G[a+1,b+1]) = \\ M_G[n+m,b] + M_G[0,b+1]-M_G[0,b]-M_G[n+m,b+1] = 0+(b+1)-b -0 = 1.
	\end{multline*}
	Thus, the entries in each row and each column of $P_G$ sum up to $1$.
	Since $P_G$ is a non-negative integer matrix, this means that each row and each column contains exactly one entry $P_G[a,b]=1$, and the remaining entries satisfy $P_G[a,b]=0$. In other words, $P_G$ is a permutation matrix
	and $M_G$ is a unit-Monge matrix.
\end{proof}

\subsubsection{Composition of Slice Alignment Graphs}\label{sec:decomp}
In~\cite[Section 4.5]{Tiskin13}, Tiskin uses \cref{thm:prod} to efficiently retrieve the seaweed matrix for a pair of strings $(XX',Y)$ in terms of two seaweed matrices for $(X,Y)$ and $(X',Y)$.
A direct generalization of his approach lets us combine the seaweed matrices of two slice alignment graphs
sharing a cut-path.

\begin{lemma}\label{lem:decomp}
	Let $G=\SAG^{n,m}(\pi,\pi',S)$ be a slice alignment graph, let $\pi''$ be a cut-path satisfying
	$\pi \preceq \pi'' \preceq \pi'$, and let $G_L = \SAG^{n,m}(\pi,\pi'', S\cap \Gr^{n,m}[\pi\dd \pi''))$
	and $G_R = \SAG^{n,m}(\pi'',\pi',S\cap \Gr^{n,m}[\pi''\dd \pi'))$.
	Then, $D_G = D_{G_L} \odot D_{G_R}$ and $P_G = P_{G_L} \boxdot P_{G_R}$.
	In particular, given the permutations representing $P_{G_L}$ and $P_{G_R}$,
	the permutation representing $P_G$ can be constructed in $\Oh((n+m)\log (n+m))$ time.
\end{lemma}
\begin{proof}
	First, observe that $G_L$ and $G_R$ are subgraphs 
	of $G$ induced by $\Gr^{n,m}[\pi\dd \pi'']$ and $\Gr^{n,m}[\pi''\dd \pi']$, respectively.
	Consequently, for every $a,b,c\in [0\dd n+m]$,
	we have 
	\begin{multline*}
		D_G[a,b]=\dist_G(\pi_a,\pi'_b)\le \dist_G(\pi_a,\pi''_c) + \dist_G(\pi''_c,\pi'_b)
		\le  \dist_{G_L}(\pi_a,\pi''_c) + \dist_{G_R}(\pi''_c,\pi'_b) \\ = D_{G_L}[a,c]+D_{G_R}[c,b].
	\end{multline*}
	Therefore, $D_G[a,b] \le (D_{G_L} \odot D_{G_R})[a,b]$.
	
	For a proof of the converse inequality, consider a shortest path from $\pi_a$ to $\pi'_b$.
	By \cref{fct:cdist}, this path passes through a vertex $\pi''_c$ for some $c\in [0\dd n+m]$.
	Moreover, by \cref{lem:hered}, we have $\dist_{G_L}(p,\pi''_c) = \dist_G(p,\pi''_c)$
	and $\dist_{G_R}(\pi''_c,q) = \dist_{G_L}(\pi''_c,q)$. 
	Consequently,
	\begin{multline*}
		D_G[a,b]= \dist_G(\pi_a,\pi'_b)=\dist_{G}(\pi_a,\pi''_c) + \dist_{G}(\pi''_c,\pi'_b) = \dist_{G_L}(\pi_a,\pi''_c) + \dist_{G_R}(\pi''_c,\pi'_b) \\= D_{G_L}[a,c]+D_{G_R}[c,b] \ge  (D_{G_L} \odot D_{G_R})[a,b].
	\end{multline*}
	This completes the proof that $D_G = D_{G_L} \odot D_{G_R}$.
	
	Next, we note that $M_G = M_{G_L} \odot M_{G_R}$ because
	the following holds for every $a,b\in [0\dd n+m]$:
	\begin{align*}
		(M_{G_L}\odot M_{G_R})[a,b] &= \min_{c\in [0\dd n+m]}(M_{G_L}[a,c] + M_{G_R}[c,b])\\
		&= \min_{c\in [0\dd n+m]}(\tfrac12(D_{G_L}[a,c]-a+c) + \tfrac12(D_{G_R}[c,b]-c+b))\\
		&= \tfrac12\left(\min_{c\in [0\dd n+m]}(D_{G_L}[a,c]+D_{G_R}[c,b])-a+b\right)\\
		&=  \tfrac12\left((D_{G_L}\odot D_{G_R})[a,b]-a+b\right)\\
		&=  \tfrac12(D_G[a,b]-a+b) \\
		&= M_G[a,b].
	\end{align*}
	Consequently, $P_G = M_G^\square = (M_{G_L}\odot M_{G_R})^\square = (P_{G_L}^\Sigma \odot P_{G_R}^\Sigma)^\square
	= P_{G_L} \boxdot P_{G_R}$ holds as claimed.
	The algorithmic claim thus follows from \cref{thm:prod}.
\end{proof}

\subsubsection{Removing Empty Rows and Columns}\label{sec:contract}
The main feature of the alignment graphs $\AG^{n,m}(S)$ originating from \textsf{LIS} instances (as opposed arbitrary \textsf{LCS} instances) is that $|S|=\Oh(n+m)$. 
Moreover, the divide-and-conquer approach suggested by \cref{lem:decomp} allows further reducing $|S|$ in each recursive call. Nevertheless, the grid dimensions remain the same.
In this section, we show that removing empty rows and columns (not containing any element of $S$) changes the seaweed matrix in a very predictable way, as described in \cref{lem:ccontract}. The procedure of \cref{lem:acontract}
uses this characterization to efficiently reverse the impact of the removal on the seaweed matrix.
This lets our divide-and-conquer algorithm reduce the grid dimensions on par with decreasing $|S|$.
Tiskin uses a similar reduction in~\cite[Algorithm 8.2]{Tiskin13}, but the analogue of \cref{lem:ccontract} is much simpler for alignment graphs (compared to slice alignment graphs). 
We also note that \cref{lem:ccontract} follows from the interpretation of seaweed matrices in terms of seaweed braids: it is a simple observation that the seaweeds originating from empty rows or columns are never combed away. 
Nevertheless, we opted for a more tedious proof avoiding the seaweed monoid
(which we would need to formally link to the seaweed matrices of slice alignment~graphs).

For $X\sub \mathbb{R}$ and $x\in \mathbb{R}$, define $\rank_{X}(x)=|\{x'\in X : x'<x\}|$.
Given $X\sub [0\dd n)$ and $Y\sub [0\dd m)$, define a mapping $\dl : \Gr^{n,m} \to \Gr^{|X|,|Y|}$ with
\[\dl(x,y) = (\rank_X(x),\rank_Y(y)).\]
We extend $\dl$ to map cut-paths in $\Gr^{n,m}$ to cut-paths in $\Gr^{|X|,|Y|}$.
\begin{definition}
	Let $\pi$ be a cut-path in $\Gr^{n,m}$, let $X\sub [0\dd n)$ and $Y\sub [0\dd m)$. The sequence $\dl(\pi)$ is obtained from $\dl(\pi_0),\ldots,\dl(\pi_{n+m})$ by removing duplicate adjacent elements.
\end{definition}

\begin{observation}\label{obs:cp}
	Let $\pi$ be a cut-path in $\Gr^{n,m}$. For every $x\in [0\dd n)$, there exists a unique index $d\in [0\dd n+m)$, denoted $\row(\pi,\bx)$, such that $x_{d} = x < x_{d+1}$. Moreover, for every $y \in [0\dd m)$, there exists a unique index $d \in [0\dd n+m)$, denoted $\col(\pi,y)$, such that $y_{d+1}=y < y_d$.
\end{observation}

\begin{fact}\label{fct:dlpi}
	Let $\pi$ be a cut-path in $\Gr^{n,m}$, $X\sub [0\dd n)$, and $Y\sub [0\dd m)$.
	Then, $\tpi:=\dl(\pi)$ is a cut-path in $\Gr^{|X|,|Y|}$ and, for $d\in [0\dd n+m]$, we have
	$\dl(\pi_d) = \tpi_{\rank_R(d)}$, where $R = \{\row(\pi,x) : x\in X\}\cup\allowbreak \{\col(\pi,y): y\in Y\}$.
\end{fact}
\begin{proof}
	Let $\pi = (x_d,y_d)_{d=0}^{n+m}$.
	Observe that $x_{d+1}\ne x_d$ if and only if $d = \row(\pi,x_d)$,
	so $\rank_X(x_{d+1})\ne \rank_X(x_d)$ if and only if $d\in \{\row(\pi,x) : x\in X\}$.
	Similarly, $y_{d+1} \ne y_d$ if and only if $d = \col(\pi,y_{d+1})$,
	so $\rank_Y(y_{d+1})\ne \rank_Y(y_d)$ if and only if $d\in \{\col(\pi,y) : y\in Y\}$.
	Hence, $\dl(\pi_{d+1})\ne \dl(\pi_d)$ if and only if $d\in R$,
	and therefore $\dl(\pi_d) = \tpi_{\rank_R(d)}$.
	In particular, $\tpi$ has $|X|+|Y|+1$ elements.
	
	Denote $\tpi = (\tx_d,\ty_d)_{d=0}^{|X|+|Y|}$.
	Since $(x_d)_{d=0}^{n+m}$ is non-increasing, so is $(\rank_X(x_d))_{d=0}^{n+m}$
	and its subsequence $(\tx_d)_{d=0}^{|X|+|Y|}$.
	Symmetrically, since $(y_d)_{d=0}^{n+m}$ is non-decreasing, so is 
	$(\rank_Y(y_d))_{d=0}^{n+m}$
	and its subsequence $(\ty_d)_{d=0}^{|X|+|Y|}$.
	Consequently, $\tpi$ is an antichain.
	Moreover, since each diagonal in $\Gr^{|X|,|Y|}$ is a chain,
	$\tpi$ contains exactly one entry from each diagonal, and thus it must be a cut-path.
\end{proof}

\begin{lemma}\label{lem:ccontract}
	Let $G=\SAG^{n,m}(\pi,\pi',S)$ be a slice alignment graph and let 
	$X \sub [0\dd n)$ and $Y\sub [0\dd m)$ be such that $S \sub X \times Y$.
	Then, $\tG := \SAG^{|X|,|Y|}(\dl(\pi),\dl(\pi'),\{\dl(p) : p\in S\})$
	is a well-defined slice alignment graph.
	Moreover,
	\[ P_G[a,b] = \begin{cases}
	P_{\tG}[\rank_R(a),\rank_{R'}(b)]&\text{if }a\in R\text{ and }b\in R',\\
	1 & \text{if }a=\row(\pi,x)\text{ and }b=\row(\pi',x)\text{ for some }x\in \bar{X},\\
	1 & \text{if }a=\col(\pi,y)\text{ and }b=\col(\pi',y)\text{ for some }y\in \bar{Y},\\
	0 & \text{otherwise,}
	\end{cases}\]
	where $R =  \{\row(\pi,x) : x\in X\}\cup\allowbreak \{\col(\pi,y): y\in Y\}$,
	$R' =  \{\row(\pi',x) : x\in X\}\cup\allowbreak \{\col(\pi',y): y\in Y\}$,
	$\bar{X} = [0\dd n)\sm X$, and $\bar{Y} = [0\dd m)\sm Y$.
\end{lemma}
\begin{proof}
	Denote $\pi = (x_d,y_d)_{d=0}^{n+m}$ and $\pi' = (x'_d,y'_d)_{i=0}^{n+m}$,
	as well as $\dx(\pi)=\tpi = (\tx_d,\ty_d)_{d=0}^{|X|+|Y|}$ and 
	$\dx(\pi') = \tpi' = (\tx'_d,\ty'_d)_{d=0}^{|X|+|Y|}$.
	
	By \cref{fct:dlpi}, both $\tpi$ and $\tpi'$ are cut-paths in $\Gr^{|X|,|Y|}$.
	In order to prove that $\tG$ is well-defined, we need to show that $\tpi_d \preceq \tpi'_d$ holds for every $d\in [0\dd |X|+|Y|)$.
	Let us choose $a,b\in [0\dd n+m]$ so that $\tpi_d = \dl(x_a,y_a)$ and $\tpi'_{d}=\dl(x'_b,y'_b)$.
	If $a \le b$, then we have $\tx_d = \rank_X(x_a) \le \rank_X(x_b) \le \rank_X(x'_b) = \tx'_d$.
	Similarly, if $a \ge b$, then we have $\ty_d = \rank_Y(y_a) \le \rank_Y(y_b) \le \rank_Y(y'_b) = \ty'_d$.
	In either case, due to $\tpi_d,\tpi'_d \in \Diag^{|X|,|Y|}_{d}$, this implies $\tpi_d \preceq \tpi'_d$.
	Consequently, $\tpi \preceq \tpi'$ and $\tG$ is well-defined.
	
	\begin{claim}\label{clm:contract}
		Let $p=(x,y)$ and $q=(x',y')$ be points in $\Gr^{n,m}[\pi\dd \pi']$ such that $\dl(p),\dl(q)\in \Gr^{|X|,|Y|}[\tpi\dd \tpi']$.
		Then, $\dist_G(p,q) = \dist_{\tG}(\dl(p),\dl(q))+ |\rank_{\bar{X}}(x)-\rank_{\bar{X}}(x')| +
		|\rank_{\bar{Y}}(y)-\rank_{\bar{Y}}(y')|$.
	\end{claim}
	\begin{proof}
		Denote $(\tx,\ty)=\dl(p)$, $(\tx',\ty')=\dl(q)$, and $\tS = \{\dx(r) : r\in S\}$.
		The assumption $S \sub X\times Y$ implies that $\dl$ restricted to $S$ 
		is a monotonically increasing bijection mapping $S$ to $\tS$,
		i.e., $\dl$ preserves chains. 
		Consequently, if $x\le x'$ and $y \le y'$,
		then $\LIS(S\cap ([x\dd x')\times [y\dd y'))) = \LIS(\tS\cap ([\tx\dd \tx')\times [y\dd y')))$,
		so \cref{lem:dist,lem:hered} yield \begin{multline*}
			\dist_G(p,q) - \dist_{\tG}(\dx(p),\dx(q)) = (|x'-x|+|y'-y|-2\LIS(S\cap ([x\dd x')\times [y\dd y'))))-\\(|\tx'-\tx|+|y'-y|-2\LIS(\tS\cap ([\tx\dd \tx')\times [y\dd y')))) = |x'-x|-|\tx'-\tx| + |y'-y|-|\ty'-\ty|.\end{multline*}
		Similarly, if $x\le x'$ and $y \ge y'$, then \cref{lem:dist,lem:hered} yield
		\begin{multline*}\dist_G(p,q) - \dist_{\tG}(\dx(p),\dx(q)) = (|x'-x|+|y'-y|)-(|\tx'-\tx|+|y'-y|) =\\ 
			|x'-x|-|\tx'-\tx|+ |y'-y|-|\ty'-\ty|.
		\end{multline*}
		The cases involving $x\ge x'$ are symmetric.
		Furthermore,
		\[|x'-x|-|\tx'-\tx| = |x'-x| - |\rank_X(x')-\rank_X(x)| = |x'-\rank_X(x')-x+\rank_X(x)| = |\rank_{\bar{X}}(x')-\rank_{\bar{X}}(x)|.\]
		Symmetrically,
		\[|y'-y|-|\ty'-\ty| = |y'-y| - |\rank_Y(y')-\rank_Y(y)| = |y'-\rank_Y(y')-y+\rank_Y(y)| = |\rank_{\bar{Y}}(y')-\rank_{\bar{Y}}(y)|,\]
		so $\dist_G(p,q) - \dist_{\tG}(\dx(p),\dx(q)) =  |x'-x|-|\tx'-\tx|+ |y'-y|-|\ty'-\ty| =   |\rank_{\bar{X}}(x')-\rank_{\bar{X}}(x)|+ |\rank_{\bar{Y}}(y')-\rank_{\bar{Y}}(y)|$ holds as claimed.
	\end{proof}
	
	Next, we use \cref{fct:dlpi,clm:contract} to characterize the entries of $D_G$ in terms of $D_{\tG}$.
	For every $a,b\in [0\dd n+m]$, we have 
	\begin{multline*}D_G[a,b] = \dist_{G}(\pi_a,\pi'_b) = \dist_{\tG}(\dl(\pi_a),\dl(\pi'_b)) + |\rank_{\bar{X}}(x'_b)-\rank_{\bar{X}}(x_a)| + |\rank_{\bar{Y}}(y'_b)-\rank_{\bar{Y}}(y_a)| \\ 
		= D_{\tG}[\rank_R(a),\rank_{R'}(b)] +  |\rank_{\bar{X}}(x'_b)-\rank_{\bar{X}}(x_a)| + |\rank_{\bar{Y}}(y'_b)-\rank_{\bar{Y}}(y_a)|.\end{multline*}
	Let us decompose $D_G = T_{\tG} + T_X + T_Y$ into three matrices corresponding to the three terms above.
	\begin{claim}\label{clm:T}
		For every $a,b\in [0\dd n+m)$, we have
		\begin{align*}T_{\tG}^\square[a,b] &= \begin{cases}
				D_G^{\square}[\rank_R(a),\rank_{R'}(b)] &\text{if }a\in R\text{ and }b\in R',\\
				0 & \text{otherwise.}
			\end{cases}\\
			T_X^\square[a,b] &= \begin{cases}
				2\phantom{_G^{\square}[\rank_R(a),\rank_{R'}(b)]\;} & \text{if }a=\row(\pi,x)\text{ and }b=\row(\pi',x)\text{ for some }x\in \bar{X},\\
				0 & \text{otherwise.}
			\end{cases}\\
			T_Y^\square[a,b] &= \begin{cases}
				2\phantom{_G^{\square}[\rank_R(a),\rank_{R'}(b)]\;} & \text{if }a=\col(\pi,y)\text{ and }b=\col(\pi',y)\text{ for some }y\in \bar{Y},\\
				0 & \text{otherwise.}
			\end{cases}
		\end{align*}
	\end{claim}
	\begin{proof}
		If $a\in R$ and $b\in R'$, then $\rank_R(a+1) = \rank_R(a)+1$ and $\rank_{R'}(b+1)=\rank_{R'}(b)+1$,
		so
		$T_{\tG}^\square[a,b] = D_{\tG}^{\square}[\rank_R(a)+1,\rank_{R'}(b)] + D_{\tG}^{\square}[\rank_R(a),\rank_{R'}(b)+1] - D_{\tG}^{\square}[\rank_R(a),\rank_{R'}(b)] - D_{\tG}^{\square}[\rank_R(a)+1,\rank_{R'}(b)+1]=D_{\tG}^{\square}[\rank_R(a),\rank_{R'}(b)]$.
		If $a\notin R$, then $\rank_R(a+1)=\rank_R(a)$,
		so $T_{\tG}^\square[a,b] = D_{\tG}^{\square}[\rank_R(a),\rank_{R'}(b)] + D_{\tG}^{\square}[\rank_R(a),\rank_{R'}(b+1)] - D_{\tG}^{\square}[\rank_R(a),\rank_{R'}(b)] - D_{\tG}^{\square}[\rank_R(a),\rank_{R'}(b+1)]=0$.
		Finally, if $b\notin R'$, then $\rank_{R'}(b+1) =\rank_{R'}(b)$, which symmetrically yields $T_{\tG}^\square[a,b] = 0$.
		
		Next, observe that if $\rank_{\bar{X}}(x_{a}) \ge \rank_{\bar{X}}(x'_{b+1})$,
		then $T_X^\square[a,b] = (\rank_{\bar{X}}(x_{a+1})-\rank_{\bar{X}}(x'_{b})) + (\rank_{\bar{X}}(x_{a})-\rank_{\bar{X}}(x'_{b+1}))- (\rank_{\bar{X}}(x_{a})-\rank_{\bar{X}}(x'_{b})) - (\rank_{\bar{X}}(x_{a+1})-\rank_{\bar{X}}(x'_{b+1})) = 0$.
		Symmetrically, if $\rank_{\bar{X}}(x'_{b}) \ge \rank_{\bar{X}}(x_{a+1})$,
		then $T_X^\square[a,b] = (\rank_{\bar{X}}(x'_{b})-\rank_{\bar{X}}(x_{a+1})) + (\rank_{\bar{X}}(x'_{b+1})-\rank_{\bar{X}}(x_{a}))- (\rank_{\bar{X}}(x'_{b})-\rank_{\bar{X}}(x_{a})) - (\rank_{\bar{X}}(x'_{b+1})-\rank_{\bar{X}}(x_{a+1})) = 0$.
		In the remaining case, we have $\rank_{\bar{X}}(x_{a}) < \rank_{\bar{X}}(x'_{b+1})\le \rank_{\bar{X}}(x'_{b})+1 < \rank_{\bar{X}}(x_{a+1})+1 \le \rank_{\bar{X}}(x_{a})+2$.
		Since all the ranks are integers, this yields $\rank_{\bar{X}}(x_{a+1})=\rank_{\bar{X}}(x_{a})+1=\rank_{\bar{X}}(x'_{b})+1=\rank_{\bar{X}}(x'_{b+1})$. Consequently, $a = \row(\pi,x_a)$, $b = \row(\pi',x'_b)$, and $x_a = x'_b \in \bar{X}$. In this case, we have $T_X^\square[a,b] = 1+1-0-0 = 2$, as claimed.
		
		Finally, observe that if $\rank_{\bar{Y}}(y'_{b}) \le \rank_{\bar{Y}}(y_{a+1})$,
		then $T_Y^\square[a,b] = (\rank_{\bar{Y}}(y_{a+1})-\rank_{\bar{Y}}(y'_{b})) + (\rank_{\bar{Y}}(y_{a})-\rank_{\bar{Y}}(y'_{b+1}))- (\rank_{\bar{Y}}(y_{a})-\rank_{\bar{Y}}(y'_{b})) - (\rank_{\bar{Y}}(y_{a+1})-\rank_{\bar{Y}}(y'_{b+1})) = 0$.
		Symmetrically, if $\rank_{\bar{Y}}(y_{a}) \le \rank_{\bar{Y}}(y'_{b+1})$,
		then $T_Y^\square[a,b] = (\rank_{\bar{Y}}(y'_{b})-\rank_{\bar{Y}}(y_{a+1})) + (\rank_{\bar{Y}}(y'_{b+1})-\rank_{\bar{Y}}(y_{a}))- (\rank_{\bar{Y}}(y'_{b})-\rank_{\bar{Y}}(y_{a})) - (\rank_{\bar{Y}}(y'_{b+1})-\rank_{\bar{Y}}(y_{a+1})) = 0$.
		In the remaining case, we have $\rank_{\bar{Y}}(y_{a}) > \rank_{\bar{Y}}(y'_{b+1})\ge \rank_{\bar{Y}}(y'_{b})-1 > \rank_{\bar{Y}}(y_{a+1})-1 \ge \rank_{\bar{Y}}(y_{a})+2$.
		Since all the ranks are integers, this yields $\rank_{\bar{Y}}(y_{a+1})=\rank_{\bar{Y}}(y_{a})-1=\rank_{\bar{Y}}(y'_{b})-1=\rank_{\bar{Y}}(y'_{b+1})$. Consequently, $a = \col(\pi,y_{a+1})$, $b = \col(\pi',y'_{b+1})$, and $y_{a+1} = y'_{b+1} \in \bar{Y}$. In this case, we have $T_Y^\square[a,b] = 1+1-0-0 = 2$, as claimed.
	\end{proof}
	
	\cref{clm:T} completes the proof due to $P_G = \frac12D_G^\square = \frac12(T_{\tG}^\square + T_X^\square + T_Y^\square)$ and $P_{\tG} = \frac12 D_{\tG}^\square$.
\end{proof}

Next, we develop an algorithmic counterpart of \cref{lem:ccontract}.

\begin{lemma}\label{lem:acontract}
	Let $G=\SAG^{n,m}(\pi,\pi',S)$ be a slice alignment graph, let 
	$X \sub [0\dd n)$ and $Y\sub [0\dd m)$ be such that $S \sub X \times Y$,
	and let $\tG = \SAG^{|X|,|Y|}(\dl(\pi),\dl(\pi'),\{\dl(p) : p\in S\})$.
	Given the cut-paths $\pi,\pi'$, the sets $X,Y$, and the permutation representing $P_{\tG}$,
	the permutation representing $P_G$ can be constructed in $\Oh(n+m)$ time.
\end{lemma}
\begin{proof}
	Let $\sigma:[0\dd n+m)\to [0\dd n+m)$ be the permutation underlying $P_G$
	and let $\tsigma: [0\dd |X|+|Y|)\to [0\dd |X|+|Y|)$ be the permutation underlying $P_{\tG}$.
	By \cref{lem:ccontract}, we have $\sigma(\row(\pi,\bx))=\row(\pi',\bx)$ for $\bx\in \bar{X}$
	and $\sigma(\col(\pi,\by))=\col(\pi',\by)$ for $\by \in \bar{Y}$.
	In order to fill these values of $\sigma$, we just need to construct the functions 
	$\row(\pi,\cdot),\row(\pi',\cdot):[0\dd n)\to [0\dd n+m)$ and  $\col(\pi,\cdot),\col(\pi',\cdot):[0\dd m)\to [0\dd n+m)$.
	By the characterization of \cref{obs:cp}, $\row(\pi,\cdot)$ and $\col(\pi,\cdot)$ can be constructed by scanning $\pi$.  An analogous scan of $\pi'$ yields $\col(\pi',\cdot)$ and $\col(\pi',\cdot)$.
	
	The first phase of the algorithm thus results in the values of $\sigma$
	for arguments in $\{\row(\pi,\bx) : \bx \in \bar{X}\}\cup\{\col(\pi,\by) : \by\in \bar{Y}\} = [0\dd n+m]\sm R$.
	In the second phase, will retrieve from $\tsigma$ the values of $\sigma$ for arguments in $R$.
	By \cref{lem:ccontract}, for every $a,b\in [0\dd |X|+|Y|]$,
	we have $P_{\tG}[a,b]=P_G[r_a,r'_b]$, where $r_0,\ldots,r_{|X|+|Y|}$ are the elements
	of $R$ in the increasing order and $r'_0,\ldots,r'_{|X|+|Y|}$ are the elements of $R'$ in the increasing order.
	In particular, $\sigma(r_a)=r'_{\tsigma(a)}$ for every $a\in [0\dd |X|+|Y|)$.
	The sets $R$ and $R'$ can be constructed using the already available functions $\row(\pi,\cdot),\row(\pi',\cdot),\col(\pi,\cdot),\col(\pi',\cdot)$ and sorted by scanning $[0\dd n+m]$ from left to right.
	
	Overall, the second phase of the algorithm results in the (remaining) values of $\sigma$ for arguments arguments in $R$.
	It is easy to see that the running time of the entire algorithm is $\Oh(n+m)$.    
\end{proof}

\subsubsection{Efficient Distance Oracle}
In this section, we combine the insight from \cref{sec:decomp,sec:contract}
to develop a divide-and-conquer algorithm constructing the seaweed matrix $P_G$
of a given slice alignment graph $G=\SAG^{n,m}(\pi,\pi',S)$.
This algorithm is optimized for the setting when $S$ is a sparse subset of $\Gr^{n-1,m-1}$, and it generalizes \cite[Algorithm 8.2]{Tiskin13}.
As a corollary, we derive an efficient construction procedure for an oracle providing random access to the distance matrix $D_G$. 

\begin{proposition}\label{prp:pg}
	Given a slice alignment graph $G=\SAG^{n,m}(\pi,\pi',S)$ (represented by $\pi$, $\pi'$, and~$S$), the permutation representing $P_G$ can be constructed in $\Oh(1+n+m+|S|\log^2|S|)$ time.
\end{proposition}
\begin{proof}
	We develop a recursive divide-and-conquer algorithm.
	The points in $S\sm \Gr^{n,m}[\pi\dd \pi')$ do not contribute any edge in $G$, so they are removed from $S$ in a preprocessing step of the algorithm. 
	Then, the algorithm computes $X=\{x : (x,y)\in S\}$ 
	and $Y=\{y : (x,y)\in S\}$. 
	This can be implemented in $\Oh(n+m+|S|)$ time by iterating over the points in $S$, with $X\sub [0\dd n)$ and $Y\sub [0\dd m)$ maintained as characteristic vectors.
	
	If $X \ne [0\dd n)$ or $Y \ne [0\dd m)$, the algorithm reduces the grid dimensions based on the combinatorial insight of \cref{sec:contract}.
	First, we construct the functions $\rank_X:[0\dd n]\to [0\dd |X|]$ and $\rank_Y:[0\dd m]\to [0\dd m]$
	so that $\dl(x,y)$ can be retrieved in $\Oh(1)$ time for every $(x,y)\in \Gr^{n,m}$.
	Next, we build $\tpi := \dl(\pi)$, $\tpi' := \dl(\pi')$, and $\tS := \{\dl(p): p\in S\}$.
	By \cref{lem:ccontract}, this yields a slice alignment graph $\tG = \SAG^{|X|,|Y|}(\tpi,\tpi',\tS)$,
	which is processed recursively,
	with the algorithm of \cref{lem:acontract} applied to transform the permutation representing $P_{\tG}$
	to the permutation representing $P_G$.
	
	It remains to consider the case when $X=[0\dd n)$ and $Y=[0\dd m)$.
	In particular, if $|S|=0$, then $n=m=0$, $D_G=\begin{bmatrix}0\end{bmatrix}$, and $P_G$ is the empty ($0\times 0$) matrix (represented by the empty permutation).
	Similarly, if $|S|=1$, then we must have $\pi = ((0,1),(0,0),(1,0))$,
	$\pi'=(0,1),(1,1),(1,1)$, and $S=\{(0,0)\}$. 
	Consequently,
	\[D_G = \begin{bmatrix}
	0 & 1 & 2\\
	1 & 0 & 1\\
	2 & 1 & 0
	\end{bmatrix}\qquad\text{and}\qquad
	P_G = \tfrac12 D_G^\square = \begin{bmatrix}
	1 & 0\\
	0 & 1
	\end{bmatrix},\]
	which is represented by the identity permutation on $\{0,1\}$.
	
	Thus, we may henceforth assume $|S| \ge \max(n,m,2)$.
	In this case, the algorithm decomposes $G$ into two smaller slice alignment graphs based on the combinatorial insight of \cref{sec:decomp}.
	First, we partition $S$ into $S_L$
	and $S_R$ so that $|S_L|=\lceil\frac12|S|\rceil$,
	$|S_R| = \lfloor\frac12 |S|\rceil$, and points in $S_L$
	are lexicographically smaller than points in $S_R$
	(note that both sets are non-empty).
	
	Let $(\bx,\by)$ be the lexicographically smallest element of $S_R$.
	We define the following cut-path $\pi^-:
	(0,m),\ldots, (\bx,m),\ldots, (\bx,\by),(\bx+1,\by),\ldots,
	(\bx+1,0),\ldots, (n,0)$.
	Formally,
	\[\pi^-_d = \begin{cases}
	(d,m) & \text{if }d\in [0\dd \bx],\\
	(\bx,\bx+m-d) & \text{if }d\in [\bx\dd \bx-\by+m],\\
	(\bx+1,\bx+m-d) & \text{if }d\in (\bx-\by+m\dd \bx+m],\\
	(d-m,0) & \text{if }d\in [\bx+m\dd n+m].
	\end{cases}\]
	This way, every point $p\in S\cap \Diag^{n,m}_d$
	satisfies $p \prec \pi^-_d$ if $p\in S_L$
	and $\pi^-_d \preceq p$ if $p\in S_R$.
	
	Next, we define another cut-path $\pi''$ with
	\[\pi''_d = \begin{cases}
	\pi_d & \text{if }\pi^-_d \preceq \pi_d \preceq \pi'_d,\\
	\pi^-_d & \text{if }\pi_d \preceq \pi^-_d \preceq \pi'_d,\\
	\pi'_d & \text{if }\pi_d \preceq \pi'_d \preceq \pi^-_d.\\
	\end{cases}\]
	This guarantees $\pi \preceq \pi'' \preceq \pi'$.
	Moreover, since $S \sub \Gr^{n,m}[\pi\dd \pi')$,
	we have $S_L = S \cap \Gr^{n,m}[\pi\dd \pi'')$ and $S_R = S \cap \Gr^{n,m}[\pi''\dd \pi')$.
	This yields slice alignment graphs $G_L := \SAG^{n,m}(\pi,\pi',S_L)$ and $G_R:=\SAG^{n,m}(\pi'',\pi',S_R)$.
	These graphs are processed recursively and then \cref{lem:decomp}
	is used to derive the permutation representing $P_G$ from the permutations representing $P_{G_L}$ and $P_{G_R}$.
	
	It remains to analyze the running time. For this, we interpret the grid size reduction
	as a preprocessing step rather than a standalone recursive call.
	If $|S|\le 1$, then the algorithm takes $\Oh(n+m+1)$ time.
	Otherwise, it takes $\Oh(n+m+|S|\log |S|)$ time and makes two recursive calls.
	The grid dimensions in these calls do not exceed $|S|$ and the sets 
	$S_L$ and $S_R$ in the calls are of size at most $\lceil \frac12 |S|\rceil$.
	This yields an overall bound of $\Oh(1+n +m + |S|\log^2 |S|)$
	on the running time.
\end{proof}

\begin{corollary}\label{cor:oracle}
	For every slice alignment graph $G=\SAG^{n,m}(\pi,\pi',S)$,
	there is a data structure of size $\Oh(n + m)$ that, given any $a,b\in [0\dd n+m]$,
	computes $D_G[a,b]$ in $\Oh(\log (n+m) / \log \log (n+m))$ time.
	Moreover, the data structure can be constructed in $\Oh((n+m)\sqrt{\log(n+m)} + |S|\log^2 |S|)$
	time given $\pi$, $\pi'$, and~$S$.
\end{corollary}
\begin{proof}
	Note that $D_G[a,b]=2M_G[a,b]+a-b = 2P^{\Sigma}_G[a,b]+a-b$ holds for $a,b\in [0\dd n+m]$.
	Hence, it suffices to store the data structure of \cref{fct:oracle} providing random access to $P_G^{\Sigma}$,
	which takes $\Oh(n+m)$ space and answers queries in $\Oh(\log (n+m) / \log \log (n+m))$ time.
	The construction time is $\Oh((n+m)\sqrt{\log(n+m)})$ from the permutation representing $P_G$,
	which can be built in $\Oh(n+m+|S|\log^2 |S|)$ time using \cref{prp:pg}.
\end{proof}

\subsubsection{Applications to LIS}
In this section, we provide a 3-step proof of \cref{thm:oracle}.
In \cref{lem:L}, we use \cref{lem:dist,lem:hered} to interpret the outcome of \cref{cor:oracle}
in terms of the values $\LIS(S')$ for appropriate subsets $S'\sub S$. 
Here, the main technical challenge is to make sure that these values form an anti-Monge matrix
even though some entries in this matrix correspond to degenerate queries. 
\cref{cor:M} generalizes \cref{lem:L} so that two arbitrary antichains are supported 
instead of two cut-paths.
We use \cref{fct:anti} to extend antichains to cut-paths, but then extra care is needed to obtain cut-paths satisfying $\pi \preceq \pi'$.
Finally, we derive \cref{thm:oracle} by interpreting a sequence $(a_i)_{i=0}^{n-1}$ as a set $S\sub \Gr^{n-1,n-1}$ and its non-increasing subsequences as antichains in $\Gr^{n,n}$.

\begin{lemma}\label{lem:L}
	Let $G=\SAG^{n,m}(\pi,\pi',S)$ be a slice alignment graph with $\pi = (x_d,y_d)_{d=0}^{n+m}$
	and $\pi' = (x'_d,y'_d)_{d=0}^{n+m}$, and let $N>0$ be an integer.
	There exists an anti-Monge matrix $L\in \R^{(n+m+1)\times (n+m+1)}$ such that,
	for every $i,j\in [0\dd n+m+1]$, we have
	\begin{align*}
		L[i,j] &= \LIS(S\cap [x_i\dd x'_j)\times [y_i\dd y'_j)) &\text{if }x_i \le x'_j\text{ and }y_i \le y'_j,\\
		L[i,j] & \le -N & \text{otherwise.}
	\end{align*}
	Moreover, after $\Oh((n+m)\sqrt{\log(n+m)} + |S|\log^2 |S|)$-time preprocessing,
	any entry of $L$ can be computed in $\Oh(\log(n+m)/ \log \log (n+m))$ time.
\end{lemma}
\begin{proof}
	Let us define another slice alignment graph $G' = \SAG^{n,m}(\pi,\pi',\emptyset)$.
	By \cref{lem:monge}, both $D_G$ and $D_{G'}$ are Monge matrices.
	Moreover, define a matrix $A\in \R^{(n+m+1)\times (n+m+1)}$
	so that $A[i,j] = x'_j-x_i+y'_j-y_i$ and note that $A^\square$ is zero matrix,
	i.e., $A$ is both a Monge and an anti-Monge matrix.
	
	We define the matrix $L$ as the following linear combination of $D_G$, $D_{G'}$, and $A$:
	\[L = \tfrac12(N\cdot A - D_G - (N-1)\cdot D_{G'}).\]
	It is an anti-Monge matrix because  $D_G$ and $D_{G'}$ are Monge matrices
	whereas $A$ is an anti-Monge matrix (since $A^\square$ is a zero matrix).
	
	It remains to check whether $L$ satisfies the required conditions.
	\begin{itemize}
		\item If $x_i \le x'_j$ and $y_i \le y'_j$, then \cref{lem:dist,lem:hered} yield
		$D_G[i,j] = x'_j - x_i + y'_j - y_i - 2\LIS(S\cap [x_i\dd x'_j)\times [y_i\dd y'_j))$
		and $D_{G'}[i,j] = x'_j - x_i + y'_j - y_i$. 
		Hence, $2L[i,j] = N(x'_j-x_i+y'_j-y_i) - (x'_j - x_i + y'_j - y_i - 2\LIS(S\cap [x_i\dd x'_j)\times [y_i\dd y'_j)))-(N-1)(x'_j - x_i + y'_j - y_i) = 2\LIS(S\cap [x_i\dd x'_j)\times [y_i\dd y'_j))$
		holds as claimed.
		\item If $x_i \le x'_j$ and $y_i > y'_j$, then \cref{lem:dist,lem:hered} yield
		$D_G[i,j] = D_{G'}[i,j] = x'_j - x_i + y_i - y'_j$.
		Hence, $2L[i,j] = N(x'_j-x_i+y'_j-y_i) - N(x'_j - x_i + y_i - y'_j) = 2N(y'_j - y_i) \le -2N$ holds as claimed.
		\item If $x_i  > x'_j$ and $y_i \le y'_j$, then \cref{lem:dist,lem:hered} yield
		$D_G[i,j] = D_{G'}[i,j] = x_i - x'_j + y'_j - y_j$.
		Hence, $2L[i,j] = N(x'_j-x_i+y'_j-y_i) - N(x_i - x'_j + y'_j - y_j) = 2N(x'_j - x_i) \le -2N$ holds as claimed.
		\item If $x_i > x'_j$ and $y_i  > y'_j$, then $D_G[i,j]\ge 0$ and \cref{lem:dist,lem:hered} yields $D_{G'}[i,j] = x_i - x'_j + y_i - y'_j$.
		Hence, $2L[i,j] \le N(x'_j-x_i+y'_j-y_i)-(N-1)(x_i - x'_i + y_i - y'_j) = (2N-1)(x'_j-x_i+y'_j-y_i) \le -(4N-2) \le -2N$ holds as claimed.\qedhere
	\end{itemize}
\end{proof}

Our next goal is to generalize \cref{lem:L} from cut-paths $\pi\preceq \pi'$ to arbitrary antichains.
\begin{corollary}\label{cor:M}
	Let $(x_i,y_i)_{i=0}^{k-1}$ and $(x'_j,y'_j)_{j=0}^{\ell-1}$ be antichains in $\Gr^{n,m}$,
	let $S\sub \Gr^{n-1,m-1}$, and let $N > 0$ be a positive integer.
	There exists an anti-Monge matrix $M\in \R^{(n+m+1)\times (n+m+1)}$ such that,
	for every $i\in [0\dd k)$ and $j\in [0\dd \ell)$, we have
	\begin{align*}
		M[i,j] &= \LIS(S\cap [x_i\dd x'_j)\times [y_i\dd y'_j)) &\text{if }x_i \le x'_j\text{ and }y_i \le y'_j,\\
		M[i,j] & \le -N & \text{otherwise.}
	\end{align*}
	Moreover, after $\Oh((n+m)\sqrt{\log(n+m)} + |S|\log^2 |S|)$-time preprocessing,
	any entry of $M$ can be computed in $\Oh(\log(n+m)/ \log \log (n+m))$ time.
\end{corollary}
\begin{proof}
	Let us extend $(x_i,y_i)_{i=0}^{k-1}$ to a cut-path $\pi$ and $(x'_j,y'_j)_{j=0}^{\ell-1}$
	to a cut-path $\pi'$ using \cref{fct:anti}. 
	Next, define a cut-paths $\tpi$ and $\tpi'$ so that, for every $d\in [0\dd n+m]$:
	\begin{align*}
		\tpi_d = \pi_d\text{ and }\tpi'_d = \pi'_d &\text{ if }\pi_d \preceq \pi'_d,\\
		\tpi_d = \pi'_d\text{ and }\tpi'_d = \pi_d &\text{ if }\pi'_d \preceq \pi_d.
	\end{align*}
	Note that this guarantees $\tpi\preceq \tpi'$.
	Moreover, define a graph $G=\SAG^{n,m}(\tpi,\tpi',S)$ and consider the matrix $L$ of \cref{lem:L}.
	Let $\tilde{M}$ be a submatrix of $L$ defined so that $\tilde{M}[i,j] = L[x_i-y_i+m,x'_j-y'_j+m]$
	for $i\in [0\dd k)$ and $j \in [0\dd \ell)$.
	Furthermore, let $M$ be obtained from $\tilde{M}$ by subtracting $N+|S|$ from any row $i$
	such that $\tpi_{x_i-y_i+m}\ne (x_i,y_i)$ and subtracting $N+|S|$ from any column $j$
	such that $\tpi'_{x'_j-y'_j+m}\ne (x'_j,y'_j)$.
	
	The sequences $(x_i-y_i+m)_{i=0}^{k-1}$ and $(x'_j-y'_j+m)_{j=0}^{\ell-1}$ are strictly increasing, 
	so $\tilde{M}$ is an anti-Monge matrix. 
	Moreover, $M^\square = \tilde{M}^\square$, so $M$ is also an anti-Monge matrix.
	Furthermore, the entries of $L$ can be computed in $\Oh(\log(n+m)/\log\log(n+m))$ time after $\Oh((n+m)\sqrt{\log(n+m)}+|S|\log^2 |S|)$-time preprocessing, so the same is true about the entries of $M$.
	
	It remains to prove that each value $M[i,j]$ satisfies the desired properties.
	Let $d = x_i-y_i+m$ and $d' = x'_j-y'_j+m$. 
	First, suppose that $\tpi_d = (x_i,y_i)$ and  $\tpi'_d = (x'_j,y'_j)$,
	in which case $M[i,j]=\tilde{M}[i,j] = L[d,d']$.
	By \cref{lem:L}, the we have $L[d,d']=\LIS(S\cap [x_i\dd x'_j)\times [y_i\dd y'_j))$ if $x_i \le x'_j$ and $y_i \le y'_j$, and $L[d,d']=-N$ otherwise. Hence, the $M[i,j]$ satisfies the claim in this case.
	
	Next, suppose that $\tpi_d \ne (x_i,y_i)$.
	In this case, we have $(x_i,y_i) = \pi_d \succ \pi'_d$. 
	Consequently, if $x_i \le x'_j$ and $y_i \le y'_j$, then $\pi'_{d'}=(x'_j,y'_j) \succ \pi'_d$,
	which contradicts $\pi'$ being an antichain.
	Thus, it suffices to prove that $M[i,j] \le -N$.
	However, we have $M[i,j] \le \tilde{M}[i,j]-N-|S| =  L[d,d']-N-|S| \le |S|-N-|S|\le -N$, as claimed.
	
	Symmetrically, if $\tpi'_{d'}\ne (x'_j,y'_j)$,
	then $(x'_j,y'_j) = \pi'_{d'} \prec \pi_{d'}$.
	If $x_i \le x'_j$ and $y_i \le y'_j$, then $\pi_d  = (x_i,y_i) \prec \pi_{d'}$, which contradicts $\pi$ being an antichain. At the same time,  $M[i,j] \le \tilde{M}[i,j]-N-|S|\le L[d,d']-N-|S|$ holds as claimed.
\end{proof}

Finally, we derive \cref{thm:oracle}, whose statement is repeated below for reader's convenience.
\thmoracle*
\begin{proof}
	Without loss of generality, we may assume that the values $a_i$ belong to $[0\dd n)$.
	If this is not the case, we can construct a set $A = \{a_i : i\in [0\dd n)\}$ and replace
	each value $a_i$ with its rank $\rank_A(a_i)$. This transformation preserves the relative order
	between any two values $a_i$ and $a_j$, so it preserves monotonicity of subsequences.
	
	We construct a set $S = \{(i,a_i) : i\in [0\dd n)\} \sub \Gr^{n-1,n-1}$, observing
	that increasing subsequences of $(a_i)_{i=0}^{n-1}$ correspond to chains in $S$.
	We also define two sequences of points in $\Gr^{n,n}$: 
	$(x_i,y_i)_{i=0}^{k-1} = (p_i+1,a_{p_i}+1)_{i=0}^{k-1}$ and
	$(x'_j,y'_j)_{j=0}^{\ell-1} = (q_j,a_{q_j})_{j=0}^{\ell-1}$.
	Both are antichains because $(a_{p_i})_{i=0}^{k-1}$ and $(a_{q_j})_{j=0}^{\ell-1}$ are non-increasing.
	
	Now, let $M'$ be the matrix of \cref{cor:M} constructed for $N' = N+2$,
	and let $M$ be obtained from $M'$ by setting $M[i,j]=M'[i,j]+2$ for each $i\in [0\dd k)$ and $j\in [0\dd \ell)$.
	Since $M'$ is an anti-Monge matrix, so is $M$.
	Furthermore, due to $|S|=n$, any entry of $M$ can be computed in $\Oh(\log n / \log \log n)$
	time after $\Oh(n\log^2 n)$-time preprocessing.
	
	Thus, it remains to prove that each value of $M[i,j]$ satisfies the desired properties.
	If $p_i < q_j$ and $a_{p_i}<a_{q_j}$, then $x_i = p_i+1 \le q_j = x'_j$
	and $y_i = a_{p_i}+1 \le a_{q_j} = y'_j$,
	so $M[i,j]=2+M'[i,j] = 2+\LIS(S\cap [p_i+1\dd q_j)\times (a_{p_i}+1\dd a_{q_j}))$.
	The latter value is equal to the length of the longest chain starting from $(p_i,a_{p_i})$
	to $(q_j,a_{q_j})$, i.e.., the longest increasing subsequence from $a_{p_i}$ to $a_{q_j}$.
	On the other hand, if $p_i \ge q_j$, then $x_i = p_i+1 > q_j = x'_j$,
	and if $a_{p_i}\ge a_{q_j}$, then $y_i = a_{p_i} + 1 > a_{q_j} = y'_j$.
	In either case, $M[i,j]=2+M'[i,j] \le 2 - N' = -N$ holds as claimed.
\end{proof}
	\newpage
	\section{Improved Approximation Algorithms for \textsf{LIS}}\label{sec:dynamic}
In this section, we present applications of extended grid packing. We begin by stating our key lemma for extended grid packing in Section~\ref{sec:newex}. We then bring a use case of the extended grid packing technique for a non-dynamic problem in Section~\ref{sec:example}. This makes it clear how the new technique can be used to approximate \textsf{LIS}. We then bring a more detailed discussion as to why extended grid packing leads to a dynamic algorithm for \textsf{LIS}. While at a high-level, both our dynamic and non-dynamic algorithms make use of extended grid packing in a similar way, the dynamic algorithm requires additional considerations to ensure the update time remains bounded in the worst case.

\subsection{Extended Grid Packing}\label{sec:newex}
As explained earlier, grid packing is a tool for approximating \textsf{LIS}. For completeness, we first state the definitions. In this problem, we have a table of size $m \times m$. Our goal is to introduce a number of segments on the table. Each segment either covers a consecutive set of cells in a row or in a column. A segment $A$ \textit{precedes} a segment $B$ if \textbf{every} cell of $A$ is strictly higher than every cell of $B$ and also \textbf{every} cell of $A$ is strictly to the right of every cell of $B$. Two segments are \textit{non-conflicting}, if one of them precedes the other one. Otherwise, we call them \textit{conflicting}.  The segments we introduce can overlap and there is no restriction on the number of segments or the length of each segment. However, we would like to minimize the maximum number of segments that cover each cell. 

\begin{figure}[ht]

\centering

\tikzset{every picture/.style={line width=0.75pt}} 

\begin{tikzpicture}[x=0.75pt,y=0.75pt,yscale=-0.7,xscale=0.7]

\draw   (181,11) -- (450,11) -- (450,281) -- (181,281) -- cycle ;
\draw    (301,10) -- (301,281) ;

\draw    (331,10) -- (331,282) ;

\draw    (361,10) -- (361,280) ;

\draw    (211,10) -- (211,281) ;

\draw    (241,10) -- (241,281) ;

\draw    (271,10) -- (271,281) ;

\draw    (391,10) -- (391,281) ;

\draw    (421,10) -- (421,280) ;

\draw    (182,41) -- (450,41) ;

\draw    (182,71) -- (450,71) ;

\draw    (182,101) -- (450,101) ;

\draw    (182,131) -- (450,131) ;

\draw    (182,161) -- (450,161) ;

\draw    (182,191) -- (450,191) ;

\draw    (182,221) -- (450,221) ;

\draw    (182,251) -- (450,251) ;

\draw [color={rgb, 255:red, 208; green, 2; blue, 27 }  ,draw opacity=1 ][line width=3.75]    (256,114) -- (256,267) ;

\draw [color={rgb, 255:red, 74; green, 144; blue, 226 }  ,draw opacity=1 ][line width=3.75]    (319,29) -- (441,29) ;

\draw [color={rgb, 255:red, 65; green, 117; blue, 5 }  ,draw opacity=1 ][line width=3.75]    (196,203) -- (196,271) ;

\draw  [color={rgb, 255:red, 0; green, 0; blue, 0 }  ,draw opacity=1 ][fill={rgb, 255:red, 248; green, 231; blue, 28 }  ,fill opacity=0.48 ] (211,41) -- (241,41) -- (241,71) -- (211,71) -- cycle ;
\draw [line width=3.75]    (226,53) -- (226,181) ;

\draw  [color={rgb, 255:red, 0; green, 0; blue, 0 }  ,draw opacity=1 ][fill={rgb, 255:red, 74; green, 144; blue, 226 }  ,fill opacity=0.35 ] (271,41) -- (301,41) -- (301,71) -- (271,71) -- cycle ;
\draw [color={rgb, 255:red, 245; green, 166; blue, 35 }  ,draw opacity=1 ][line width=3.75]    (287,57) -- (409,57) ;

\end{tikzpicture}
\caption{Segments are shown on the grid. The pair (black, orange) is conflicting since the yellow cell (covered by the black segment) is on the same row as the blue cell (covered by the orange segment). The following pairs are non-conflicting: (green, black), (green, orange), (green, blue), (red, orange), (red, blue), (black, blue).} \label{fig:crossing}
\end{figure}
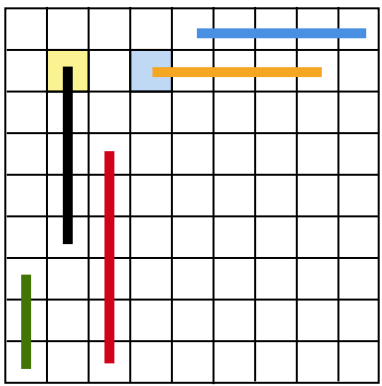

After we choose the segments, an adversary puts a non-negative number on each cell of the table. The score of a subset of cells of the table would be the sum of their values and the overall score of the table is the maximum score of a path of length $2m-1$ from the bottom-left corner to the top-right corner. In such a path, we always either move up or to the right.

The score of a segment is the sum of the numbers on the cells it covers. We obtain the maximum sum of the scores of a non-conflicting set of segments.  The score of the table is an upper bound on the score of any set of non-conflicting segments. We would like to choose segments so that the ratio of the score of the table and our score is bounded by a constant, no matter how the adversary puts the numbers on the table. More precisely, we call a solution $(\alpha,\beta)$-approximate, if at most $\alpha$ segments cover each cell and it guarantees a $1/\beta$ fraction of the score of the table for us for any assignment of numbers to the table cells.

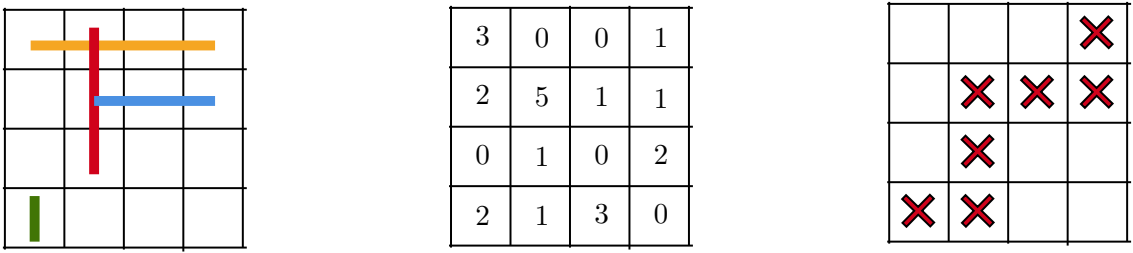
\begin{figure*}

\centering

\tikzset{every picture/.style={line width=0.75pt}} 

\begin{tikzpicture}[x=0.75pt,y=0.75pt,yscale=-1,xscale=1]

\draw    (270,30) -- (270,152) ;

\draw    (300,31) -- (300,151) ;

\draw    (330,30) -- (330,152) ;

\draw    (360,30) -- (360,153) ;

\draw    (390,30) -- (390,151) ;

\draw    (269,31) -- (392,31) ;

\draw    (269,61) -- (392,61) ;

\draw    (269,91) -- (392,91) ;

\draw    (269,121) -- (392,121) ;

\draw    (269,151) -- (392,151) ;

\draw    (45,31) -- (45,153) ;

\draw    (75,32) -- (75,152) ;

\draw    (105,31) -- (105,153) ;

\draw    (135,31) -- (135,154) ;

\draw    (165,31) -- (165,152) ;

\draw    (44,32) -- (167,32) ;

\draw    (44,62) -- (167,62) ;

\draw    (44,92) -- (167,92) ;

\draw    (44,122) -- (167,122) ;

\draw    (44,152) -- (167,152) ;

\draw    (491,28) -- (491,150) ;

\draw    (521,29) -- (521,149) ;

\draw    (551,28) -- (551,150) ;

\draw    (581,28) -- (581,151) ;

\draw    (611,28) -- (611,149) ;

\draw    (490,29) -- (613,29) ;

\draw    (490,59) -- (613,59) ;

\draw    (490,89) -- (613,89) ;

\draw    (490,119) -- (613,119) ;

\draw    (490,149) -- (613,149) ;

\draw [color={rgb, 255:red, 245; green, 166; blue, 35 }  ,draw opacity=1 ][line width=3.75]    (58,50) -- (151,50) ;

\draw [color={rgb, 255:red, 208; green, 2; blue, 27 }  ,draw opacity=1 ][line width=3.75]    (90,41) -- (90,115) ;

\draw [color={rgb, 255:red, 74; green, 144; blue, 226 }  ,draw opacity=1 ][line width=3.75]    (90,78) -- (151,78) ;

\draw [color={rgb, 255:red, 65; green, 117; blue, 5 }  ,draw opacity=1 ][line width=3.75]    (60,126) -- (60,149) ;

\draw  [color={rgb, 255:red, 0; green, 0; blue, 0 }  ,draw opacity=1 ][fill={rgb, 255:red, 208; green, 2; blue, 27 }  ,fill opacity=1 ] (497.54,127.53) -- (499.49,125.52) -- (505.58,131.45) -- (511.51,125.36) -- (513.51,127.31) -- (507.58,133.4) -- (513.67,139.33) -- (511.72,141.33) -- (505.63,135.41) -- (499.7,141.49) -- (497.7,139.54) -- (503.63,133.45) -- cycle ;
\draw  [color={rgb, 255:red, 0; green, 0; blue, 0 }  ,draw opacity=1 ][fill={rgb, 255:red, 208; green, 2; blue, 27 }  ,fill opacity=1 ] (527.54,127.53) -- (529.49,125.52) -- (535.58,131.45) -- (541.51,125.36) -- (543.51,127.31) -- (537.58,133.4) -- (543.67,139.33) -- (541.72,141.33) -- (535.63,135.41) -- (529.7,141.49) -- (527.7,139.54) -- (533.63,133.45) -- cycle ;
\draw  [color={rgb, 255:red, 0; green, 0; blue, 0 }  ,draw opacity=1 ][fill={rgb, 255:red, 208; green, 2; blue, 27 }  ,fill opacity=1 ] (527.54,97.53) -- (529.49,95.52) -- (535.58,101.45) -- (541.51,95.36) -- (543.51,97.31) -- (537.58,103.4) -- (543.67,109.33) -- (541.72,111.33) -- (535.63,105.41) -- (529.7,111.49) -- (527.7,109.54) -- (533.63,103.45) -- cycle ;
\draw  [color={rgb, 255:red, 0; green, 0; blue, 0 }  ,draw opacity=1 ][fill={rgb, 255:red, 208; green, 2; blue, 27 }  ,fill opacity=1 ] (527.54,67.53) -- (529.49,65.52) -- (535.58,71.45) -- (541.51,65.36) -- (543.51,67.31) -- (537.58,73.4) -- (543.67,79.33) -- (541.72,81.33) -- (535.63,75.41) -- (529.7,81.49) -- (527.7,79.54) -- (533.63,73.45) -- cycle ;
\draw  [color={rgb, 255:red, 0; green, 0; blue, 0 }  ,draw opacity=1 ][fill={rgb, 255:red, 208; green, 2; blue, 27 }  ,fill opacity=1 ] (557.54,67.53) -- (559.49,65.52) -- (565.58,71.45) -- (571.51,65.36) -- (573.51,67.31) -- (567.58,73.4) -- (573.67,79.33) -- (571.72,81.33) -- (565.63,75.41) -- (559.7,81.49) -- (557.7,79.54) -- (563.63,73.45) -- cycle ;
\draw  [color={rgb, 255:red, 0; green, 0; blue, 0 }  ,draw opacity=1 ][fill={rgb, 255:red, 208; green, 2; blue, 27 }  ,fill opacity=1 ] (587.54,67.53) -- (589.49,65.52) -- (595.58,71.45) -- (601.51,65.36) -- (603.51,67.31) -- (597.58,73.4) -- (603.67,79.33) -- (601.72,81.33) -- (595.63,75.41) -- (589.7,81.49) -- (587.7,79.54) -- (593.63,73.45) -- cycle ;
\draw  [color={rgb, 255:red, 0; green, 0; blue, 0 }  ,draw opacity=1 ][fill={rgb, 255:red, 208; green, 2; blue, 27 }  ,fill opacity=1 ] (587.54,37.53) -- (589.49,35.52) -- (595.58,41.45) -- (601.51,35.36) -- (603.51,37.31) -- (597.58,43.4) -- (603.67,49.33) -- (601.72,51.33) -- (595.63,45.41) -- (589.7,51.49) -- (587.7,49.54) -- (593.63,43.45) -- cycle ;

\draw (286,136) node   {$2$};
\draw (286,105) node   {$0$};
\draw (286,75) node   {$2$};
\draw (286,45) node   {$3$};
\draw (316,136) node   {$1$};
\draw (316,106) node   {$1$};
\draw (316,76) node   {$5$};
\draw (316,46) node   {$0$};
\draw (346,76) node   {$1$};
\draw (346,105) node   {$0$};
\draw (346,135) node   {$3$};
\draw (376,135) node   {$0$};
\draw (376,105) node   {$2$};
\draw (376,77) node   {$1$};
\draw (346,46) node   {$0$};
\draw (376,46) node   {$1$};

\end{tikzpicture}

\caption{After we introduce the segments (left figure), the adversary puts the numbers on the table (middle figure). In this case, the score of the table is equal to $12$ (via the path depicted on the right figure), and our score is equal to $9$ obtained from two non-conflicting segments green and blue.} \label{fig:crossing}
\end{figure*}

Mitzenmacher and Seddighin~\cite{our-stoc-paper} prove the following theorem: For any $m \times m$ table and any $0 < \kappa < 1$, there exists a grid packing solution with  guarantee $(O_{\kappa}(m^\kappa \log m),O(1/\kappa))$. That is, each cell is covered by at most $O_{\kappa}(m^\kappa \log m)$ segments and the ratio of the table's score over our score is bounded by $O(1/\kappa)$ in the worst case. 

The general framework of grid packing remains the same for our extension: The problem can be thought of as a game played on an $m \times m$ table against an adversary and the goal is to introduce some multisegments (a generalization of segments explained below) such that after the adversary puts her numbers on the table cells, the score we obtain is comparable to table's score. However, extended grid packing differs with grid packing in two ways: First, we introduce a new notion that we call a \textit{multisegment} and we allow the use of multisegments instead of segments. Second, we do not enforce any bound on the number of multisegments that cover each cell. That is, we only have one objective which is maximizing the ratio of our score over the score of the table. Without the bound, utilizing extended grid packing for \textsf{LIS} becomes harder as previous solutions require a cap on the maximum number of segments covering each cell. However, we present in Section~\ref{sec:example} a more clever application of extended grid packing that does not depend on this bound.

Before we introduce multisegments, let us give an example to illustrate why segments fall short of our purpose which is obtaining a $(1-\epsilon)$ fraction of the table's score. For an $m \times m$ table, there are $m\binom{m}{2}+m^2$ distinct horizontal segments and $m\binom{m}{2}+m^2$ distinct vertical segments that amount to $2m\binom{m}{2} + m^2$ segments in total (there are $m^2$ single cell segments that can be regarded as both vertical and horizontal). Figure~\ref{fig:zeros} gives an example that proves we cannot obtain a score more than $2/3$ of the score of the table. In this example, even if we introduce all possible $2m\binom{m}{2} + m^2$ segments, from every three consecutive cells with value $1$ no more than two can be covered by non-conflicting segments. Thus, we cannot obtain more than $2/3$ of the score of the table even if there is no restriction on the maximum number of segments covering each cell.

\input{figs/zeros}

The example of Figure~\ref{fig:zeros} highlights  the fact that even if all possible segments can be used in a solution, there is no hope to obtain a $1-\epsilon$ fraction of the score of the table. This is the main motivation behind the definition of multisegments. As we show later, multisegments enable us to obtain a score arbitrarily close to the score of the table.

For a horizontal/vertical segment, we define its \textit{first} cell as its leftmost/bottommost cell and its \textit{last} cell as the rightmost/topmost cell of the segment. A $\Delta$-multisegment is defined as a combination of $\Delta$ segments $s_1, s_2, \ldots, s_{\Delta}$ where for each $1 \leq i \leq \Delta-1$, the last cell of segment $s_i$ is the same as the first cell of segment $s_{i+1}$. By definition, the notion of $1$-multisegment collides with that of segment. We say a multisegment covers a cell, if any of its segments covers that cell. Moreover, two multisegments $S_1$ and $S_2$ are non-conflicting, if for each segment $x$ of $S_1$ and each segment $y$ of $S_2$, $x$ and $y$ are non-conflicting. To avoid confusion, we use uppercase letters for multisegments and lowercase letters for segments. Based on this definition, for any $1 \leq i < \Delta$, an $i$-multisegment is also a $\Delta$-multisegment (we may add $\Delta-i$ single cell segments to an $i$-multisegment to make it compatible with the definition of $\Delta$-multisegment without any change in its shape).
\begin{figure}[ht]

\centering

\tikzset{every picture/.style={line width=0.75pt}} 

\begin{tikzpicture}[x=0.75pt,y=0.75pt,yscale=-0.5,xscale=.5]

\draw   (141.5,15) -- (539.5,15) -- (539.5,411) -- (141.5,411) -- cycle ;
\draw    (189.5,16) -- (189.5,412) ;
\draw    (239.5,16) -- (239.5,410) ;
\draw    (289.5,16) -- (289.5,410) ;
\draw    (339.5,16) -- (339.5,410) ;
\draw    (389.5,16) -- (389.5,412) ;
\draw    (439.5,16) -- (439.5,411) ;
\draw    (489.5,16) -- (489.5,411) ;
\draw    (539.5,66) -- (142.5,66) ;
\draw    (540.5,116) -- (142.5,116) ;
\draw    (540.5,166) -- (142.5,166) ;
\draw    (539.5,216) -- (142.5,216) ;
\draw    (538.5,266) -- (142.5,266) ;
\draw    (539.5,316) -- (142.5,316) ;
\draw    (539.5,364) -- (142.5,364) ;
\draw  [draw opacity=0][fill={rgb, 255:red, 248; green, 231; blue, 28 }  ,fill opacity=1 ] (157.48,404.5) -- (157.52,334.5) -- (170.52,334.5) -- (170.48,404.5) -- cycle ;
\draw  [draw opacity=0][fill={rgb, 255:red, 248; green, 231; blue, 28 }  ,fill opacity=1 ] (157.54,334.81) -- (273.5,334) -- (273.57,344.69) -- (157.62,345.5) -- cycle ;
\draw  [draw opacity=0][fill={rgb, 255:red, 248; green, 231; blue, 28 }  ,fill opacity=1 ] (260.48,343) -- (260.52,276.99) -- (273.5,277) -- (273.45,343.01) -- cycle ;
\draw  [draw opacity=0][fill={rgb, 255:red, 189; green, 16; blue, 224 }  ,fill opacity=0.4 ] (360.48,303) -- (360.52,236.99) -- (373.5,237) -- (373.45,303.01) -- cycle ;
\draw  [draw opacity=0][fill={rgb, 255:red, 189; green, 16; blue, 224 }  ,fill opacity=0.4 ] (371.51,236.94) -- (471.5,236) -- (471.61,248.04) -- (371.62,248.98) -- cycle ;
\draw  [draw opacity=0][fill={rgb, 255:red, 189; green, 16; blue, 224 }  ,fill opacity=0.4 ] (459.48,236) -- (459.52,178.99) -- (471.5,179) -- (471.46,236.01) -- cycle ;
\draw  [draw opacity=0][fill={rgb, 255:red, 189; green, 16; blue, 224 }  ,fill opacity=0.4 ] (469.51,178.94) -- (524.39,178.43) -- (524.5,190) -- (469.62,190.52) -- cycle ;
\draw  [draw opacity=0][fill={rgb, 255:red, 0; green, 0; blue, 0 }  ,fill opacity=0.31 ] (307.54,131.81) -- (423.5,131) -- (423.57,141.69) -- (307.62,142.5) -- cycle ;
\draw  [draw opacity=0][fill={rgb, 255:red, 0; green, 0; blue, 0 }  ,fill opacity=0.31 ] (413.42,43.77) -- (523.49,43) -- (523.57,54.23) -- (413.49,55) -- cycle ;
\draw  [draw opacity=0][fill={rgb, 255:red, 0; green, 0; blue, 0 }  ,fill opacity=0.31 ] (423.5,53) -- (423.51,133.04) -- (413.51,133.04) -- (413.49,53) -- cycle ;
\draw  [draw opacity=0][fill={rgb, 255:red, 184; green, 233; blue, 134 }  ,fill opacity=1 ] (210.48,253) -- (210.52,186.99) -- (223.5,187) -- (223.45,253.01) -- cycle ;
\draw  [draw opacity=0][fill={rgb, 255:red, 184; green, 233; blue, 134 }  ,fill opacity=1 ] (219.1,199.68) -- (161.66,200.43) -- (161.5,188) -- (218.93,187.25) -- cycle ;
\draw  [draw opacity=0][fill={rgb, 255:red, 184; green, 233; blue, 134 }  ,fill opacity=1 ] (160.48,200) -- (160.52,133.99) -- (173.5,134) -- (173.45,200.01) -- cycle ;

\end{tikzpicture}

\caption{All polylines except for the green one are valid multisegments. Yellow and gray multisegments are non-conflicting, while the rest of the multisegment pairs are conflicting.} \label{fig:milti}
\end{figure}
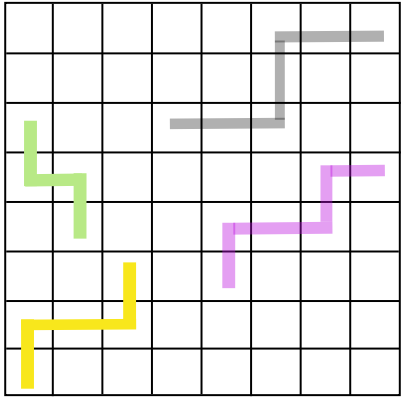

We define an extended version of the grid packing problem as a game between us and an adversary. Similar to grid packing, we first introduce a number of multisegments and then the adversary puts nonnegative numbers on the cells of the table. Then, table's score is formulated as the largest sum the adversary can collect from the values of the cells by moving from the bottom-left corner to the top-right corner of the table. Our score is the largest sum we can collect by non-conflicting multisegments where the value of a multisegment is equal to the total sum of the numbers of the cells it covers.

As we show in Lemma~\ref{lemma:multi}, if for a $\Delta$, we consider all possible $\Delta$-multisegments in our solution, our score is always at least a $\frac{\Delta-1}{\Delta}$ fraction of the table's score. Notice that by introducing all such multisegments, a cell may be covered by many multisegments.

\begin{lemma}\label{lemma:multi}
	For a fixed $1 \leq \Delta$, if we introduce all $\Delta$-multisegments in the extended grid packing problem, our score will be least a $\frac{\Delta-1}{\Delta}$ fraction of the table's score regardless of the values of the table cells.
\end{lemma}
\begin{proof}
	Let us fix the optimal path of length $2m-1$ from the bottom-left cell of the table to the top-right cell which gives the highest score to the table. For the sake of this proof, we give an ordering to the cells of this path. More precisely, let $c_1, c_2, \ldots, c_{2m-1}$ be the sequence of these cells where $c_1$ is the bottom-left corners and $c_{2m-1}$ is the top-right corner and as we increase $i$, the distance of $c_i$ from $c_1$ increases. We call a cell $c_i$ \textit{critical}, if $1 < i < 2m-1$ and cells $\{c_{i-1},c_i, c_{i+1}\}$ cannot be covered by a single segment (i.e., they neither lie on the same row nor lie on the same column). In addition to this, we consider both $c_1$ and $c_{2m-1}$ to be critical cells. The proof is based on the following observations:
	\begin{figure}[ht]

\centering

\tikzset{every picture/.style={line width=0.75pt}} 

\begin{tikzpicture}[x=0.75pt,y=0.75pt,yscale=-0.6,xscale=0.6]

\draw   (141.5,15) -- (539.5,15) -- (539.5,411) -- (141.5,411) -- cycle ;
\draw    (189.5,16) -- (189.5,412) ;
\draw    (239.5,16) -- (239.5,410) ;
\draw    (289.5,16) -- (289.5,410) ;
\draw    (339.5,16) -- (339.5,410) ;
\draw    (389.5,16) -- (389.5,412) ;
\draw    (439.5,16) -- (439.5,411) ;
\draw    (489.5,16) -- (489.5,411) ;
\draw    (539.5,66) -- (142.5,66) ;
\draw    (540.5,116) -- (142.5,116) ;
\draw    (540.5,166) -- (142.5,166) ;
\draw    (539.5,216) -- (142.5,216) ;
\draw    (538.5,266) -- (142.5,266) ;
\draw    (539.5,316) -- (142.5,316) ;
\draw    (539.5,364) -- (142.5,364) ;
\draw  [color={rgb, 255:red, 0; green, 0; blue, 0 }  ,draw opacity=1 ][fill={rgb, 255:red, 80; green, 227; blue, 194 }  ,fill opacity=0.22 ] (141.5,216) -- (189.5,216) -- (189.5,266) -- (141.5,266) -- cycle ;
\draw  [color={rgb, 255:red, 0; green, 0; blue, 0 }  ,draw opacity=1 ][fill={rgb, 255:red, 80; green, 227; blue, 194 }  ,fill opacity=0.22 ] (339.5,116) -- (389.5,116) -- (389.5,166) -- (339.5,166) -- cycle ;
\draw  [color={rgb, 255:red, 0; green, 0; blue, 0 }  ,draw opacity=1 ][fill={rgb, 255:red, 80; green, 227; blue, 194 }  ,fill opacity=0.22 ] (339.5,216) -- (389.5,216) -- (389.5,266) -- (339.5,266) -- cycle ;
\draw  [color={rgb, 255:red, 0; green, 0; blue, 0 }  ,draw opacity=1 ][fill={rgb, 255:red, 80; green, 227; blue, 194 }  ,fill opacity=0.22 ] (389.5,116) -- (439.5,116) -- (439.5,166) -- (389.5,166) -- cycle ;
\draw  [color={rgb, 255:red, 0; green, 0; blue, 0 }  ,draw opacity=1 ][fill={rgb, 255:red, 80; green, 227; blue, 194 }  ,fill opacity=0.22 ] (389.5,66) -- (439.5,66) -- (439.5,116) -- (389.5,116) -- cycle ;
\draw  [color={rgb, 255:red, 0; green, 0; blue, 0 }  ,draw opacity=1 ][fill={rgb, 255:red, 80; green, 227; blue, 194 }  ,fill opacity=0.22 ] (489.5,66) -- (539.5,66) -- (539.5,116) -- (489.5,116) -- cycle ;
\draw  [color={rgb, 255:red, 0; green, 0; blue, 0 }  ,draw opacity=1 ][fill={rgb, 255:red, 80; green, 227; blue, 194 }  ,fill opacity=0.22 ] (141.5,364) -- (189.5,364) -- (189.5,411) -- (141.5,411) -- cycle ;
\draw  [color={rgb, 255:red, 0; green, 0; blue, 0 }  ,draw opacity=1 ][fill={rgb, 255:red, 80; green, 227; blue, 194 }  ,fill opacity=0.22 ] (489.5,15) -- (539.5,15) -- (539.5,66) -- (489.5,66) -- cycle ;

\draw (155,380.4) node [anchor=north west][inner sep=0.75pt]    {$\star $};
\draw (306,231.4) node [anchor=north west][inner sep=0.75pt]    {$\star $};
\draw (355,231.4) node [anchor=north west][inner sep=0.75pt]    {$\star $};
\draw (354,181.4) node [anchor=north west][inner sep=0.75pt]    {$\star $};
\draw (403,131.4) node [anchor=north west][inner sep=0.75pt]    {$\star $};
\draw (456,80.4) node [anchor=north west][inner sep=0.75pt]    {$\star $};
\draw (506,79.4) node [anchor=north west][inner sep=0.75pt]    {$\star $};
\draw (506,29.4) node [anchor=north west][inner sep=0.75pt]    {$\star $};
\draw (355,130.4) node [anchor=north west][inner sep=0.75pt]    {$\star $};
\draw (402,80.4) node [anchor=north west][inner sep=0.75pt]    {$\star $};
\draw (156,230.4) node [anchor=north west][inner sep=0.75pt]    {$\star $};
\draw (207,230.4) node [anchor=north west][inner sep=0.75pt]    {$\star $};
\draw (256,230.4) node [anchor=north west][inner sep=0.75pt]    {$\star $};
\draw (156,280.4) node [anchor=north west][inner sep=0.75pt]    {$\star $};
\draw (155,330.4) node [anchor=north west][inner sep=0.75pt]    {$\star $};

\end{tikzpicture}

\caption{An example of the optimal path is shown by starred cells. In this example, all the critical cells are colored in blue.} \label{fig:stars}
\end{figure}
	\begin{itemize}
		\item  For some $1 \leq i \leq j \leq 2m-1$ all cells $c_i, c_{i+1},\ldots, c_{j-1},c_j$ can be covered by a single $\Delta$-multisegment if no more than $\Delta-1$ cells in this set are critical. 
		\item If cell $c_i$ is critical, then all cells $c_{i+1},c_{i+2},\ldots,c_{2m-1}$ are higher and to the right of all cells $c_1, c_2, \ldots, c_{i-1}$.
	\end{itemize}
	To complete the proof, we index the critical cells by their ordering. In other words, let $r_1, r_2, \ldots, r_{k}$ be the critical cells of the table in the order they are listed in the sequence $c_1, \ldots, c_{2m-1}$ (here $k$ is the number of critical cells). It follows from the above observations that there are some non-conflicting $\Delta$-multisegements that cover all cells of the path except for $r_1, r_{\Delta+1}, r_{2\Delta+1}, r_{3\Delta+1},\ldots$. More generally, for each $1 \leq j \leq \Delta$, all cells of the optimal path except for $r_{j}, r_{\Delta+j}, r_{2\Delta+j}, r_{3\Delta+j},\ldots$ can be covered by non-conflicting $\Delta$-multisegments. Since for at least one $j$, such cells contribute to no more than a $1/\Delta$ fraction of the score of the table, we can obtain at least a $\frac{\Delta-1}{\Delta}$ fraction of the table's score.
\end{proof}

Lemma~\ref{lemma:multi} alone does not suffice to improve the dynamic algorithm of Mitzenmacher and Seddighin~\cite{our-stoc-paper} for \textsf{LIS}  since their algorithm relies on a bound on the number of segments that cover each cell. We remedy this issue by presenting an alternative algorithm that does not require this bound. We explain the new algorithm in details in Section~\ref{sec:realdynamic}.

\subsection{A Data Structure for \textsf{LIS}}\label{sec:example}
Previous work~\cite{our-stoc-paper,our-soda-paper} illustrate natural connections between \textsf{LIS} and grid packing. Generally, we represent the input sequence by points on the plane in a way that a point with coordinates $(x,y)$ indicates that the $x$'s element of the sequence has value $y$. Thus, instead of a sequence of length $n$, we have a plane with $n$ points. This enables us to construct a grid, where the rows and the columns evenly divide the points. Assuming the number the adversary puts on each cell is the contribution of the points inside that cell to the \textsf{LIS} of the sequence (which may be different from the \textsf{LIS} of the points included in that cell), the score of the table (which we try to be competitive to in extended grid packing) is always equal to the size of the optimal solution. Moreover, the notion of non-conflicting segments/multisegments gives us a way to construct a global solution by combining partial solutions. In other words, if we associate a partial solution to each segment/multisegment that only incorporates the points covered by those segments/multisegments, then for a set of non-conflicting segments/multisegments, a combination of partial solutions forms a valid increasing subsequence.

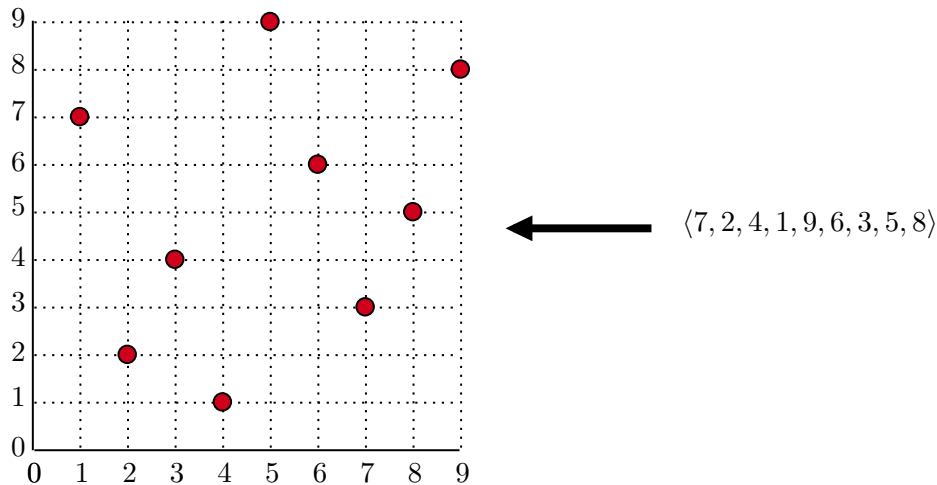
\begin{figure*}[ht]

\centering

\tikzset{every picture/.style={line width=0.75pt}} 

\begin{tikzpicture}[x=0.75pt,y=0.75pt,yscale=-0.8,xscale=0.8]

\draw  [dash pattern={on 0.84pt off 2.51pt}]  (147,49) -- (147,320) ;

\draw  [dash pattern={on 0.84pt off 2.51pt}]  (177,49) -- (177,321) ;

\draw  [dash pattern={on 0.84pt off 2.51pt}]  (207,49) -- (207,319) ;

\draw  [dash pattern={on 0.84pt off 2.51pt}]  (57,49) -- (57,320) ;

\draw  [dash pattern={on 0.84pt off 2.51pt}]  (87,49) -- (87,320) ;

\draw  [dash pattern={on 0.84pt off 2.51pt}]  (117,49) -- (117,320) ;

\draw  [dash pattern={on 0.84pt off 2.51pt}]  (237,49) -- (237,320) ;

\draw  [dash pattern={on 0.84pt off 2.51pt}]  (267,49) -- (267,319) ;

\draw  [dash pattern={on 0.84pt off 2.51pt}]  (28,80) -- (296,80) ;

\draw  [dash pattern={on 0.84pt off 2.51pt}]  (28,110) -- (296,110) ;

\draw  [dash pattern={on 0.84pt off 2.51pt}]  (28,140) -- (296,140) ;

\draw  [dash pattern={on 0.84pt off 2.51pt}]  (28,170) -- (296,170) ;

\draw  [dash pattern={on 0.84pt off 2.51pt}]  (28,200) -- (296,200) ;

\draw  [dash pattern={on 0.84pt off 2.51pt}]  (28,230) -- (296,230) ;

\draw  [dash pattern={on 0.84pt off 2.51pt}]  (28,260) -- (296,260) ;

\draw  [dash pattern={on 0.84pt off 2.51pt}]  (28,290) -- (296,290) ;

\draw    (27,49) -- (27,320) ;

\draw  [dash pattern={on 0.84pt off 2.51pt}]  (297,49) -- (297,319) ;

\draw    (28,320) -- (296,320) ;

\draw  [dash pattern={on 0.84pt off 2.51pt}]  (28,50) -- (296,50) ;

\draw  [color={rgb, 255:red, 0; green, 0; blue, 0 }  ,draw opacity=1 ][fill={rgb, 255:red, 208; green, 2; blue, 27 }  ,fill opacity=1 ] (51,109.5) .. controls (51,106.46) and (53.46,104) .. (56.5,104) .. controls (59.54,104) and (62,106.46) .. (62,109.5) .. controls (62,112.54) and (59.54,115) .. (56.5,115) .. controls (53.46,115) and (51,112.54) .. (51,109.5) -- cycle ;
\draw  [color={rgb, 255:red, 0; green, 0; blue, 0 }  ,draw opacity=1 ][fill={rgb, 255:red, 208; green, 2; blue, 27 }  ,fill opacity=1 ] (81,259.5) .. controls (81,256.46) and (83.46,254) .. (86.5,254) .. controls (89.54,254) and (92,256.46) .. (92,259.5) .. controls (92,262.54) and (89.54,265) .. (86.5,265) .. controls (83.46,265) and (81,262.54) .. (81,259.5) -- cycle ;
\draw  [color={rgb, 255:red, 0; green, 0; blue, 0 }  ,draw opacity=1 ][fill={rgb, 255:red, 208; green, 2; blue, 27 }  ,fill opacity=1 ] (111,199.5) .. controls (111,196.46) and (113.46,194) .. (116.5,194) .. controls (119.54,194) and (122,196.46) .. (122,199.5) .. controls (122,202.54) and (119.54,205) .. (116.5,205) .. controls (113.46,205) and (111,202.54) .. (111,199.5) -- cycle ;
\draw  [color={rgb, 255:red, 0; green, 0; blue, 0 }  ,draw opacity=1 ][fill={rgb, 255:red, 208; green, 2; blue, 27 }  ,fill opacity=1 ] (141,289.5) .. controls (141,286.46) and (143.46,284) .. (146.5,284) .. controls (149.54,284) and (152,286.46) .. (152,289.5) .. controls (152,292.54) and (149.54,295) .. (146.5,295) .. controls (143.46,295) and (141,292.54) .. (141,289.5) -- cycle ;
\draw  [color={rgb, 255:red, 0; green, 0; blue, 0 }  ,draw opacity=1 ][fill={rgb, 255:red, 208; green, 2; blue, 27 }  ,fill opacity=1 ] (171,49.5) .. controls (171,46.46) and (173.46,44) .. (176.5,44) .. controls (179.54,44) and (182,46.46) .. (182,49.5) .. controls (182,52.54) and (179.54,55) .. (176.5,55) .. controls (173.46,55) and (171,52.54) .. (171,49.5) -- cycle ;
\draw  [color={rgb, 255:red, 0; green, 0; blue, 0 }  ,draw opacity=1 ][fill={rgb, 255:red, 208; green, 2; blue, 27 }  ,fill opacity=1 ] (201,139.5) .. controls (201,136.46) and (203.46,134) .. (206.5,134) .. controls (209.54,134) and (212,136.46) .. (212,139.5) .. controls (212,142.54) and (209.54,145) .. (206.5,145) .. controls (203.46,145) and (201,142.54) .. (201,139.5) -- cycle ;
\draw  [color={rgb, 255:red, 0; green, 0; blue, 0 }  ,draw opacity=1 ][fill={rgb, 255:red, 208; green, 2; blue, 27 }  ,fill opacity=1 ] (231,229.5) .. controls (231,226.46) and (233.46,224) .. (236.5,224) .. controls (239.54,224) and (242,226.46) .. (242,229.5) .. controls (242,232.54) and (239.54,235) .. (236.5,235) .. controls (233.46,235) and (231,232.54) .. (231,229.5) -- cycle ;
\draw  [color={rgb, 255:red, 0; green, 0; blue, 0 }  ,draw opacity=1 ][fill={rgb, 255:red, 208; green, 2; blue, 27 }  ,fill opacity=1 ] (261,169.5) .. controls (261,166.46) and (263.46,164) .. (266.5,164) .. controls (269.54,164) and (272,166.46) .. (272,169.5) .. controls (272,172.54) and (269.54,175) .. (266.5,175) .. controls (263.46,175) and (261,172.54) .. (261,169.5) -- cycle ;
\draw  [color={rgb, 255:red, 0; green, 0; blue, 0 }  ,draw opacity=1 ][fill={rgb, 255:red, 208; green, 2; blue, 27 }  ,fill opacity=1 ] (291,79.5) .. controls (291,76.46) and (293.46,74) .. (296.5,74) .. controls (299.54,74) and (302,76.46) .. (302,79.5) .. controls (302,82.54) and (299.54,85) .. (296.5,85) .. controls (293.46,85) and (291,82.54) .. (291,79.5) -- cycle ;
\draw [line width=3]    (417,180) -- (330,180) ;
\draw [shift={(325,180)}, rotate = 360] [fill={rgb, 255:red, 0; green, 0; blue, 0 }  ][line width=3]  [draw opacity=0] (16.97,-8.15) -- (0,0) -- (16.97,8.15) -- cycle    ;

\draw (-34,147) node   {$ \begin{array}{l}
	\end{array}$};
\draw (517,177) node   {$\langle 7,2,4,1,9,6,3,5,8\rangle$};
\draw (28,334) node   {$0$};
\draw (58,334) node   {$1$};
\draw (88,334) node   {$2$};
\draw (118,334) node   {$3$};
\draw (148,334) node   {$4$};
\draw (178,334) node   {$5$};
\draw (208,334) node   {$6$};
\draw (238,334) node   {$7$};
\draw (268,334) node   {$8$};
\draw (298,334) node   {$9$};
\draw (28,334) node   {$0$};
\draw (18,318) node   {$0$};
\draw (18,286) node   {$1$};
\draw (18,258) node   {$2$};
\draw (18,226) node   {$3$};
\draw (18,198) node   {$4$};
\draw (18,166) node   {$5$};
\draw (18,138) node   {$6$};
\draw (18,106) node   {$7$};
\draw (18,76) node   {$8$};
\draw (18,47) node   {$9$};

\end{tikzpicture}

\caption{An array $\langle 7, 2, 4, 1, 9, 6, 3, 5, 8\rangle$ is mapped to the 2D plane.} \label{fig:lis-grid}
\end{figure*}
\input{figs/lis-grid2}

We refer the reader to previous work~\cite{our-stoc-paper,our-soda-paper} for discussions on how to use grid packing for approximating \textsf{LIS}. Roughly speaking, after making a grid, they construct a partial solution for each segment that keeps an approximation to the \textsf{LIS} of the points covered by that segment. Thus, every time a change is made, their algorithm has to update the solution for all segments that cover the modified point. Therefore, previous techniques require a bound on the number of segments that cover each cell of the grid. In order to obtain a score arbitrarily close to the score of the table, we need to include a lot of multisegments in our solution for extended grid packing, many of which cover the same cells of the table. This makes it impossible for previous applications to use the new construction. In this work, we introduce a new application that does not require the bound. To illustrate the new technique, we utilize extended grid packing to approximate the size of \textsf{LIS} in a non-dynamic setting. We later bring a more involved description of the dynamic algorithm.

Assume that we are given $n$ points on the plane with distinct coordinates. We would like to make a preprocess on the points such that after that we can (approximately) answer the following queries in polylogarithmic time: Given a rectangle parallel to the axis lines, what is the \textsf{LIS} of the points included in the rectangle? Recall that the \textsf{LIS} of a set of points is equal to the size of the longest list of the points where the $x$ and $y$ coordinates increase as the index of the elements increase in the list. To avoid confusion, we assume that the $x$ and $y$ coordinates of the borders of the rectangles are different from the coordinates of the points. We refer to this problem as \textit{query-\textsf{LIS}}.

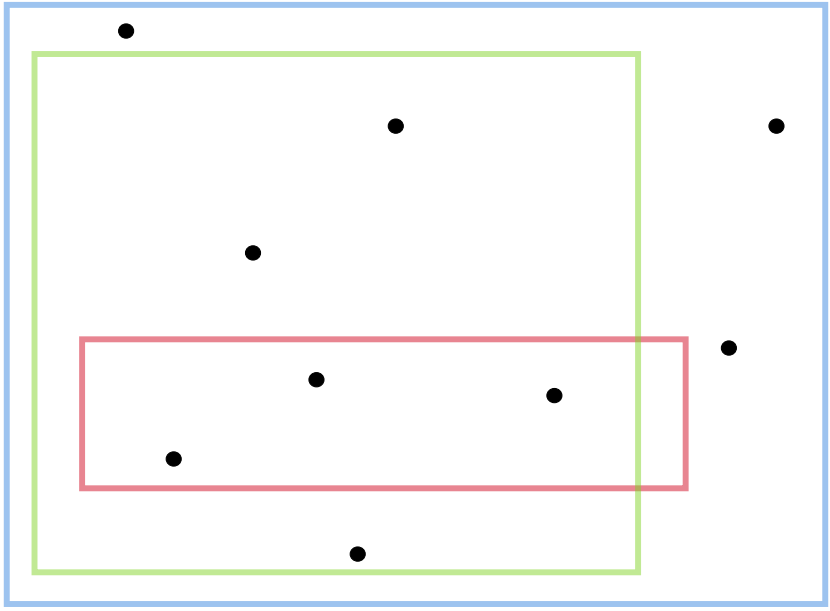
\begin{figure}

\centering

\begin{tikzpicture}[x=0.75pt,y=0.75pt,yscale=-0.8,xscale=0.8]

\draw  [fill={rgb, 255:red, 0; green, 0; blue, 0 }  ,fill opacity=1 ] (178,291.5) .. controls (178,289.01) and (180.13,287) .. (182.75,287) .. controls (185.37,287) and (187.5,289.01) .. (187.5,291.5) .. controls (187.5,293.99) and (185.37,296) .. (182.75,296) .. controls (180.13,296) and (178,293.99) .. (178,291.5) -- cycle ;
\draw  [fill={rgb, 255:red, 0; green, 0; blue, 0 }  ,fill opacity=1 ] (228,161.5) .. controls (228,159.01) and (230.13,157) .. (232.75,157) .. controls (235.37,157) and (237.5,159.01) .. (237.5,161.5) .. controls (237.5,163.99) and (235.37,166) .. (232.75,166) .. controls (230.13,166) and (228,163.99) .. (228,161.5) -- cycle ;
\draw  [fill={rgb, 255:red, 0; green, 0; blue, 0 }  ,fill opacity=1 ] (268,241.5) .. controls (268,239.01) and (270.13,237) .. (272.75,237) .. controls (275.37,237) and (277.5,239.01) .. (277.5,241.5) .. controls (277.5,243.99) and (275.37,246) .. (272.75,246) .. controls (270.13,246) and (268,243.99) .. (268,241.5) -- cycle ;
\draw  [fill={rgb, 255:red, 0; green, 0; blue, 0 }  ,fill opacity=1 ] (318,81.5) .. controls (318,79.01) and (320.13,77) .. (322.75,77) .. controls (325.37,77) and (327.5,79.01) .. (327.5,81.5) .. controls (327.5,83.99) and (325.37,86) .. (322.75,86) .. controls (320.13,86) and (318,83.99) .. (318,81.5) -- cycle ;
\draw  [fill={rgb, 255:red, 0; green, 0; blue, 0 }  ,fill opacity=1 ] (294,351.5) .. controls (294,349.01) and (296.13,347) .. (298.75,347) .. controls (301.37,347) and (303.5,349.01) .. (303.5,351.5) .. controls (303.5,353.99) and (301.37,356) .. (298.75,356) .. controls (296.13,356) and (294,353.99) .. (294,351.5) -- cycle ;
\draw  [fill={rgb, 255:red, 0; green, 0; blue, 0 }  ,fill opacity=1 ] (418,251.5) .. controls (418,249.01) and (420.13,247) .. (422.75,247) .. controls (425.37,247) and (427.5,249.01) .. (427.5,251.5) .. controls (427.5,253.99) and (425.37,256) .. (422.75,256) .. controls (420.13,256) and (418,253.99) .. (418,251.5) -- cycle ;
\draw  [fill={rgb, 255:red, 0; green, 0; blue, 0 }  ,fill opacity=1 ] (148,21.5) .. controls (148,19.01) and (150.13,17) .. (152.75,17) .. controls (155.37,17) and (157.5,19.01) .. (157.5,21.5) .. controls (157.5,23.99) and (155.37,26) .. (152.75,26) .. controls (150.13,26) and (148,23.99) .. (148,21.5) -- cycle ;
\draw  [fill={rgb, 255:red, 0; green, 0; blue, 0 }  ,fill opacity=1 ] (528,221.5) .. controls (528,219.01) and (530.13,217) .. (532.75,217) .. controls (535.37,217) and (537.5,219.01) .. (537.5,221.5) .. controls (537.5,223.99) and (535.37,226) .. (532.75,226) .. controls (530.13,226) and (528,223.99) .. (528,221.5) -- cycle ;
\draw  [fill={rgb, 255:red, 0; green, 0; blue, 0 }  ,fill opacity=1 ] (558,81.5) .. controls (558,79.01) and (560.13,77) .. (562.75,77) .. controls (565.37,77) and (567.5,79.01) .. (567.5,81.5) .. controls (567.5,83.99) and (565.37,86) .. (562.75,86) .. controls (560.13,86) and (558,83.99) .. (558,81.5) -- cycle ;
\draw  [color={rgb, 255:red, 208; green, 2; blue, 27 }  ,draw opacity=0.48 ][line width=2.25]  (125,216) -- (505.5,216) -- (505.5,310) -- (125,310) -- cycle ;
\draw  [color={rgb, 255:red, 126; green, 211; blue, 33 }  ,draw opacity=0.48 ][line width=2.25]  (95,36) -- (475.5,36) -- (475.5,363) -- (95,363) -- cycle ;
\draw  [color={rgb, 255:red, 74; green, 144; blue, 226 }  ,draw opacity=0.54 ][line width=2.25]  (77.5,5) -- (593.5,5) -- (593.5,383) -- (77.5,383) -- cycle ;

\end{tikzpicture}

\caption{The \textsf{LIS} of the points inside the red rectangle is equal to $2$. This value for the points inside the green rectangle is equal to $3$ and for the blue rectangle the \textsf{LIS} is equal to $4$.} \label{fig:rectangles}
\end{figure}

There is a straightforward solution for \textit{query-\textsf{LIS}} with preprocessing time $O(n^5 \log n)$ and query time $O(\log n)$. We first sort the $x$ and $y$ coordinates separately, and for every interval of the $x$ coordinates and $y$ coordinates we compute the \textsf{LIS} of the points within those intervals in time $O(n \log n)$. Since there are $O(n^2)$ intervals for $x$ coordinates and $O(n^2)$ intervals for $y$ coordinates, in total we compute the \textsf{LIS} of $O(n^4)$ subsets of the points and thus the overall runtime is $O(n^5 \log n)$. After this, whenever we are given a rectangle, we find the $x$ and $y$ intervals of the points covered by the rectangle and report the solution in time $O(\log n)$. This algorithm is very inefficient and can be easily improved in terms of preprocessing time, but for the sake of simplicity we do not discuss those improvements here.

Instead, we explain an application of extended grid packing that improves the preprocessing time of \textit{query-\textsf{LIS}} down to $\tilde O(n^{5/2})$ while keeping its query time polylogarithmic. This comes at the expense of losing a factor $1-\epsilon$ in the approximation. To this end, we define $\epsilon' = \epsilon /2$ and $\Delta = \lceil 1/\epsilon' \rceil$. Let $m = n^{5/8}$ be the table size in the extended grid packing problem and we construct the table in a way that the rows and columns evenly divide the points (in case the number of points is not divisible by $m$, we allow a difference of $+1/-1$ in the number of points covered by each row and column). After constructing the table, we define $2m$ subproblems each concerning the points that fall within a row or a column. That is, in each subproblem, we seek to solve \textit{query-\textsf{LIS}} for a subset of points that lie either in a row or in a column of the table. Thus, the size of each subproblem is $O(n^{3/8})$. We use the naive solution with preprocessing time of $O(n^5 \log n)$ for each subproblem. This amounts to a total preprocessing time of $$O(2m (n/m)^5 \log n) = O(n^5 \log n/m^4) = O(n^5 \log n / (n^{5/8})^4) = O(n^{5/2} \log n).$$

\begin{figure}

\centering

\includegraphics[width=11cm]{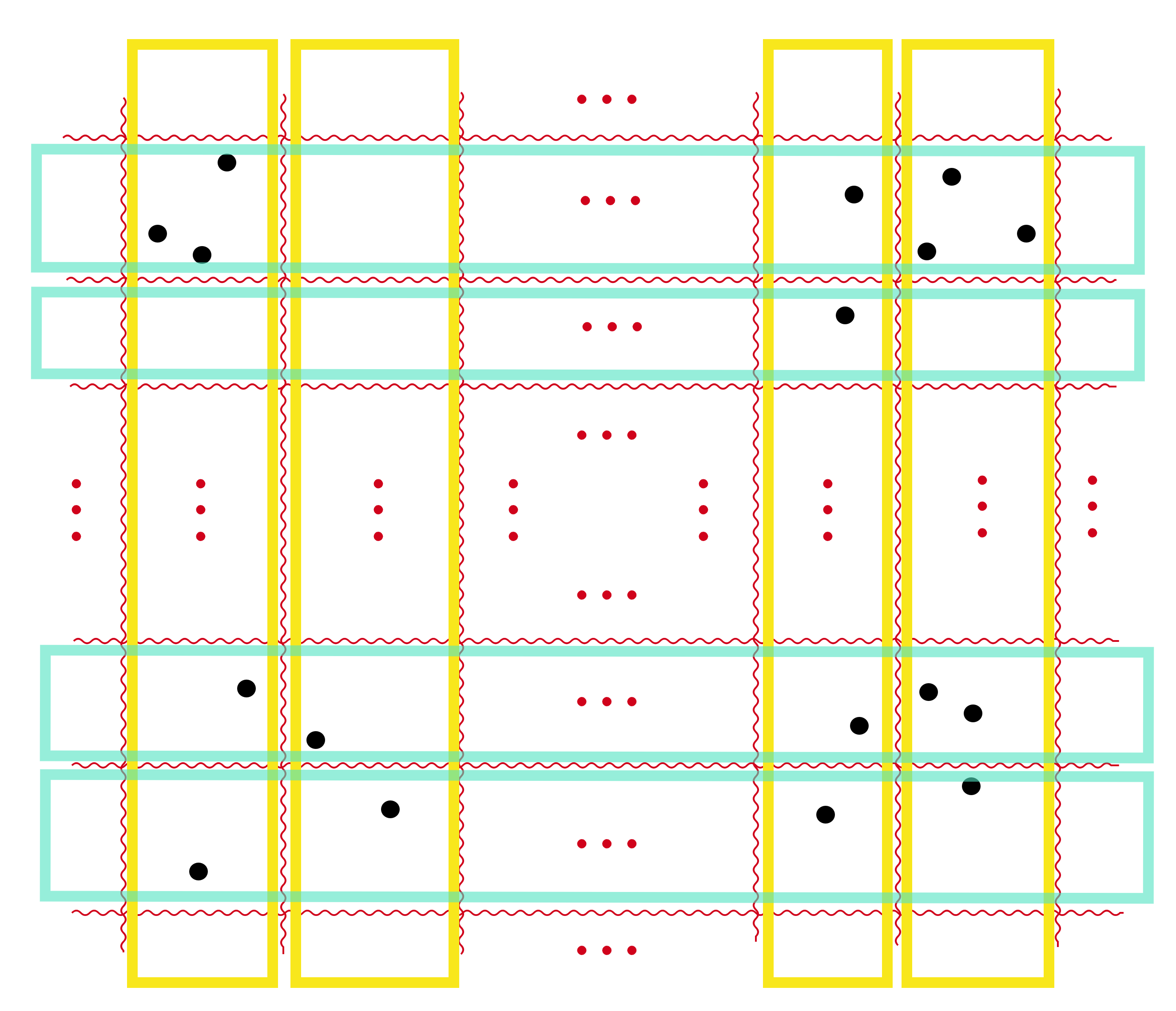}

\caption{The red snake-lines show the grid and yellow and blue boxes present the points included in each of the subproblems.} \label{fig:subproblems}
\end{figure}

Before answering the queries, we run another algorithm to preprocess an approximation to the \textsf{LIS} of some subsets of the points. More precisely, for every interval of rows and every interval of columns of the table, we preprocess the \textsf{LIS} of all the points that fall within the corresponding rectangle and store the computed values. In other words, for $(\binom{m}{2}+m)^2$ rectangles that can be formed by the rows and columns of the table, we approximate their \textsf{LIS}. In our algorithm, every time we fix the bottom-left corner of the rectangle, and for all possible top-right corners, we dynamically approximate the solution. Thus, in what follows, we fix a cell $c$ of the table and explain how we can approximate the value of \textsf{LIS} for all rectangles whose bottom-left corner is $c$  in time $\tilde O_{\epsilon}(m^2)$.

Our algorithm is based on a dynamic program. For each cell $c'$ which is not to the left of $c$ and is not below $c$, we define $f(c')$ as the  \textsf{LIS} of the points inside the rectangle formed by $c$ and $c'$ as opposite corners and define $g(c')$ that keeps an approximation of $f(c')$. We compute $g$ iteratively. Therefore, we start with $c' = c$ and move $c'$ to the right step by step. When we reach the end of the row, we start with the cell on top of $c$ and move it to the right until we reach the end of the row and then we start with two cell above $c$. We continue this procedure until we find a solution for the top-right corner as the last rectangle.

\begin{algorithm}[h!]
	\SetAlgoLined
	\KwResult{}
	construct an $m \times m$ grid whose rows and columns evenly divide the points\;
	Initialize the subproblems\;
	\For{$x_1 \leftarrow 1$ to $m$}{
		\For{$y_1 \leftarrow 1$ to $m$}{
			$c \leftarrow$ cell $(x_1,y_1)$\;
			\For{$x_2 \leftarrow x_1$ to $m$}{
				\For{$y_2 \leftarrow y_1$ to $m$}{
					$c' \leftarrow$ cell $(x_2,y_2)$\;
					\If {$x_1 = y_1$ or $x_2 = y_2$}{
						$g(c') \leftarrow $\textsf{LIS} of the rectangle $(x_1,y_1, x_2,y_2)$\;
					}\Else{
						$g(c') \leftarrow \max\{g(\text{cell}(x_2-1,y_2)),g(\text{cell}(x_2,y_2-1))\}$\;
						\For{$c'' \in X$}{
							\For{$S \in Y(c'')$}{
								$g(c') \leftarrow \max\{g(c'), g(c'')+\textsf{LIS}(S)\}$\;
							}
						}
					}
				}
			}
		}
	}
	\caption{Inefficient preprocessing}\label{alg:inefficient}
\end{algorithm}

In the DP, we use  extended grid packing to estimate the \textsf{LIS} for each rectangle. The base cases are when $c'$ and $c$ are either within the same row or within the same column in which case we can query the solution in time $O(\log n)$ in one of the subproblems. Otherwise we estimate the value of $f(c')$ based on an analysis that is inspired by extended grid packing. Suppose the extended grid packing is defined on our $m \times m$ table and in our solution for extended grid packing, we introduce all possible $\Delta$-multisegments. Moreover, assume that the numbers that the adversary puts on the table cells are the contributions of the corresponding cells to the \textsf{LIS} of the elements in the rectangle between $c$ and $c'$. By Lemma \ref{lemma:multi}, there is a set of non-conflicting $\Delta$-multisegments that gives us a score of at least $\frac{\Delta-1}{\Delta} f(c')$. Moreover, the \textsf{LIS} of the points covered by any multisegment is certainly an upper bound on the value of that multisegment. Let $X$ be the set of all cells in the rectangle between $c$ and $c'$ except for $c'$ and for any cell $c''$, let $Y(c'')$ be the set of all $\Delta$-multisegments that start from the cell to the top and right of $c''$ and end at cell $c'$. Moreover, for a multisegment $S$, we define $\textsf{LIS}(S)$ as the \textsf{LIS} of the points covered by $S$. For now, we introduce the following recursive formula for approximating $g(c')$. 
\begin{equation*}
g(c') := \max_{c'' \in X} \big[g(c'') + \max_{S \in Y(c'')} \textsf{LIS}(S)\big]
\end{equation*}
\begin{figure}[H]

\centering

\tikzset{every picture/.style={line width=0.75pt}} 

\begin{tikzpicture}[x=0.75pt,y=0.75pt,yscale=-0.8,xscale=0.8]

\draw   (241.5,15) -- (639.5,15) -- (639.5,411) -- (241.5,411) -- cycle ;
\draw    (289.5,16) -- (289.5,412) ;
\draw    (339.5,16) -- (339.5,410) ;
\draw    (389.5,16) -- (389.5,410) ;
\draw    (439.5,16) -- (439.5,410) ;
\draw    (489.5,16) -- (489.5,412) ;
\draw    (539.5,16) -- (539.5,411) ;
\draw    (589.5,16) -- (589.5,411) ;
\draw    (639.5,66) -- (242.5,66) ;
\draw    (640.5,116) -- (242.5,116) ;
\draw    (640.5,166) -- (242.5,166) ;
\draw    (639.5,216) -- (242.5,216) ;
\draw    (638.5,266) -- (242.5,266) ;
\draw    (639.5,316) -- (242.5,316) ;
\draw    (639.5,364) -- (242.5,364) ;
\draw  [color={rgb, 255:red, 0; green, 0; blue, 0 }  ,draw opacity=1 ][fill={rgb, 255:red, 245; green, 166; blue, 35 }  ,fill opacity=0.27 ] (289.5,266) -- (339.5,266) -- (339.5,316) -- (289.5,316) -- cycle ;
\draw  [color={rgb, 255:red, 0; green, 0; blue, 0 }  ,draw opacity=1 ][fill={rgb, 255:red, 245; green, 166; blue, 35 }  ,fill opacity=0.27 ] (539.5,15) -- (589.5,15) -- (589.5,66) -- (539.5,66) -- cycle ;
\draw  [draw opacity=0][fill={rgb, 255:red, 80; green, 227; blue, 194 }  ,fill opacity=0.37 ] (495.5,20) -- (584.5,20) -- (584.5,59) -- (495.5,59) -- cycle ;
\draw  [draw opacity=0][fill={rgb, 255:red, 80; green, 227; blue, 194 }  ,fill opacity=0.37 ] (445.5,121) -- (535.5,121) -- (535.5,160) -- (445.5,160) -- cycle ;
\draw  [draw opacity=0][fill={rgb, 255:red, 80; green, 227; blue, 194 }  ,fill opacity=0.37 ] (495.5,57) -- (534.5,57) -- (534.5,123) -- (495.5,123) -- cycle ;
\draw  [color={rgb, 255:red, 0; green, 0; blue, 0 }  ,draw opacity=1 ][fill={rgb, 255:red, 245; green, 166; blue, 35 }  ,fill opacity=0.27 ] (389.5,166) -- (439.5,166) -- (439.5,216) -- (389.5,216) -- cycle ;
\draw  [dash pattern={on 4.5pt off 4.5pt}] (300.5,171) -- (435.5,171) -- (435.5,308) -- (300.5,308) -- cycle ;
\draw  [dash pattern={on 4.5pt off 4.5pt}] (295.5,20) -- (584.5,20) -- (584.5,316) -- (295.5,316) -- cycle ;
\draw    (100.5,151) -- (303.52,173.67) ;
\draw [shift={(306.5,174)}, rotate = 186.37] [fill={rgb, 255:red, 0; green, 0; blue, 0 }  ][line width=0.08]  [draw opacity=0] (8.93,-4.29) -- (0,0) -- (8.93,4.29) -- cycle    ;
\draw    (36.5,125) -- (297.7,25.07) ;
\draw [shift={(300.5,24)}, rotate = 519.06] [fill={rgb, 255:red, 0; green, 0; blue, 0 }  ][line width=0.08]  [draw opacity=0] (8.93,-4.29) -- (0,0) -- (8.93,4.29) -- cycle    ;

\draw (308,281.4) node [anchor=north west][inner sep=0.75pt]  [font=\large]  {$c$};
\draw (556,29.4) node [anchor=north west][inner sep=0.75pt]  [font=\large]  {$c'$};
\draw (406,178.4) node [anchor=north west][inner sep=0.75pt]  [font=\large]  {$c''$};
\draw (506,80.4) node [anchor=north west][inner sep=0.75pt]  [font=\large]  {$S$};
\draw (16,126.4) node [anchor=north west][inner sep=0.75pt]    {$f( c') \ \geq \ f( c'') \ +\textsf{LIS}(S)$};

\end{tikzpicture}

\caption{The update process is explained in this figure} \label{fig:formula}
\end{figure}
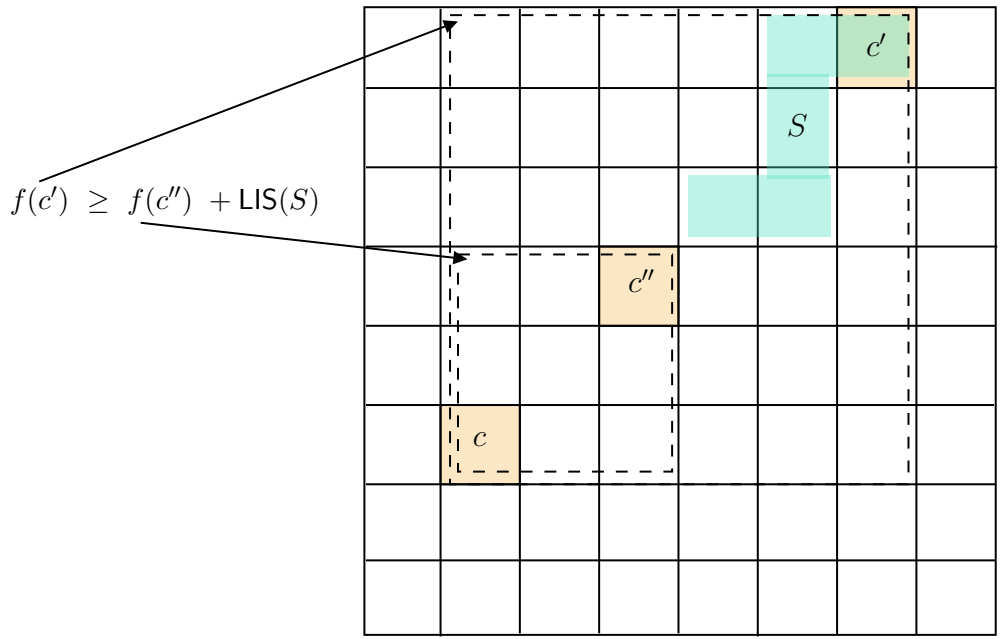
It follows from Lemma~\ref{lemma:multi} that if we formulate $g(c')$ as above, $\frac{\Delta-1}{\Delta} f(c') \leq g(c') \leq f(c')$ always holds. This gives us a formulation to recursively compute the value of $g$ for all cells. However, there are two issues to be resolved: (i) there are many possible cells $c''$ and multisegments $S$ that our algorithm needs to loop over and thus the runtime of the algorithm is not as desired. (ii) For a multisegment $S$, we do not have the value (or even an estimate) of $\textsf{LIS}(S)$. We show in the following that both issues can be resolved. 

We begin by considering the simpler case of $\Delta = 1$. Since in this case we are only concerned with segments, for each segment $s$ we can compute $\textsf{LIS}(s)$ by querying one of the subproblems (Recall that $s$ either completely fits in a row or in a column of the table). Thus, for each segment $s$, $\textsf{LIS}(s)$ is available in time $O(\log n)$. However, we still need to resolve the first issue since there may be up to $O(m)$ different segments that end at $c'$. The idea is to reduce the number of possible segments down to $O(\log n/\epsilon')$ by losing a factor of at most $1-\epsilon'$ in the approximation. The base cases are trivial as discussed previously, so for a cell $c'$ we begin by initializing $g(c')$ as the maximum value for its left and bottom cells. This ensures that our approximations are always monotone as should be. At a high-level, for every value $v \in \mathbb{D} = \{1,\lceil 1-\epsilon'\rceil , \lceil (1-\epsilon')^2 \rceil , \lceil (1-\epsilon')^3\rceil ,\ldots,n\}$ we only consider minimal (rightmost or topmost) segments that end at $c'$ and their \textsf{LIS} is at least $v$. There are at most $O(\log n/\epsilon')$ vertical and at most $O(\log n/\epsilon')$ such horizontal segments and each one can be found via a binary search in time $O(\log^2 n)$ (an $O(\log n)$ overhead for binary search and an $O(\log n)$ overhead for finding the \textsf{LIS} of a potential solution). Thus, if we only consider these segments, the runtime improves to $O(\log^3 n/\epsilon')$ for each pair of cells $(c,c'')$ which in total amounts to a runtime of $O(m^4 \log^3n /\epsilon')$.

\begin{figure*}[h!]

\centering

\tikzset{every picture/.style={line width=0.75pt}} 

\begin{tikzpicture}[x=0.75pt,y=0.75pt,yscale=-0.7,xscale=0.7]

\draw   (131.5,55) -- (529.5,55) -- (529.5,451) -- (131.5,451) -- cycle ;
\draw    (179.5,56) -- (179.5,452) ;
\draw    (229.5,56) -- (229.5,450) ;
\draw    (279.5,56) -- (279.5,450) ;
\draw    (329.5,56) -- (329.5,450) ;
\draw    (379.5,56) -- (379.5,452) ;
\draw    (429.5,56) -- (429.5,451) ;
\draw    (479.5,56) -- (479.5,451) ;
\draw    (529.5,106) -- (132.5,106) ;
\draw    (530.5,156) -- (132.5,156) ;
\draw    (530.5,206) -- (132.5,206) ;
\draw    (529.5,256) -- (132.5,256) ;
\draw    (528.5,306) -- (132.5,306) ;
\draw    (529.5,356) -- (132.5,356) ;
\draw    (529.5,404) -- (132.5,404) ;
\draw  [color={rgb, 255:red, 0; green, 0; blue, 0 }  ,draw opacity=1 ][fill={rgb, 255:red, 245; green, 166; blue, 35 }  ,fill opacity=0.27 ] (131.5,404) -- (179.5,404) -- (179.5,451) -- (131.5,451) -- cycle ;
\draw  [color={rgb, 255:red, 0; green, 0; blue, 0 }  ,draw opacity=1 ][fill={rgb, 255:red, 245; green, 166; blue, 35 }  ,fill opacity=0.27 ] (479.5,55) -- (529.5,55) -- (529.5,106) -- (479.5,106) -- cycle ;
\draw  [dash pattern={on 4.5pt off 4.5pt}]  (529.5,451) -- (529.5,495) ;
\draw  [dash pattern={on 4.5pt off 4.5pt}]  (479.5,451) -- (479.5,495) ;
\draw  [dash pattern={on 4.5pt off 4.5pt}]  (429.5,451) -- (429.5,495) ;
\draw  [dash pattern={on 4.5pt off 4.5pt}]  (379.5,451) -- (379.5,495) ;
\draw  [dash pattern={on 4.5pt off 4.5pt}]  (329.5,451) -- (329.5,495) ;
\draw  [dash pattern={on 4.5pt off 4.5pt}]  (279.5,451) -- (279.5,495) ;
\draw  [dash pattern={on 4.5pt off 4.5pt}]  (229.5,451) -- (229.5,495) ;
\draw  [dash pattern={on 4.5pt off 4.5pt}]  (179.5,451) -- (179.5,495) ;
\draw  [dash pattern={on 4.5pt off 4.5pt}]  (132.5,451) -- (132.5,495) ;
\draw  [dash pattern={on 4.5pt off 4.5pt}]  (529.5,11) -- (529.5,55) ;
\draw  [dash pattern={on 4.5pt off 4.5pt}]  (479.5,11) -- (479.5,55) ;
\draw  [dash pattern={on 4.5pt off 4.5pt}]  (429.5,11) -- (429.5,55) ;
\draw  [dash pattern={on 4.5pt off 4.5pt}]  (379.5,11) -- (379.5,55) ;
\draw  [dash pattern={on 4.5pt off 4.5pt}]  (329.5,11) -- (329.5,55) ;
\draw  [dash pattern={on 4.5pt off 4.5pt}]  (279.5,11) -- (279.5,55) ;
\draw  [dash pattern={on 4.5pt off 4.5pt}]  (229.5,11) -- (229.5,55) ;
\draw  [dash pattern={on 4.5pt off 4.5pt}]  (179.5,11) -- (179.5,55) ;
\draw  [dash pattern={on 4.5pt off 4.5pt}]  (132.5,11) -- (132.5,55) ;
\draw  [dash pattern={on 4.5pt off 4.5pt}]  (575.21,451.48) -- (531.21,451.52) ;
\draw  [dash pattern={on 4.5pt off 4.5pt}]  (575.16,403.48) -- (531.16,403.52) ;
\draw  [dash pattern={on 4.5pt off 4.5pt}]  (574.09,356.48) -- (530.09,356.52) ;
\draw  [dash pattern={on 4.5pt off 4.5pt}]  (575.04,305.48) -- (531.04,305.52) ;
\draw  [dash pattern={on 4.5pt off 4.5pt}]  (574.98,255.48) -- (530.98,255.52) ;
\draw  [dash pattern={on 4.5pt off 4.5pt}]  (574.93,205.48) -- (530.93,205.52) ;
\draw  [dash pattern={on 4.5pt off 4.5pt}]  (574.88,155.48) -- (530.88,155.52) ;
\draw  [dash pattern={on 4.5pt off 4.5pt}]  (574.82,106.48) -- (530.82,106.52) ;
\draw  [dash pattern={on 4.5pt off 4.5pt}]  (574.79,55.48) -- (530.79,55.52) ;
\draw  [dash pattern={on 4.5pt off 4.5pt}]  (131.21,451.48) -- (87.21,451.52) ;
\draw  [dash pattern={on 4.5pt off 4.5pt}]  (131.16,403.48) -- (87.16,403.52) ;
\draw  [dash pattern={on 4.5pt off 4.5pt}]  (130.09,356.48) -- (86.09,356.52) ;
\draw  [dash pattern={on 4.5pt off 4.5pt}]  (131.04,305.48) -- (87.04,305.52) ;
\draw  [dash pattern={on 4.5pt off 4.5pt}]  (130.98,255.48) -- (86.98,255.52) ;
\draw  [dash pattern={on 4.5pt off 4.5pt}]  (130.93,205.48) -- (86.93,205.52) ;
\draw  [dash pattern={on 4.5pt off 4.5pt}]  (130.88,155.48) -- (86.88,155.52) ;
\draw  [dash pattern={on 4.5pt off 4.5pt}]  (130.82,106.48) -- (86.82,106.52) ;
\draw  [dash pattern={on 4.5pt off 4.5pt}]  (130.79,55.48) -- (86.79,55.52) ;
\draw  [color={rgb, 255:red, 208; green, 2; blue, 27 }  ,draw opacity=1 ][fill={rgb, 255:red, 208; green, 2; blue, 27 }  ,fill opacity=0.2 ] (437.5,66) -- (525,66) -- (525,96) -- (437.5,96) -- cycle ;
\draw  [color={rgb, 255:red, 208; green, 2; blue, 27 }  ,draw opacity=1 ][fill={rgb, 255:red, 208; green, 2; blue, 27 }  ,fill opacity=0.2 ] (288.5,61) -- (525,61) -- (525,100) -- (288.5,100) -- cycle ;
\draw  [color={rgb, 255:red, 208; green, 2; blue, 27 }  ,draw opacity=1 ][fill={rgb, 255:red, 208; green, 2; blue, 27 }  ,fill opacity=0.2 ] (179.5,55) -- (530.5,55) -- (530.5,106) -- (179.5,106) -- cycle ;
\draw  [color={rgb, 255:red, 0; green, 0; blue, 0 }  ,draw opacity=1 ][fill={rgb, 255:red, 80; green, 227; blue, 194 }  ,fill opacity=0.2 ] (131.5,106) -- (179.5,106) -- (179.5,156) -- (131.5,156) -- cycle ;
\draw  [color={rgb, 255:red, 0; green, 0; blue, 0 }  ,draw opacity=1 ][fill={rgb, 255:red, 80; green, 227; blue, 194 }  ,fill opacity=0.2 ] (229.5,106) -- (279.5,106) -- (279.5,156) -- (229.5,156) -- cycle ;
\draw  [color={rgb, 255:red, 0; green, 0; blue, 0 }  ,draw opacity=1 ][fill={rgb, 255:red, 80; green, 227; blue, 194 }  ,fill opacity=0.2 ] (379.5,106) -- (429.5,106) -- (429.5,156) -- (379.5,156) -- cycle ;
\draw  [color={rgb, 255:red, 65; green, 117; blue, 5 }  ,draw opacity=1 ][fill={rgb, 255:red, 65; green, 117; blue, 5 }  ,fill opacity=0.2 ] (530.02,55.52) -- (528.99,404.08) -- (478.99,403.93) -- (480.02,55.37) -- cycle ;
\draw  [color={rgb, 255:red, 65; green, 117; blue, 5 }  ,draw opacity=1 ][fill={rgb, 255:red, 65; green, 117; blue, 5 }  ,fill opacity=0.2 ] (524.06,60.05) -- (523.49,252.08) -- (485.48,251.97) -- (486.05,59.94) -- cycle ;
\draw  [color={rgb, 255:red, 0; green, 0; blue, 0 }  ,draw opacity=1 ][fill={rgb, 255:red, 80; green, 227; blue, 194 }  ,fill opacity=0.2 ] (429.5,256) -- (479.5,256) -- (479.5,306) -- (429.5,306) -- cycle ;
\draw  [color={rgb, 255:red, 0; green, 0; blue, 0 }  ,draw opacity=1 ][fill={rgb, 255:red, 80; green, 227; blue, 194 }  ,fill opacity=0.2 ] (429.5,404) -- (479.5,404) -- (479.5,451) -- (429.5,451) -- cycle ;

\draw (148,416.4) node [anchor=north west][inner sep=0.75pt]  [font=\large]  {$c$};
\draw (496,66.4) node [anchor=north west][inner sep=0.75pt]  [font=\large]  {$c'$};

\end{tikzpicture}

\caption{An example is shown for computing the value of $g(c')$. Candidate horizontal segments are colored in red and candidate vertical segments are colored in green. When we use each of the candidate segments in our dynamic program, we update the solution based on the computed value for the corresponding blue cell. } \label{fig:dp-example}
\end{figure*}
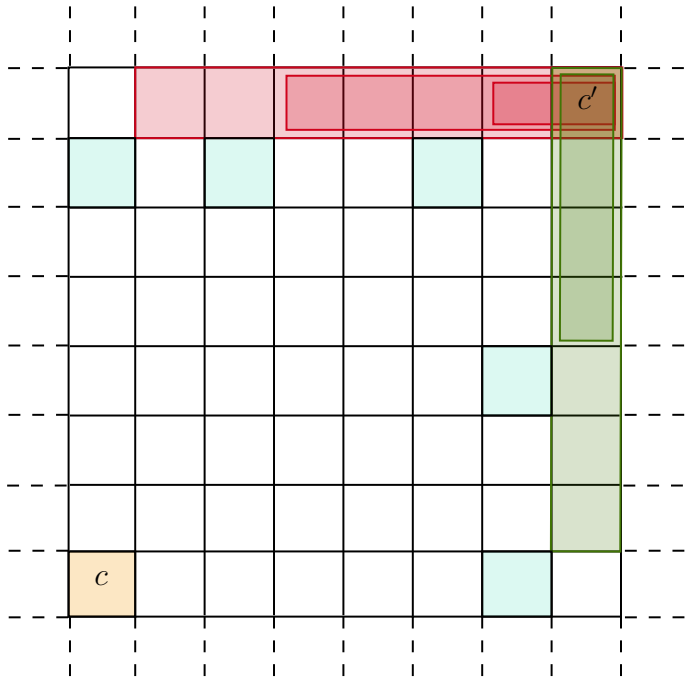

The correctness of this case is easy. Assume that the optimal segment that gives us the highest value for $g(c')$ in Algorithm~\ref{alg:inefficient} is a horizontal segment $s_1$, but we do not consider $s_1$ in the improved algorithm. Let $v$ be the largest value in sequence $\mathbb{D}$ which is not larger than the \textsf{LIS} of $s_1$. Let the corresponding horizontal segment for value $v$ in our algorithm be $s_2$. We define $c''_1, c''_2$ as the cells to the left and bottom of $s_1$ and $s_2$ respectively. It follows from our algorithm that $s_2$ is not larger than $s_1$ and thus the value we store for $g(c''_2)$ is at least as large as $g(c''_1)$. Moreover, we have $\textsf{LIS}(s''_2) \geq (1-\epsilon') \textsf{LIS}(s''_1)$ and in the update process we only lose an $\epsilon'$ fraction of the \textsf{LIS} of the last segment. The same analysis works for vertical segments as well. This implies that our estimation for $g$ loses a factor of at most $(1-\epsilon')$ throughout the DP. 

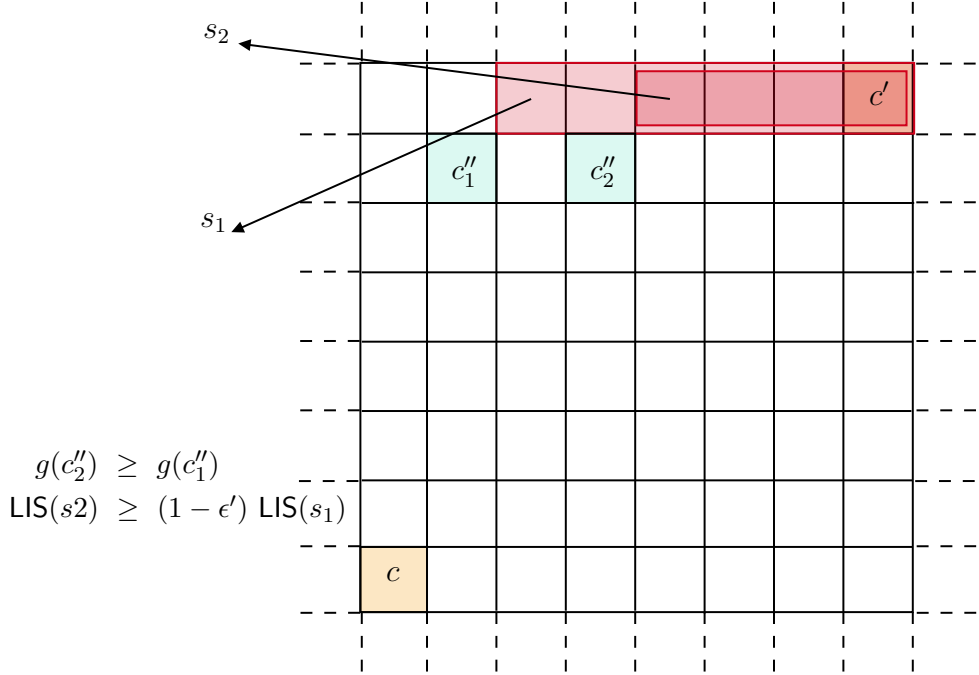
\begin{figure*}

\centering

\tikzset{every picture/.style={line width=0.75pt}} 

\begin{tikzpicture}[x=0.75pt,y=0.75pt,yscale=-0.7,xscale=0.7]

\draw   (211.5,55) -- (609.5,55) -- (609.5,451) -- (211.5,451) -- cycle ;
\draw    (259.5,56) -- (259.5,452) ;
\draw    (309.5,56) -- (309.5,450) ;
\draw    (359.5,56) -- (359.5,450) ;
\draw    (409.5,56) -- (409.5,450) ;
\draw    (459.5,56) -- (459.5,452) ;
\draw    (509.5,56) -- (509.5,451) ;
\draw    (559.5,56) -- (559.5,451) ;
\draw    (609.5,106) -- (212.5,106) ;
\draw    (610.5,156) -- (212.5,156) ;
\draw    (610.5,206) -- (212.5,206) ;
\draw    (609.5,256) -- (212.5,256) ;
\draw    (608.5,306) -- (212.5,306) ;
\draw    (609.5,356) -- (212.5,356) ;
\draw    (609.5,404) -- (212.5,404) ;
\draw  [color={rgb, 255:red, 0; green, 0; blue, 0 }  ,draw opacity=1 ][fill={rgb, 255:red, 245; green, 166; blue, 35 }  ,fill opacity=0.27 ] (211.5,404) -- (259.5,404) -- (259.5,451) -- (211.5,451) -- cycle ;
\draw  [color={rgb, 255:red, 0; green, 0; blue, 0 }  ,draw opacity=1 ][fill={rgb, 255:red, 245; green, 166; blue, 35 }  ,fill opacity=0.27 ] (559.5,55) -- (609.5,55) -- (609.5,106) -- (559.5,106) -- cycle ;
\draw  [dash pattern={on 4.5pt off 4.5pt}]  (609.5,451) -- (609.5,495) ;
\draw  [dash pattern={on 4.5pt off 4.5pt}]  (559.5,451) -- (559.5,495) ;
\draw  [dash pattern={on 4.5pt off 4.5pt}]  (509.5,451) -- (509.5,495) ;
\draw  [dash pattern={on 4.5pt off 4.5pt}]  (459.5,451) -- (459.5,495) ;
\draw  [dash pattern={on 4.5pt off 4.5pt}]  (409.5,451) -- (409.5,495) ;
\draw  [dash pattern={on 4.5pt off 4.5pt}]  (359.5,451) -- (359.5,495) ;
\draw  [dash pattern={on 4.5pt off 4.5pt}]  (309.5,451) -- (309.5,495) ;
\draw  [dash pattern={on 4.5pt off 4.5pt}]  (259.5,451) -- (259.5,495) ;
\draw  [dash pattern={on 4.5pt off 4.5pt}]  (212.5,451) -- (212.5,495) ;
\draw  [dash pattern={on 4.5pt off 4.5pt}]  (609.5,11) -- (609.5,55) ;
\draw  [dash pattern={on 4.5pt off 4.5pt}]  (559.5,11) -- (559.5,55) ;
\draw  [dash pattern={on 4.5pt off 4.5pt}]  (509.5,11) -- (509.5,55) ;
\draw  [dash pattern={on 4.5pt off 4.5pt}]  (459.5,11) -- (459.5,55) ;
\draw  [dash pattern={on 4.5pt off 4.5pt}]  (409.5,11) -- (409.5,55) ;
\draw  [dash pattern={on 4.5pt off 4.5pt}]  (359.5,11) -- (359.5,55) ;
\draw  [dash pattern={on 4.5pt off 4.5pt}]  (309.5,11) -- (309.5,55) ;
\draw  [dash pattern={on 4.5pt off 4.5pt}]  (259.5,11) -- (259.5,55) ;
\draw  [dash pattern={on 4.5pt off 4.5pt}]  (212.5,11) -- (212.5,55) ;
\draw  [dash pattern={on 4.5pt off 4.5pt}]  (655.21,451.48) -- (611.21,451.52) ;
\draw  [dash pattern={on 4.5pt off 4.5pt}]  (655.16,403.48) -- (611.16,403.52) ;
\draw  [dash pattern={on 4.5pt off 4.5pt}]  (654.09,356.48) -- (610.09,356.52) ;
\draw  [dash pattern={on 4.5pt off 4.5pt}]  (655.04,305.48) -- (611.04,305.52) ;
\draw  [dash pattern={on 4.5pt off 4.5pt}]  (654.98,255.48) -- (610.98,255.52) ;
\draw  [dash pattern={on 4.5pt off 4.5pt}]  (654.93,205.48) -- (610.93,205.52) ;
\draw  [dash pattern={on 4.5pt off 4.5pt}]  (654.88,155.48) -- (610.88,155.52) ;
\draw  [dash pattern={on 4.5pt off 4.5pt}]  (654.82,106.48) -- (610.82,106.52) ;
\draw  [dash pattern={on 4.5pt off 4.5pt}]  (654.79,55.48) -- (610.79,55.52) ;
\draw  [dash pattern={on 4.5pt off 4.5pt}]  (211.21,451.48) -- (167.21,451.52) ;
\draw  [dash pattern={on 4.5pt off 4.5pt}]  (211.16,403.48) -- (167.16,403.52) ;
\draw  [dash pattern={on 4.5pt off 4.5pt}]  (210.09,356.48) -- (166.09,356.52) ;
\draw  [dash pattern={on 4.5pt off 4.5pt}]  (211.04,305.48) -- (167.04,305.52) ;
\draw  [dash pattern={on 4.5pt off 4.5pt}]  (210.98,255.48) -- (166.98,255.52) ;
\draw  [dash pattern={on 4.5pt off 4.5pt}]  (210.93,205.48) -- (166.93,205.52) ;
\draw  [dash pattern={on 4.5pt off 4.5pt}]  (210.88,155.48) -- (166.88,155.52) ;
\draw  [dash pattern={on 4.5pt off 4.5pt}]  (210.82,106.48) -- (166.82,106.52) ;
\draw  [dash pattern={on 4.5pt off 4.5pt}]  (210.79,55.48) -- (166.79,55.52) ;
\draw  [color={rgb, 255:red, 208; green, 2; blue, 27 }  ,draw opacity=1 ][fill={rgb, 255:red, 208; green, 2; blue, 27 }  ,fill opacity=0.2 ] (410.5,61) -- (605,61) -- (605,100) -- (410.5,100) -- cycle ;
\draw  [color={rgb, 255:red, 208; green, 2; blue, 27 }  ,draw opacity=1 ][fill={rgb, 255:red, 208; green, 2; blue, 27 }  ,fill opacity=0.2 ] (309.5,55) -- (610.5,55) -- (610.5,106) -- (309.5,106) -- cycle ;
\draw  [color={rgb, 255:red, 0; green, 0; blue, 0 }  ,draw opacity=1 ][fill={rgb, 255:red, 80; green, 227; blue, 194 }  ,fill opacity=0.2 ] (259.5,106) -- (309.5,106) -- (309.5,156) -- (259.5,156) -- cycle ;
\draw  [color={rgb, 255:red, 0; green, 0; blue, 0 }  ,draw opacity=1 ][fill={rgb, 255:red, 80; green, 227; blue, 194 }  ,fill opacity=0.2 ] (359.5,106) -- (409.5,106) -- (409.5,156) -- (359.5,156) -- cycle ;
\draw    (334.5,81) -- (121.24,175.78) ;
\draw [shift={(118.5,177)}, rotate = 336.03999999999996] [fill={rgb, 255:red, 0; green, 0; blue, 0 }  ][line width=0.08]  [draw opacity=0] (8.93,-4.29) -- (0,0) -- (8.93,4.29) -- cycle    ;
\draw    (434.5,81) -- (126.48,41.38) ;
\draw [shift={(123.5,41)}, rotate = 367.33000000000004] [fill={rgb, 255:red, 0; green, 0; blue, 0 }  ][line width=0.08]  [draw opacity=0] (8.93,-4.29) -- (0,0) -- (8.93,4.29) -- cycle    ;

\draw (228,416.4) node [anchor=north west][inner sep=0.75pt]  [font=\large]  {$c$};
\draw (576,66.4) node [anchor=north west][inner sep=0.75pt]  [font=\large]  {$c'$};
\draw (275,120.4) node [anchor=north west][inner sep=0.75pt]    {$c''_{1}$};
\draw (375,120.4) node [anchor=north west][inner sep=0.75pt]    {$c''_{2}$};
\draw (94,161.4) node [anchor=north west][inner sep=0.75pt]  [font=\large]  {$s_{1}$};
\draw (96,25.4) node [anchor=north west][inner sep=0.75pt]  [font=\large]  {$s_{2}$};
\draw (-25,332.4) node [anchor=north west][inner sep=0.75pt]    {$g( c''_{2}) \ \geq \ g( c''_{1})$};
\draw (-44,364.4) node [anchor=north west][inner sep=0.75pt]    {$\textsf{LIS}( s2) \ \geq \ ( 1-\epsilon ') \ \textsf{LIS}( s_{1})$};

\end{tikzpicture}

\caption{If we use segment $s_2$ instead of segment $s_1$, we only lose a factor of $1-\epsilon'$ in the approximation.} \label{fig:dp-proof}
\end{figure*}

For $\Delta = 2$, we first consider all 1-multisegments (segments) as explained above and determine an initial value for $g(c')$. Then we proceed by a generalization of the above idea for $2$-multisegments. Let $\mathbb{D} = \{1,\lceil 1-\epsilon'\rceil , \lceil (1-\epsilon')^2 \rceil , \lceil (1-\epsilon')^3\rceil ,\ldots,n\}$. For any two values $v_1, v_2 \in \mathbb{D}$ we consider the following $2$-multisegment: We define a sweeping line which is initially equal to the right edge of cell $c'$. We move the sweeping line parallel to that edge to the left, until the \textsf{LIS} of the rectangle which is covered by the first sweeping line is at least $v_1$. Let $e$ be the cell that contains the sweeping line. Now, we define another sweeping line which is equal to the portion of the bottom edge of $e$ which is not covered by the first sweeping line. Starting from there, we move the second sweeping line parallel to that edge downward, until the \textsf{LIS} of the corresponding rectangle becomes at least $v_2$ (See Figure~\ref{fig:dp-example2} for a visualization of the two rectangles). We define 2-multisegment $S$ as the set of all the cells that intersect with either rectangles and estimate its \textsf{LIS} by  $v_1 + v_2$. Next, we determine cell $c''$ which is to the left and bottom of $S$ and update the value of $g(c')$ as  $g(c'):= \max\{g(c'), g(c'') + v_1 + v_2\}$. Similarly we repeat the same procedure starting with vertical and then horizontal rectangles. The total number of segments that we investigate via this algorithm is $2|D|^2 = O(\log^2 n/\epsilon'^2)$ and finding each segment takes time $O(\log^2 n)$. Thus, the update time becomes $O(\log^4 n/\epsilon'^2)$ for each pair of cells $(c,c')$ and $O(m^4 \log^4 n/\epsilon'^2)$ in total.

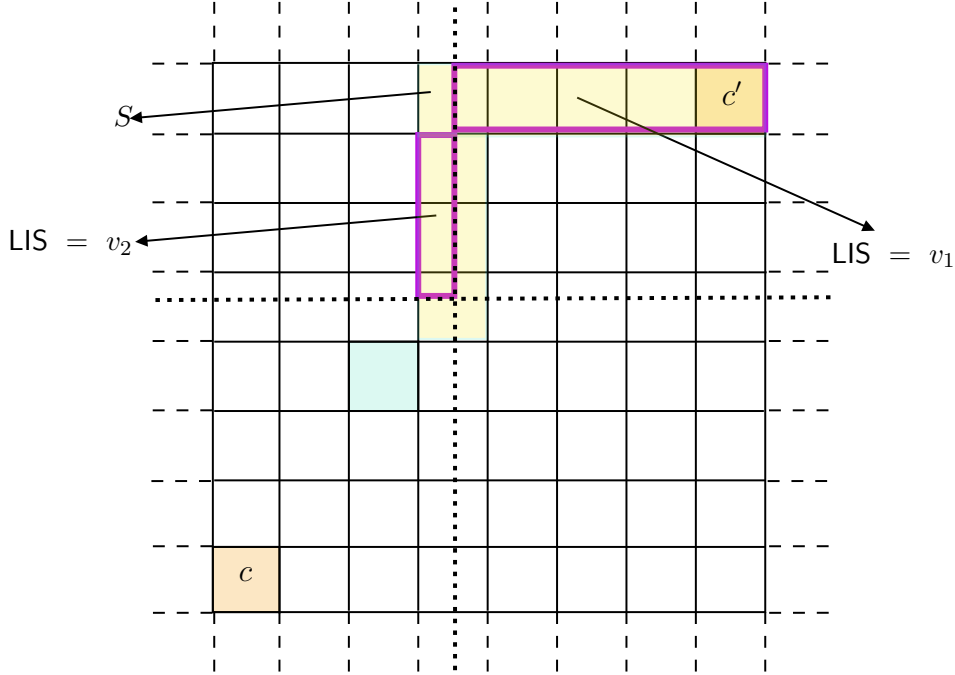
\begin{figure*}

\centering

\tikzset{every picture/.style={line width=0.75pt}} 

\begin{tikzpicture}[x=0.75pt,y=0.75pt,yscale=-0.7,xscale=0.7]

\draw   (131.5,55) -- (529.5,55) -- (529.5,451) -- (131.5,451) -- cycle ;
\draw    (179.5,56) -- (179.5,452) ;
\draw    (229.5,56) -- (229.5,450) ;
\draw    (279.5,56) -- (279.5,450) ;
\draw    (329.5,56) -- (329.5,450) ;
\draw    (379.5,56) -- (379.5,452) ;
\draw    (429.5,56) -- (429.5,451) ;
\draw    (479.5,56) -- (479.5,451) ;
\draw    (529.5,106) -- (132.5,106) ;
\draw    (530.5,156) -- (132.5,156) ;
\draw    (530.5,206) -- (132.5,206) ;
\draw    (529.5,256) -- (132.5,256) ;
\draw    (528.5,306) -- (132.5,306) ;
\draw    (529.5,356) -- (132.5,356) ;
\draw    (529.5,404) -- (132.5,404) ;
\draw  [color={rgb, 255:red, 0; green, 0; blue, 0 }  ,draw opacity=1 ][fill={rgb, 255:red, 245; green, 166; blue, 35 }  ,fill opacity=0.27 ] (131.5,404) -- (179.5,404) -- (179.5,451) -- (131.5,451) -- cycle ;
\draw  [color={rgb, 255:red, 0; green, 0; blue, 0 }  ,draw opacity=1 ][fill={rgb, 255:red, 245; green, 166; blue, 35 }  ,fill opacity=0.27 ] (479.5,55) -- (529.5,55) -- (529.5,106) -- (479.5,106) -- cycle ;
\draw  [dash pattern={on 4.5pt off 4.5pt}]  (529.5,451) -- (529.5,495) ;
\draw  [dash pattern={on 4.5pt off 4.5pt}]  (479.5,451) -- (479.5,495) ;
\draw  [dash pattern={on 4.5pt off 4.5pt}]  (429.5,451) -- (429.5,495) ;
\draw  [dash pattern={on 4.5pt off 4.5pt}]  (379.5,451) -- (379.5,495) ;
\draw  [dash pattern={on 4.5pt off 4.5pt}]  (329.5,451) -- (329.5,495) ;
\draw  [dash pattern={on 4.5pt off 4.5pt}]  (279.5,451) -- (279.5,495) ;
\draw  [dash pattern={on 4.5pt off 4.5pt}]  (229.5,451) -- (229.5,495) ;
\draw  [dash pattern={on 4.5pt off 4.5pt}]  (179.5,451) -- (179.5,495) ;
\draw  [dash pattern={on 4.5pt off 4.5pt}]  (132.5,451) -- (132.5,495) ;
\draw  [dash pattern={on 4.5pt off 4.5pt}]  (529.5,11) -- (529.5,55) ;
\draw  [dash pattern={on 4.5pt off 4.5pt}]  (479.5,11) -- (479.5,55) ;
\draw  [dash pattern={on 4.5pt off 4.5pt}]  (429.5,11) -- (429.5,55) ;
\draw  [dash pattern={on 4.5pt off 4.5pt}]  (379.5,11) -- (379.5,55) ;
\draw  [dash pattern={on 4.5pt off 4.5pt}]  (329.5,11) -- (329.5,55) ;
\draw  [dash pattern={on 4.5pt off 4.5pt}]  (279.5,11) -- (279.5,55) ;
\draw  [dash pattern={on 4.5pt off 4.5pt}]  (229.5,11) -- (229.5,55) ;
\draw  [dash pattern={on 4.5pt off 4.5pt}]  (179.5,11) -- (179.5,55) ;
\draw  [dash pattern={on 4.5pt off 4.5pt}]  (132.5,11) -- (132.5,55) ;
\draw  [dash pattern={on 4.5pt off 4.5pt}]  (575.21,451.48) -- (531.21,451.52) ;
\draw  [dash pattern={on 4.5pt off 4.5pt}]  (575.16,403.48) -- (531.16,403.52) ;
\draw  [dash pattern={on 4.5pt off 4.5pt}]  (574.09,356.48) -- (530.09,356.52) ;
\draw  [dash pattern={on 4.5pt off 4.5pt}]  (575.04,305.48) -- (531.04,305.52) ;
\draw  [dash pattern={on 4.5pt off 4.5pt}]  (574.98,255.48) -- (530.98,255.52) ;
\draw  [dash pattern={on 4.5pt off 4.5pt}]  (574.93,205.48) -- (530.93,205.52) ;
\draw  [dash pattern={on 4.5pt off 4.5pt}]  (574.88,155.48) -- (530.88,155.52) ;
\draw  [dash pattern={on 4.5pt off 4.5pt}]  (574.82,106.48) -- (530.82,106.52) ;
\draw  [dash pattern={on 4.5pt off 4.5pt}]  (574.79,55.48) -- (530.79,55.52) ;
\draw  [dash pattern={on 4.5pt off 4.5pt}]  (131.21,451.48) -- (87.21,451.52) ;
\draw  [dash pattern={on 4.5pt off 4.5pt}]  (131.16,403.48) -- (87.16,403.52) ;
\draw  [dash pattern={on 4.5pt off 4.5pt}]  (130.09,356.48) -- (86.09,356.52) ;
\draw  [dash pattern={on 4.5pt off 4.5pt}]  (131.04,305.48) -- (87.04,305.52) ;
\draw  [dash pattern={on 4.5pt off 4.5pt}]  (130.98,255.48) -- (86.98,255.52) ;
\draw  [dash pattern={on 4.5pt off 4.5pt}]  (130.93,205.48) -- (86.93,205.52) ;
\draw  [dash pattern={on 4.5pt off 4.5pt}]  (130.88,155.48) -- (86.88,155.52) ;
\draw  [dash pattern={on 4.5pt off 4.5pt}]  (130.82,106.48) -- (86.82,106.52) ;
\draw  [dash pattern={on 4.5pt off 4.5pt}]  (130.79,55.48) -- (86.79,55.52) ;
\draw  [color={rgb, 255:red, 189; green, 16; blue, 224 }  ,draw opacity=1 ][line width=2.25]  (305.5,57) -- (529.5,57) -- (529.5,103) -- (305.5,103) -- cycle ;
\draw  [color={rgb, 255:red, 189; green, 16; blue, 224 }  ,draw opacity=1 ][line width=2.25]  (279.5,107) -- (305.5,107) -- (305.5,223) -- (279.5,223) -- cycle ;
\draw    (394,80) -- (605.76,174.77) ;
\draw [shift={(608.5,176)}, rotate = 204.11] [fill={rgb, 255:red, 0; green, 0; blue, 0 }  ][line width=0.08]  [draw opacity=0] (8.93,-4.29) -- (0,0) -- (8.93,4.29) -- cycle    ;
\draw    (292.5,165) -- (78.49,183.74) ;
\draw [shift={(75.5,184)}, rotate = 355] [fill={rgb, 255:red, 0; green, 0; blue, 0 }  ][line width=0.08]  [draw opacity=0] (8.93,-4.29) -- (0,0) -- (8.93,4.29) -- cycle    ;
\draw  [color={rgb, 255:red, 80; green, 227; blue, 194 }  ,draw opacity=0.2 ][fill={rgb, 255:red, 248; green, 231; blue, 28 }  ,fill opacity=0.2 ] (279.5,56) -- (529.5,56) -- (529.5,108) -- (279.5,108) -- cycle ;
\draw  [color={rgb, 255:red, 80; green, 227; blue, 194 }  ,draw opacity=0.2 ][fill={rgb, 255:red, 248; green, 231; blue, 28 }  ,fill opacity=0.2 ] (279.5,106) -- (328.5,106) -- (328.5,254) -- (279.5,254) -- cycle ;
\draw    (288.5,76) -- (74.49,94.74) ;
\draw [shift={(71.5,95)}, rotate = 355] [fill={rgb, 255:red, 0; green, 0; blue, 0 }  ][line width=0.08]  [draw opacity=0] (8.93,-4.29) -- (0,0) -- (8.93,4.29) -- cycle    ;
\draw  [color={rgb, 255:red, 0; green, 0; blue, 0 }  ,draw opacity=1 ][fill={rgb, 255:red, 80; green, 227; blue, 194 }  ,fill opacity=0.2 ] (229.5,256) -- (279.5,256) -- (279.5,306) -- (229.5,306) -- cycle ;
\draw [line width=1.5]  [dash pattern={on 1.69pt off 2.76pt}]  (306,15) -- (306,496) ;
\draw [line width=1.5]  [dash pattern={on 1.69pt off 2.76pt}]  (576.5,224) -- (86.5,226) ;

\draw (148,416.4) node [anchor=north west][inner sep=0.75pt]  [font=\large]  {$c$};
\draw (496,66.4) node [anchor=north west][inner sep=0.75pt]  [font=\large]  {$c'$};
\draw (575,182.4) node [anchor=north west][inner sep=0.75pt]    {$\textsf{LIS}\ =\ v_{1}$};
\draw (-18,173.4) node [anchor=north west][inner sep=0.75pt]    {$\textsf{LIS}\ =\ v_{2}$};
\draw (58,84.4) node [anchor=north west][inner sep=0.75pt]    {$S$};

\end{tikzpicture}

\caption{This figure shows how a segment $S$ is made by two value $v_1$ and $v_2$. The solution is then updated based on the DP value for the blue cell plus $v_1+v_2$. Purple rectangles are made by sweeping lines that move to the left for the top rectangle and move down for the bottom rectangle.} \label{fig:dp-example2}
\end{figure*}
\input{figs/example-proof2}

We prove in the following that by doing so, we only lose a factor of at most $1-\epsilon'$ in the approximation. Similar to previous discussion, let $S'$ be the last multisegment that is used in Algorithm~\ref{alg:inefficient} to update the value of $g(c')$. If $S'$ fits in a row or column (meaning it is a segment), then the proof follows from our previous discussion. Thus, without loss of generality, we assume that $S'$ consists of a horizontal segment (on top) and a vertical segment (on the bottom). Let $v'_1$ be the contribution of the top horizontal part of $S'$ to the \textsf{LIS} of $S'$ and $v'_2$ be the contribution of the remainder of $S'$ to $\textsf{LIS}(S')$. We define $v_1$ and $v_2$ as the largest numbers in set $\mathbb{D}$ that are bounded by $v'_1$ and $v'_2$, respectively. In our algorithm, we consider pair of values $(v_1,v_2)$ for constructing a 2-multisegment in the following way: we first make a rectangle via a sweeping line that moves horizontally from $c'$ to the left. We do the same thing downward for $v_2$. Let the two rectangles be $R_1$ and $R_2$. If rectangle $R_1$ does not touch the top-left cell of $S'$, then $g(c')$ can be updated via a segment with parameter $v_1$ (see Figure~\ref{fig:example-proof2}). This is because without the horizontal part of $S'$, the \textsf{LIS} of the remainder of $S$ is at least $v'_2$. Thus, if we define $d'$ as the topmost part of $S'$ which is not in its horizontal part, then $g(d') \geq g(c') - v'_1$ holds. Thus, if we define $d$ to be the cell to the left and bottom of $R_1$ then $g(d) \geq g(d') \geq g(c') - v'_1$. Thus, we only lose a $1-\epsilon'$ factor if we update the value of $g(c')$ from a segment with parameter $v_1$.

If $R_1$ touches the entire horizontal part of $S$, then $R_2$ also falls within $S'$ and therefore, the 2-multisegment that our algorithm constructs will be entirely inside $S'$. Moreover, the \textsf{LIS} of the 2-multisegment that our algorithm makes is at least $v_1 + v_2 \geq (1-\epsilon') (v'_1+v'_2)$. Thus, we only lose a factor $1-\epsilon'$ in our estimation.

This approach is generalizable to larger $\Delta$. In order to update the value for $g(c')$, for any tuple of $\Delta$ elements $v_1, v_2, \ldots, v_\Delta \in \mathbb{D}$ we construct a $\Delta$-multisegment by making $\Delta$ rectangles and we estimate the value of the segment by the summation of the \textsf{LIS} of the rectangles. Next, we update the value of $g(c')$ from the constructed segment. This takes time $O(\Delta m^4 \log^{\Delta+2} n/\epsilon'^{\Delta})$ since making each $\Delta$-multisegment takes time $O(\Delta \log^2 n)$. The proof for the correctness of the approximation factor is similar to the proof explained above for $2$-multisegments.
\input{figs/dp-example3}

After a preprocessing time of $O(\Delta m^4 \log^{\Delta+2} n/\epsilon'^{\Delta})$, for every pair of table cells, we have a $(1-\epsilon')\frac{\Delta-1}{\Delta} \geq 1-\epsilon$ approximation of the \textsf{LIS} of the points included in the rectangle  between the two corners. We show that using this information, we will be able to answer each query in time $O(\log^{2\Delta+2}/\epsilon'^{2\Delta+2})$. The idea is similar to what we do above. Again, we consider solutions made by non-conflicting multisegments. The only difference is that we only take into account the points that lie inside the given query rectangle. Moreover, instead of $\Delta$-multisegments, we consider $(\Delta+1)$-multisegments here. Since we already have a desirable estimation for any rectangle that starts from a table cell and ends at another table cell, we only need to fix the bottom-left and top-right multisegment. By using the above idea, we can construct $O(\log^{\Delta+1} n/\epsilon'^{\Delta+1})$ different candidate $(\Delta+1)$-multisegments for bottom-left and  $O(\log^{\Delta+1} n/\epsilon'^{\Delta+1})$ different candidate $(\Delta+1)$-multisegments for top-right corner which amounts to $O(\log^{2\Delta+2} n/\epsilon'^{2\Delta+2})$ combinations. Thus, we can obtain an $\frac{\Delta-1}{\Delta}1-\epsilon' \geq 1-\epsilon$ approximation of the solution in time $O(\log^{2\Delta+2} n/\epsilon'^{2\Delta+2})$. Notice that once we fix the bottom-left and top-right multisegments, the solution for the area between them is already available.

The reason we use $(\Delta+1)$-multisegments instead of $\Delta$-multisegments is the following: consider the optimal non-conflicting $\Delta$-multisegments that provide the solution for a query. They may not necessarily cover the bottom-left and top-right corners of the query-rectangle. By using $(\Delta+1)$-multisegments, we can simply modify the optimal non-conflicting $\Delta$-multisegments to cover both corners as well.

\input{figs/dp-example4}

In Theorem \ref{theorem:first} we elaborate more on the above idea to show that we can approximate the query-\textsf{LIS} problem within a factor $1-\epsilon$ with near linear preprocessing time and polylogarithmic query time.

\begin{theorem}\label{theorem:first}
	For any $0 < \epsilon, \kappa$, \textit{query-\textsf{LIS}} can be approximated within a factor $1-\epsilon$ with preprocessing time $O((\log {1/\kappa}\log n/\epsilon)^{O((\log 1/\kappa)^2/\epsilon)}n^{1+\kappa})$ and query time $O((\log {1/\kappa}\log n/\epsilon)^{O((\log 1/\kappa)^2/\epsilon)})$.
\end{theorem}
\begin{proof}
	We discussed how to use extended grid packing to improve the naive algorithm with $\tilde O(n^5)$ preprocessing time to an algorithm with $\tilde O(n^{5/2})$ preprocessing time. The drawback is adding a $1-\epsilon$ multiplicative factor to the approximation guarantee and a polylogarithmic multiplicative factor to the runtime of answering each query. We use the same idea to further improve the preprocessing time down to $\ogood(n^{1+\kappa})$. To this end, we define $k = \lceil 3\log 1/\kappa \rceil + 5$ and $\epsilon' = \epsilon /(2k)$.
	
	\begin{table}[!htbp]
	\centering
	\begin{tabular}{|l|c|c|c|c|}
		\hline
		algorithm & preprocessing time & $m$ &$q_i$ & $\frac{q_{i+1}-1}{q_i-1}$\\
		\hline
		$\mathcal{A}_0$ & $\tilde O(n^5)$  & -  & 5 & $\simeq$ 0.375\\
		\hline
		$\mathcal{A}_1$  & $\tilde O(n^{5/2})$ & $n^{5/8}$ & 2.5 & $\simeq$ 0.5454\\
		\hline
		$\mathcal{A}_2$ & $\tilde O(n^{20/11})$ & $n^{5/11}$ & $\simeq$ 1.818 & $\simeq$0.6226\\
		\hline
		$\mathcal{A}_3$ & $\tilde O(n^{80/53})$ & $n^{20/53}$ & $\simeq$ 1.509 & $\simeq$0.6652\\
		\hline
		$\mathcal{A}_4$ & $\tilde O(n^{320/239})$ & $n^{80/239}$ & $\simeq$ 1.338 & $\simeq$ 0.6914 \\
		\hline
		$\mathcal{A}_5$ & $\tilde O(n^{1280/1037})$ & $n^{320/1037}$ & $\simeq$ 1.234 & $\simeq$ 0.7084 \\
		\hline
		$\mathcal{A}_6$ & $\tilde O(n^{5120/4391})$ & $n^{1280/4391}$ & $\simeq$ 1.166 & $\simeq$ 0.7201\\
		\hline
		$\mathcal{A}_7$ & $\tilde O(n^{20480/18293})$ & $n^{5120/18293}$ & $\simeq$ 1.119 & $\simeq$ 0.7282 \\
		\hline
		$\mathcal{A}_8$ & $\tilde O(n^{81920/75359})$ & $n^{20480/75359}$ & $\simeq$ 1.087 & $\simeq$  0.7340 \\
		\hline
		$\mathcal{A}_9$ & $\tilde O(n^{327680/307997})$ & $n^{81920/307997}$ & $\simeq$ 1.063 & $\simeq$ 0.7382 \\
		\hline
		$\mathcal{A}_{10}$ & $\tilde O(n^{1310720/1251671})$ & $n^{327680/1251671}$ & $\simeq$ 1.047 & $\simeq$ 0.7412\\
		\hline
	\end{tabular}
	\caption{$q_i$ is the exponent of $n$ in the preprocessing time of Algorithm $\mathcal{A}_i$.}\label{table:static}
\end{table}
	
	We denote the naive algorithm (with preprocessing time $O(n^5 \log n)$) by $\mathcal{A}_0$. Each time, we use a similar technique as explained above to obtain an improved algorithm $\mathcal{A}_i$ from $\mathcal{A}_{i-1}$. The construction of $\mathcal{A}_1$ from $\mathcal{A}_0$ is already discussed. The only parameter of the construction that changes for new algorithms is the value of $m$. To be precise, let $q_i$ be the exponent of $n$ in the preprocessing time of Algorithm $\mathcal{A}_i$ and $r_i$ be $\log m/\log n$ when we use $\mathcal{A}_{i-1}$ to construct $\mathcal{A}_{i}$. As explained, we have $q_0 = 5$. For each $0 \leq i$, we set $r_{i+1} = q_i / (q_i +3)$ and $q_{i+1} = r_{i+1} + q_i (1-r_{i+1})$. There are two steps in the preprocessing phase of Algorithm $\mathcal{A}_{i+1}$. In the first step, we construct a grid and make $2m$ subproblems and for each subproblem we use $\mathcal{A}_i$. Thus, the runtime of the first step is $$\ogood(m (n/m)^{q_i}) =  \ogood(n^{r_{i+1}} n^{q_i} / (n^{r_{i+1}})^{q_i}) = \ogood(n^{r_{i+1} + q_i (1-r_{i+1})}).$$ The second step of preprocessing is constructing an approximate solution for $O(m^4)$ subproblems each in polylogarithmic time (the exponent of the log factor may depend on $1/\epsilon$ or $1/\kappa$). This takes time $\ogood(n^{4r_{i+1}})$. Since $r_{i+1} = q_i / (q_i+3)$ then we have 
	\begin{equation*}
	\begin{split}
		r_{i+1} + q_i (1-r_{i+1}) &= q_i / (q_i+3) + q_i (1-(q_i / (q_i+3))) \\ 
										   &= q_i (1/(q_i + 3) + (1-(q_i / (q_i+3))))\\
										   &= q_i (1/(q_i + 3) + (3 / (q_i+3)))\\
										   &= q_i (4/(q_i + 3))\\
										   &= 4q_i /(q_i + 3)\\
										   &= 4r_{i+1}
	\end{split}
	\end{equation*}
	 and therefore the preprocessing time of Algorithm $\mathcal{A}_{i+1}$ would be bounded by $\ogood(n^{r_{i+1} + q_i (1-r_{i+1})}) =  \ogood(n^{q_{i+1}})$.
	
	Table~\ref{table:static} presents the runtime of each algorithm along with parameters that we use for its construction.
	
	By this construction, we always have 
	\begin{equation*}
	\begin{split}
	\frac{q_{i+1}-1}{q_i-1} &= \frac{r_{i+1} + q_i (1-r_{i+1})-1}{q_i-1}\\
							    &= \frac{ q_i / (q_i +3) + q_i (1-q_i / (q_i +3))-1}{q_i-1}\\
							    &= \frac{ q_i / (q_i +3) + q_i (3 / (q_i +3))-1}{q_i-1}\\
							    &= \frac{4q_i / (q_i +3)-1}{q_i-1}\\
							    &= \frac{(3q_i-3) / (q_i +3)}{q_i-1}\\
							    &= \frac{3}{q_i+3}\\
	 							& \leq 3/4 \label{inequality:last}.
	\end{split}
	\end{equation*}
	where the last inequality follows from the fact that $1 < q_i$. Therefore after $k$ levels of recursive calls, the exponent of $n$ in the preprocessing step of Algorithm $A_k$ is smaller than $1+\kappa$. This implies that the overall preprocessing time is bounded by $\ogood(n^{1+\kappa})$. Moreover, by defining $\Delta = \lceil 1/\epsilon'\rceil$, in each level of recursion we lose a factor of at most $(1-\epsilon')^2$ in the approximation and thus in total the approximation factor is at least $(1-\epsilon')^{2k} \geq 1-\epsilon$. Query time for $\mathcal{A}_0$ is $O(\log n)$ and in each level of recursion, there is a multiplicative $O((\log n/\epsilon')^{2k/\epsilon+2k})$ overhead and thus the query time as well as the preprocessing time is multiplied by $O((\log {1/\kappa}\log n/\epsilon)^{O((\log 1/\kappa)^2/\epsilon)})$ for $\mathcal{A}_k$.
\end{proof}

\subsection{A Dynamic Solution for \textsf{LIS}}\label{sec:realdynamic}
We present a dynamic algorithm for \textsf{LIS} that approximates the solution within a $1-\epsilon$ multiplicative factor. This is achieved by using extended grid packing. For the purpose of our technique, we generalize the problem in the following way: In the original dynamic problem, our goal is to maintain an approximation to the size of the longest increasing subsequence. Thus after each operation, one needs to update the solution size and therefore, we can simply assume that after each update the size of the \textsf{LIS} is desired. In our generalization, we still consider the same set of operations. However, our algorithm is required to answer a stronger type of queries: Assuming that we describe the sequence as points on the plane (as discussed earlier), for each query, we provide a rectangle with edges parallel to the axis lines and the algorithm should give us an estimate for the \textsf{LIS} of the points inside the rectangle. This is similar to the queries that we consider in the query-\textsf{LIS} problem.

Since in the previous setting, we only asked the size of the longest increasing subsequence, one could argue that after each update, there is at most one question to be answered. In our new setting, we may need to answer multiple queries after each update and thus it makes sense to separate the concept of query from operations. That is, we may be able to answer each query faster than the time which our algorithm requires for updating the sequence. Thus in our setting, we define the update time to be the time our algorithm needs to update the sequence after an operation arrives and the query time to be the time our algorithm needs to answer a query. We show in the following that we can answer each query faster than the update time and this is an important part of the analysis. 
 
We show that for any $0 < \epsilon,\kappa$ our algorithm is able to provide a $1-\epsilon$ approximation of the solution with update time $\ogood(n^{\kappa})$ and query time $\ogood(1)$. Recall that $\ogood$ hides all the factors that depend only on $\epsilon$, $\kappa$, or $\log n$. An exact dynamic solution for \textsf{LIS} with update time $O(n^5 \log n)$ and query time $O(\log n)$ follows from the straightforward algorithm discussed in Section~\ref{sec:example}. In other words, if after every operation, we compute the solution for all possible rectangles that can be given to the algorithm as a query, then we can answer each query in time $O(\log n)$. Moreover, since there are at most $O(n^4)$ rectangles that cover distinct sets of points, we can update the solution in time $O(n^5 \log n)$ each time an operation arrives. Notice that after answering a query, we do not require to update the solution. To improve the update time, we use the notion of block-based algorithms~\cite{our-stoc-paper}. Roughly speaking, if we design a block-based algorithm whose bound over the update time is amortized, we can then turn the block-based algorithm into an equivalent dynamic algorithm with a worst-case bound on its update time.

For a block-based algorithm, we start with an initial sequence of size $n$. Our algorithm is then allowed to preprocess the input in time $f(n)$. After the preprocessing step, our block-based algorithm is responsible for $g(n)$ operations and all the queries that come prior to the last operation. After $g(n)$ operations, our algorithm terminates. Mitzenmacher and Seddighin~\cite{our-stoc-paper} prove that if the block-based algorithm performs each operations in time $h(n)$ in the worst case, then it can be transformed into a dynamic algorithm whose worst-case update time is $O(f(n)/g(n)+h(n))$. The approximation factor of the dynamic algorithm is exactly the same as the block-based algorithm and also the query time remains asymptotically the same. Since Mitzenmacher and Seddighin~\cite{our-stoc-paper} use a slightly different setting in which queries and operations are treated equivalently, we bring a formal proof for our discussion in Appendix~\ref{appendix:block-based}.

To improve the naive algorithm, we design a block-based algorithm with preprocessing time $f(n) = O(n^{19/7} \log n)$, $g(n) = n^{3/7}$, and $h(n) = \tilde O_{\epsilon}(n^{16/7})$. In the preprocessing step, we construct a grid of size $m \times m$ where $m = n^{4/7}$. Rows and columns of our grid evenly divide the points and thus each row or column covers $O(n^{3/7})$ points. Similar to what we discussed in Section~\ref{sec:example}, we make $2m$ subproblems for dynamic \textsf{LIS} where each subproblem is concerned with the subset of points that is covered by the corresponding row or the corresponding column of the grid. The construction of the grid only requires to sort the points based on $x$ and $y$ coordinates and thus we can do it in time $O(n \log n)$. Moreover, each subproblem includes $O(n^{3/7})$ points and therefore precomputing the solutions for all possible rectangles within each subproblem takes time $O(n^{15/7} \log n)$. Since we run this procedure for all $2m$ subproblems, processing time is $O(n^{19/7} \log n)$.

\begin{figure}

\centering

\includegraphics[width=10cm]{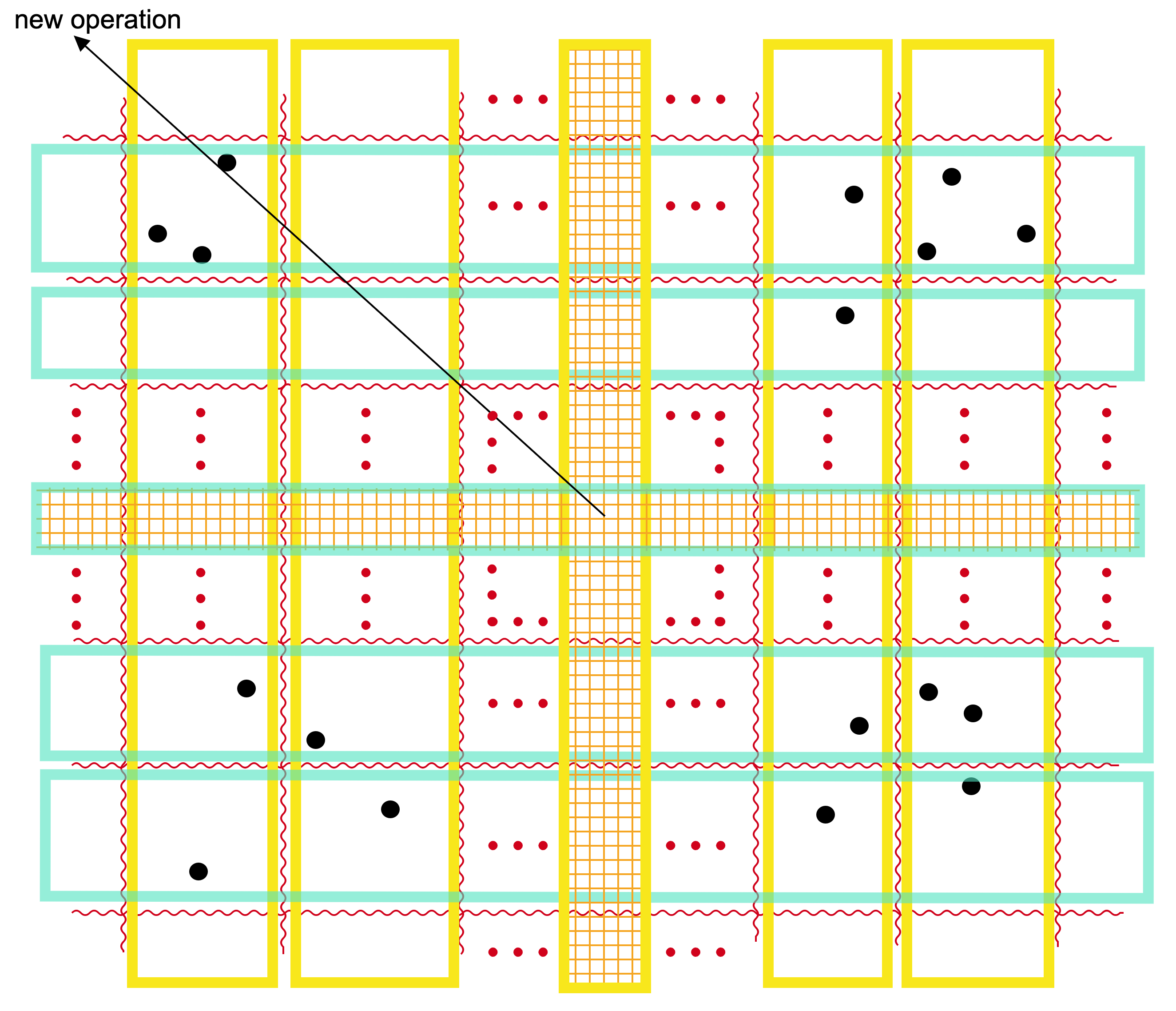}

\caption{The red snake lines show the grid and yellow and blue boxes present the points included in each of the subproblems. Patterned boxes specify the subproblems that are affected by an operation.} \label{fig:dynamicgrid}
\end{figure}

Due to our construction, each time an operation arrives, there is only one row and one column which is affected by the modification. Also, since the size of each subproblem (the number of points covered by each row or column) is bounded by $O(n^{3/7})$ then by precomputing the solution for all possible rectangles, we can update the entries for each subproblem in time $O(n^{15/7} \log n)$. In addition to updating the solution for each subproblem, we also maintain a table of size $m^2 \times m^2$  that keeps an approximation to the value of \textsf{LIS} for each rectangle formed by the cells of the grid. As discussed in Section~\ref{sec:example}, using the extended grid packing, we can approximate the solution for each rectangle within a factor $1-\epsilon$ and doing this for all possible rectangles takes time $\tilde O_{\epsilon}(m^4)$. Therefore, after each operation, we update our table in time $\tilde O_{\epsilon}(m^4) = \tilde O_{\epsilon}(n^{16/7})$. Similar to the algorithm of Section~\ref{sec:example}, each query can be approximated via the precomputed solutions for the rectangles and the solutions for subproblems in time $\tilde O_{\epsilon}(1)$. Since $g(n) = n^{3/7}$, we are sure that after $g(n)$ operations, the size of each subproblem remains bounded by $O(n^{3/7})$ and thus the runtimes do not increase asymptotically as operations add new points to the subproblems. Therefore, our block-based algorithm has preprocessing time $f(n) = O(n^{19/7} \log n)$, $g(n) = n^{3/7}$, and $h(n) = \tilde O_{\epsilon}(n^{16/7})$. This leads to a dynamic solution with worst-case update time $\tilde O_{\epsilon}(n^{16/7})$ and approximation factor $1-\epsilon$. Also, the query complexity of our dynamic algorithm is $\tilde O_{\epsilon}(1)$.

\input{figs/example1}

Since inserting elements to the sequence or removing elements from the sequence may change the indices of other elements, the points on the plane are subject to moves. Each time an operation arrives, we update the solution for the corresponding subproblems. Let us be more specific about this. Initially, $m-1$ vertical lines divide the array into chunks of size roughly $n/m$. As operations arrive, the elements are shifted to the left or to the right (their indices are updated). Each vertical line can be thought of as a separator between two consecutive elements that is also shifted to the left or to the right when elements are added or removed. Thus, although the vertical lines move, each element which is inserted or deleted lies between two thresholds and corresponds to a unique column of the grid. The corresponding row is uniquely determined by the horizontal lines (those lines remain unchanged). Thus, every element insertion or deletion affects only one cell of the grid which is included in at most two subproblems. Mitzenmacher and Seddighin~\cite{our-stoc-paper} show that the shifts can be efficiently done by an $O(\log n)$ overhead in the runtime of the algorithm. More precisely, they give a data structure that is able to insert and delete elements from the sequence while giving access to any position of the sequence in time $O(\log n)$. We use the same data structure and therefore we incorporate an additional $O(\log n)$ overhead in the runtime of our dynamic algorithms.

\input{figs/example2}

In what follows, we bring a recursive application of the above ideas that leads to a solution with update time $\ogood(n^\kappa)$ and approximation factor $1-\epsilon$ for arbitrarily small $0 < \epsilon, \kappa$. The query time of our algorithm is $\ogood(1)$.

\begin{theorem}\label{theorem:second}
	For any $0 < \epsilon, \kappa$, there exists a dynamic algorithm for \textsf{LIS} with worst-case update time $O((\log n/(\epsilon \kappa))^{O(1/(\epsilon \kappa^2))}n^{\kappa})$ and approximation factor $1-\epsilon$ and query time $O((\log n/(\epsilon \kappa))^{O(1/(\epsilon \kappa^2))})$.
\end{theorem}
\begin{proof}
Similar to the proof of Theorem~\ref{theorem:first}, we construct several algorithms which we denote by $\mathcal{A}_0, \mathcal{A}_1, \mathcal{A}_2, \ldots$ and $\mathcal{A}'_1, \mathcal{A}'_2, \mathcal{A}'_3, \ldots$. Each $\mathcal{A}'_i$ represents a block-based algorithm and each $\mathcal{A}_i$ is a dynamic algorithm. We begin by a dynamic algorithm $\mathcal{A}_0$  for \textsf{LIS} with preprocessing and update times $\ogood(n^{1+\kappa/2})$ and query time $\ogood(1)$. $\mathcal{A}_0$ basically uses the non-dynamic algorithm of Theorem~\ref{theorem:first} to update the sequence after each operation. At each step, we first construct $\mathcal{A}'_{i+1}$ from $\mathcal{A}_i$ and then using the reduction of~\cite{our-stoc-paper}, we turn $\mathcal{A}'_{i+1}$ into a dynamic algorithm $\mathcal{A}_{i+1}$.

As shown by Mitzenmacher and Seddighin~\cite{our-stoc-paper}, if we construct a dynamic algorithm $\mathcal{A}$ from a block-based algorithm $\mathcal{A}'$ then in order to initialize  $\mathcal{A}$ for a sequence of length $n$, we only need to spend preprocessing time $f(n)$ corresponding to $\mathcal{A}'$. For constructing $\mathcal{A}'_{i+1}$ from $\mathcal{A}_i$ we set $m = n^{\kappa/4}$. In the construction block-based algorithm $\mathcal{A}'_{i+1}$, we set $f(n) = \ogood(n^{1+\kappa/2})$, $g(n) = n^{1-\kappa/4}$ and $h(n) = \ogood(n^{\kappa} + n^{(1+\kappa/2)(1-\kappa/4)^{i+1}})$. As a result, the update time of our dynamic algorithm $\mathcal{A}_{i+1}$ is always bounded by $\ogood(n^{\kappa} + n^{(1+\kappa/2)(1-\kappa/4)^{i+1}})$. Each time we construct a grid and divide the problem into smaller subproblems as explained above. In what follows, we analyze the time bounds.

To initialize Algorithm $\mathcal{A}'_{i+1}$, we draw a grid of size $n^{\kappa/4} \times n^{\kappa/4}$ in time $O(n \log n)$ in a way that each row and each column covers $O(n^{1-\kappa/4})$ points. Then, for each row and column we use Algorithm $\mathcal{A}_i$ to initialize a solution for the corresponding subproblem. Since for $\mathcal{A}_0$ the initialization time is $\ogood(n^{1+\kappa/2})$, the same bound carries over to the preprocessing times of all other algorithms $\mathcal{A}'_1, \mathcal{A}'_2, \ldots$ (The initialization times can only improve as we increase $i$). Also, since $m = n^{\kappa/4}$ the $\ogood(m^4)$ part of the preprocessing time is dominated by the $\ogood(n^{1+\kappa/2})$ term in the preprocessing time.

For each $\mathcal{A}_{i+1}$, in order to update the solution, we need to update the subproblems concerning the corresponding row and the corresponding column. Since each row and each column cover at most $O(n^{1-\kappa/4})$ elements, then this implies that the update time $h(n)$ for $\mathcal{A}'_1$ is equal to $\ogood(n^{\kappa} + (n^{1-\kappa/4})^{1+\kappa/2})$. Moreover, since for $\mathcal{A}_1$, the dominant term in $f(n)/g(n) + h(n)$ is $h(n)$, this implies that the update time of the dynamic Algorithm $\mathcal{A}_{i}$ is equal to $\ogood(n^{\kappa} + (n^{1-\kappa/4})^{1+\kappa/2})$. Thus, we can inductively prove that for each $i \geq 1$, the update time for block-based algorithm $\mathcal{A}'_{i}$ is equal to the update time of the dynamic algorithm $\mathcal{A}_{i}$  which is bounded by $\ogood(n^{\kappa} + (n^{(1+\kappa/2)(1-\kappa/4)^{i}})$. By setting $k = \lceil 20/{\kappa}\rceil$, we can be sure that the update time of Algorithm $\mathcal{A}_k$ is bounded by $\ogood(n^{\kappa})$. Also, the query time remains $\ogood(1)$ for all the algorithms.

In what follows, we discuss the approximation factor and the runtimes that are suppressed by the $\ogood$ notation. We set the approximation factor of our first algorithm $\mathcal{A}_0$ equal to $1-\epsilon/2$ and for the rest of the constructions we use $\epsilon' = \epsilon / (40k)$. Thus, the overall approximation factor will be $(1-\epsilon/2)(1-\epsilon')^k \geq 1-\epsilon$. This also adds a multiplicative overhead $O(\log n/(\epsilon \kappa))^{O(1/(\epsilon \kappa))}$ to the preprocessing time, update time, and query time of each level of recursion and therefore the update time for the final algorithm would be $O((\log n/(\epsilon \kappa))^{O(1/(\epsilon \kappa^2))}n^{\kappa})$.

\end{proof}

By setting $\kappa = 1/(\log n)^{1/3}$ in Theorem~\ref{theorem:second}, we obtain an algorithm with update time $$O((\log n/\epsilon)^{O((\log n)^{2/3}/\epsilon)}).$$ Since in this case the query time and update time are equal, we can use this algorithm for the original dynamic problem wherein the queries and the operations are treated the same way.

\begin{theorem}[a corollary of Theorem~\ref{theorem:second}]\label{theorem:third}
For any $0 < \epsilon$, there exists a dynamic algorithm for \textsf{LIS} with worst-case update time $O((\log n/\epsilon)^{O((\log n)^{2/3}/\epsilon)})$ and approximation factor $1-\epsilon$.
\end{theorem}

	\bibliographystyle{alphaurl}
	\bibliography{draft}
	
	\newpage
	
	\appendix
	\newpage
\section{Block-based Algorithms~\cite{our-stoc-paper}}\label{appendix:block-based}
We present the block-based framework of Mitzenmacher and Seddighin~\cite{our-stoc-paper} in this section. They show a reduction that simplifies the problem with respect to worst-case time constraints. Ultimately, in our algorithms, we prove that the update time of each operation is bounded in the worst case. However, it is more convenient to allow for larger update times in some cases, while keeping a bounded amortized update time. 

In the framework of Mitzenmacher and Seddighin~\cite{our-stoc-paper}, we start with an array $a$ of size $n$ and our algorithm is allowed to make a preprocessing of time $f(n)$. For the next $g(n)$ steps, the processing time of each operation is bounded by $h(n)$ in the worst case. After $g(n)$ steps, our algorithm is no longer responsible for the operations and terminates. We refer to such an algorithm as \textit{block-based}. Note $f(n)$, $g(n)$, and $h(n)$ are determined based only on the size $n$ of the initial sequence $a$. We only consider $g(n) \leq n/2+10$ so that after $g(n)$ operations the size of the sequence remains asymptotically the same. 

We show in the following that a block-based algorithm $\mathcal{A}$ for \textsf{LIS} with identifiers $\langle f,g,h \rangle$ can be used as a black box to obtain a dynamic algorithm $\mathcal{A'}$ with worst-case update time $O(\max\{h(n), f(n) / g(n)\})$. The approximation factor of the algorithm is preserved in this reduction. 
 
\begin{lemma}[proven by Mitzenmacher and Seddighin~\cite{our-stoc-paper}]\label{lemma:reduction}
	Let $\mathcal{A}$ be a block-based algorithm with preprocessing time $f(n)$ that approximates dynamic \textsf{LIS} for up to $g(n)$ many steps with worst-case update time $h(n)$. If $g(n) \leq n/2+10$ then there exists a dynamic algorithm $\mathcal{A}'$ for the same problem whose worst-case update time is bounded by $O(\max\{h(n), f(n) / g(n)\})$ and whose approximation factor is the same as $\mathcal{A}$.
\end{lemma}
\begin{figure}[ht]

\tikzset{every picture/.style={line width=0.75pt}} 

\begin{tikzpicture}[x=0.75pt,y=0.75pt,yscale=-1,xscale=1]

\draw   (39,29) -- (301.5,29) -- (301.5,52) -- (39,52) -- cycle ;
\draw [color={rgb, 255:red, 208; green, 2; blue, 27 }  ,draw opacity=1 ] [dash pattern={on 4.5pt off 4.5pt}]  (239,29) -- (239,79) ;
\draw [shift={(239,81)}, rotate = 270] [fill={rgb, 255:red, 208; green, 2; blue, 27 }  ,fill opacity=1 ][line width=0.75]  [draw opacity=0] (10.72,-5.15) -- (0,0) -- (10.72,5.15) -- (7.12,0) -- cycle    ;

\draw   (239,88) -- (501.5,88) -- (501.5,111) -- (239,111) -- cycle ;
\draw [color={rgb, 255:red, 208; green, 2; blue, 27 }  ,draw opacity=1 ] [dash pattern={on 4.5pt off 4.5pt}]  (439,32) -- (439,138) ;
\draw [shift={(439,140)}, rotate = 270] [fill={rgb, 255:red, 208; green, 2; blue, 27 }  ,fill opacity=1 ][line width=0.75]  [draw opacity=0] (10.72,-5.15) -- (0,0) -- (10.72,5.15) -- (7.12,0) -- cycle    ;

\draw   (436,143) -- (698.5,143) -- (698.5,166) -- (436,166) -- cycle ;
\draw [color={rgb, 255:red, 189; green, 16; blue, 224 }  ,draw opacity=1 ][fill={rgb, 255:red, 189; green, 16; blue, 224 }  ,fill opacity=1 ][line width=1.5]    (246,41) -- (296.5,41) ;
\draw [shift={(299.5,41)}, rotate = 180] [color={rgb, 255:red, 189; green, 16; blue, 224 }  ,draw opacity=1 ][line width=1.5]    (14.21,-6.37) .. controls (9.04,-2.99) and (4.3,-0.87) .. (0,0) .. controls (4.3,0.87) and (9.04,2.99) .. (14.21,6.37)   ;
\draw [shift={(243,41)}, rotate = 0] [color={rgb, 255:red, 189; green, 16; blue, 224 }  ,draw opacity=1 ][line width=1.5]    (14.21,-6.37) .. controls (9.04,-2.99) and (4.3,-0.87) .. (0,0) .. controls (4.3,0.87) and (9.04,2.99) .. (14.21,6.37)   ;
\draw [color={rgb, 255:red, 189; green, 16; blue, 224 }  ,draw opacity=1 ][fill={rgb, 255:red, 189; green, 16; blue, 224 }  ,fill opacity=1 ][line width=1.5]    (446,100) -- (496.5,100) ;
\draw [shift={(499.5,100)}, rotate = 180] [color={rgb, 255:red, 189; green, 16; blue, 224 }  ,draw opacity=1 ][line width=1.5]    (14.21,-6.37) .. controls (9.04,-2.99) and (4.3,-0.87) .. (0,0) .. controls (4.3,0.87) and (9.04,2.99) .. (14.21,6.37)   ;
\draw [shift={(443,100)}, rotate = 0] [color={rgb, 255:red, 189; green, 16; blue, 224 }  ,draw opacity=1 ][line width=1.5]    (14.21,-6.37) .. controls (9.04,-2.99) and (4.3,-0.87) .. (0,0) .. controls (4.3,0.87) and (9.04,2.99) .. (14.21,6.37)   ;
\draw  [dash pattern={on 0.84pt off 2.51pt}]  (116.45,134.57) -- (272.5,100) ;
\draw [shift={(272.5,100)}, rotate = 347.51] [color={rgb, 255:red, 0; green, 0; blue, 0 }  ][fill={rgb, 255:red, 0; green, 0; blue, 0 }  ][line width=0.75]      (0, 0) circle [x radius= 3.35, y radius= 3.35]   ;
\draw [shift={(114.5,135)}, rotate = 347.51] [fill={rgb, 255:red, 0; green, 0; blue, 0 }  ][line width=0.75]  [draw opacity=0] (8.93,-4.29) -- (0,0) -- (8.93,4.29) -- cycle    ;
\draw  [dash pattern={on 0.84pt off 2.51pt}]  (116.5,135.12) -- (469.5,156) ;
\draw [shift={(469.5,156)}, rotate = 3.39] [color={rgb, 255:red, 0; green, 0; blue, 0 }  ][fill={rgb, 255:red, 0; green, 0; blue, 0 }  ][line width=0.75]      (0, 0) circle [x radius= 3.35, y radius= 3.35]   ;
\draw [shift={(114.5,135)}, rotate = 3.39] [fill={rgb, 255:red, 0; green, 0; blue, 0 }  ][line width=0.75]  [draw opacity=0] (8.93,-4.29) -- (0,0) -- (8.93,4.29) -- cycle    ;
\draw  [fill={rgb, 255:red, 155; green, 155; blue, 155 }  ,fill opacity=1 ] (20,124) -- (64.5,124) -- (64.5,147) -- (20,147) -- cycle ;
\draw  [fill={rgb, 255:red, 255; green, 255; blue, 255 }  ,fill opacity=1 ] (65,124) -- (109.5,124) -- (109.5,147) -- (65,147) -- cycle ;
\draw [line width=0.75]  [dash pattern={on 4.5pt off 4.5pt}]  (301.5,56) -- (301.5,116) ;

\draw [line width=0.75]  [dash pattern={on 4.5pt off 4.5pt}]  (501.5,30) -- (501.5,167) ;

\draw    (94.5,138) -- (94.5,167) ;
\draw [shift={(94.5,169)}, rotate = 270] [fill={rgb, 255:red, 0; green, 0; blue, 0 }  ][line width=0.75]  [draw opacity=0] (8.93,-4.29) -- (0,0) -- (8.93,4.29) -- cycle    ;

\draw    (42.5,113) -- (42.5,142) ;

\draw [shift={(42.5,111)}, rotate = 90] [fill={rgb, 255:red, 0; green, 0; blue, 0 }  ][line width=0.75]  [draw opacity=0] (8.93,-4.29) -- (0,0) -- (8.93,4.29) -- cycle    ;

\draw (23,39) node   {$\mathcal{B}_{1}$};
\draw (41,66) node [scale=0.9]  {$a^{(1)}$};
\draw (223,98) node   {$\mathcal{B}_{2}$};
\draw (420,153) node   {$\mathcal{B}_{3}$};
\draw (556,153) node [scale=2.488]  {$\dotsc $};
\draw (67,96) node  [align=left] {preprocessing};
\draw (132,176) node  [align=left] {updating two operations in each step};
\draw (269,62) node [scale=0.7,color={rgb, 255:red, 144; green, 19; blue, 254 }  ,opacity=1 ]  {$g(n_1) /10$};
\draw (470,121) node [scale=0.7,color={rgb, 255:red, 144; green, 19; blue, 254 }  ,opacity=1 ]  {$g(n_2) /10$};
\draw (242,126) node [scale=0.9]  {$a^{(2)}$};
\draw (439,181) node [scale=0.9]  {$a^{(3)}$};
\draw (43,15) node [scale=0.7]  {$1$};
\draw (240,15) node [scale=0.7]  {$\frac{9g(n_1)}{10}$};
\draw (301,13) node [scale=0.7]  {$g(n_1)$};
\draw (512,15) node [scale=0.7]  {$\frac{9g(n_1)}{10} +g(n_2)$};
\draw (428,15) node [scale=0.7]  {$\frac{9( g(n_1) +g(n_2))}{10}$};

\end{tikzpicture}

\caption{The reduction is shown in this figure.} \label{fig:reduction}
\end{figure}
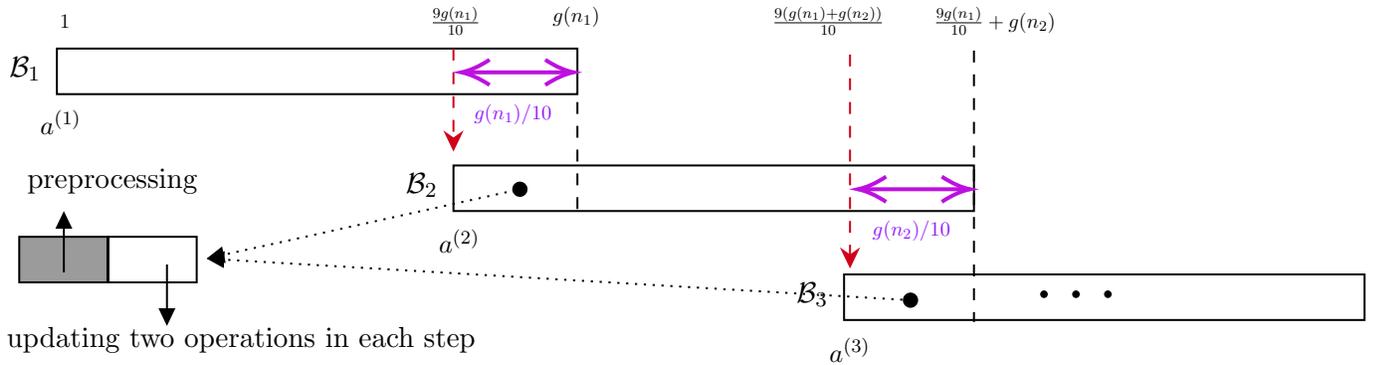

Figure \ref{fig:reduction} gives a pictorial depiction of the proof idea for Lemma~\ref{lemma:reduction}. Mitzenmacher and Seddighin~\cite{our-stoc-paper} construct an algorithm $\mathcal{A'}$ in the following way: $\mathcal{A}'$ uses algorithm $\mathcal{A}$ repeatedly. To distinguish between multiple instances of $\mathcal{A}$, we add subscripts; the first time we use algorithm $\mathcal{A}$ we call it $\mathcal{B}_1$. Every $\mathcal{B}_i$ is basically a copy of the block-based algorithm $\mathcal{A}$ which is modified slightly to execute the preprocessing part in multiple steps. We begin with using our block-based algorithm $\mathcal{B}_1$ at step 1. At this point we call the initial array (which is empty) $a^{(1)}$. Also, we refer to its size by $n_1$ which is equal to 0. Since the size of the array is constant, so is the preprocessing time and therefore we can ignore it when bounding the time complexity. For $g(n_1)$ many steps, we use algorithm $\mathcal{B}_1$ to preserve an approximate solution and from then on, we use a separate algorithm for the rest of the operations, namely $\mathcal{B}_2$. The construction of $\mathcal{B}_2$ is given below:

When $\mathcal{B}_1$ has gone $9/10$ of the way and is only responsible for $g(n_1)/10$ more operations, we initiate algorithm $\mathcal{B}_2$. Let $a^{(2)}$ be the array at this point and $n_2$ be its size. $\mathcal{B}_2$ needs to run the preprocessing step which requires $f(n_2)$ many operations. This may not be possible in a single step, therefore, we break the computation into $g(n_1)/20$ pieces and execute each piece in the next $g(n_1)/20$ steps. Moreover, in the next $g(n_1)/20$ steps operations that arrive after the construction of $\mathcal{B}_2$ are processed: two operations in each step. While this is happening, algorithm $\mathcal{B}_1$ processes the operations and updates the solution size. $g(n_1)/10$ many steps after the construction of $\mathcal{B}_2$, algorithm $\mathcal{B}_2$ has already finished the preprocessing and all the operations that have arrived so far are applied to it.  This is exactly the time that $\mathcal{B}_1$ terminates, and from then on, we use algorithm $\mathcal{B}_2$ to process each operation. Similarly, $\mathcal{B}_3$ is constructed when $\mathcal{B}_2$ has applied $g(n_2) 9/10$ operations. This construction goes on as long as operations arrive.

We emphasize that there is a constant-factor overhead in the update-time of the reduction which is hidden in the $O$ notation. Moreover, the constant factor is regardless of algorithm $\mathcal{A}$ and is the same for all possible algorithms. It follows from the same idea that if we separate the notion of query from the notion of operation (as we do in Section~\ref{sec:dynamic}), still the same reduction gives us a dynamic algorithm from a block-based algorithm.

	\newpage
\section{The Algorithm of Chen, Chu, and Pinsker~\cite{chen2013dynamic}}\label{sec:chen-ap}
For formal proofs, we refer the reader to~\cite{chen2013dynamic}. Chen, Chu, and Pinsker~\cite{chen2013dynamic} propose the following algorithm to maintain a solution for dynamic \textsf{LIS}.

For each element $i$ of the array, define $l(i)$ to be the size of the longest increasing subsequence ending at element $a_i$ of the array. Chen, Chu, and Pinsker~\cite{chen2013dynamic} refer to this quantity as the \textit{level} of element $i$. Notice that $l(i)$ can be computed in time $\tilde O(n)$ for all elements of the array using the patience sorting algorithm.

Define $L'_k$ to be the set of elements whose levels are equal to $k$. The algorithm of Chen, Chu, and Pinsker~\cite{chen2013dynamic} maintains a balanced binary tree for each $L'_k$ that contains the corresponding elements. One key observation is that for each $k$, all the elements  of $L'_k$ are decreasing, otherwise their levels would not be the same.

When a new element is added to the array, $L'_k$'s may change. More precisely, after an element addition, the levels of some elements may change (but only by 1). Similarly, element removal may change the levels of the elements of the array but again the change is bounded by $1$. Chen, Chu, and Pinsker~\cite{chen2013dynamic}, show that after an insertion, for each $L'_k$, the levels of only one interval of the elements may increase. In other words, for each $L'_k$, there are two numbers $\alpha$ and $\beta$ such that all the elements whose values are within $[\alpha,\beta]$ increase their levels and the rest remain in $L'_k$.

Thus, they use a special balanced tree structure that allows for interval deletion and interval addition in logarithmic time. Therefore, all that remains is to detect which interval of each $L'_k$ changes after each operation. They show that this can be computed in time $O(\log n)$ for all $L'_k$'s via binary search. Since the number of different levels is equal to the size of the \textsf{LIS}, their update time depends on the size of the solution.

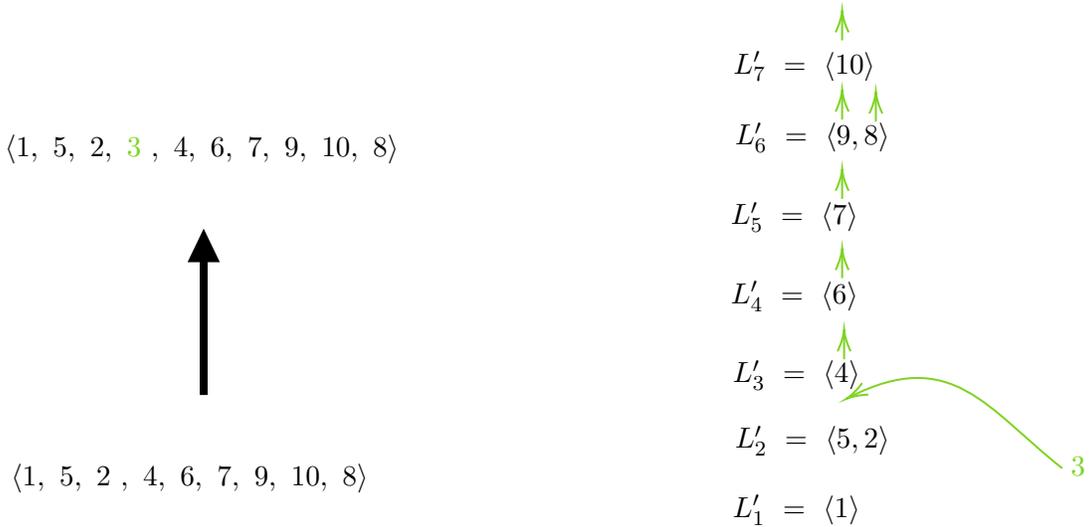
\begin{figure}[ht]

\centering

\tikzset{every picture/.style={line width=0.75pt}} 

\begin{tikzpicture}[x=0.75pt,y=0.75pt,yscale=-1,xscale=1]

\draw [line width=3]    (163,205) -- (163,126) ;
\draw [shift={(163,121)}, rotate = 450] [fill={rgb, 255:red, 0; green, 0; blue, 0 }  ][line width=3]  [draw opacity=0] (16.97,-8.15) -- (0,0) -- (16.97,8.15) -- cycle    ;

\draw [color={rgb, 255:red, 126; green, 211; blue, 33 }  ,draw opacity=1 ]   (596,242) .. controls (558.38,212.3) and (540.36,179.66) .. (488.58,206.18) ;
\draw [shift={(487,207)}, rotate = 332.15] [color={rgb, 255:red, 126; green, 211; blue, 33 }  ,draw opacity=1 ][line width=0.75]    (10.93,-3.29) .. controls (6.95,-1.4) and (3.31,-0.3) .. (0,0) .. controls (3.31,0.3) and (6.95,1.4) .. (10.93,3.29)   ;

\draw [color={rgb, 255:red, 126; green, 211; blue, 33 }  ,draw opacity=1 ][fill={rgb, 255:red, 126; green, 211; blue, 33 }  ,fill opacity=1 ]   (486,187) -- (486,174) ;
\draw [shift={(486,172)}, rotate = 450] [color={rgb, 255:red, 126; green, 211; blue, 33 }  ,draw opacity=1 ][line width=0.75]    (10.93,-3.29) .. controls (6.95,-1.4) and (3.31,-0.3) .. (0,0) .. controls (3.31,0.3) and (6.95,1.4) .. (10.93,3.29)   ;

\draw [color={rgb, 255:red, 126; green, 211; blue, 33 }  ,draw opacity=1 ][fill={rgb, 255:red, 126; green, 211; blue, 33 }  ,fill opacity=1 ]   (485,146) -- (485,133) ;
\draw [shift={(485,131)}, rotate = 450] [color={rgb, 255:red, 126; green, 211; blue, 33 }  ,draw opacity=1 ][line width=0.75]    (10.93,-3.29) .. controls (6.95,-1.4) and (3.31,-0.3) .. (0,0) .. controls (3.31,0.3) and (6.95,1.4) .. (10.93,3.29)   ;

\draw [color={rgb, 255:red, 126; green, 211; blue, 33 }  ,draw opacity=1 ][fill={rgb, 255:red, 126; green, 211; blue, 33 }  ,fill opacity=1 ]   (485,106) -- (485,93) ;
\draw [shift={(485,91)}, rotate = 450] [color={rgb, 255:red, 126; green, 211; blue, 33 }  ,draw opacity=1 ][line width=0.75]    (10.93,-3.29) .. controls (6.95,-1.4) and (3.31,-0.3) .. (0,0) .. controls (3.31,0.3) and (6.95,1.4) .. (10.93,3.29)   ;

\draw [color={rgb, 255:red, 126; green, 211; blue, 33 }  ,draw opacity=1 ][fill={rgb, 255:red, 126; green, 211; blue, 33 }  ,fill opacity=1 ]   (485,66) -- (485,53) ;
\draw [shift={(485,51)}, rotate = 450] [color={rgb, 255:red, 126; green, 211; blue, 33 }  ,draw opacity=1 ][line width=0.75]    (10.93,-3.29) .. controls (6.95,-1.4) and (3.31,-0.3) .. (0,0) .. controls (3.31,0.3) and (6.95,1.4) .. (10.93,3.29)   ;

\draw [color={rgb, 255:red, 126; green, 211; blue, 33 }  ,draw opacity=1 ][fill={rgb, 255:red, 126; green, 211; blue, 33 }  ,fill opacity=1 ]   (502,67) -- (502,54) ;
\draw [shift={(502,52)}, rotate = 450] [color={rgb, 255:red, 126; green, 211; blue, 33 }  ,draw opacity=1 ][line width=0.75]    (10.93,-3.29) .. controls (6.95,-1.4) and (3.31,-0.3) .. (0,0) .. controls (3.31,0.3) and (6.95,1.4) .. (10.93,3.29)   ;

\draw [color={rgb, 255:red, 126; green, 211; blue, 33 }  ,draw opacity=1 ][fill={rgb, 255:red, 126; green, 211; blue, 33 }  ,fill opacity=1 ]   (485,26) -- (485,13) ;
\draw [shift={(485,11)}, rotate = 450] [color={rgb, 255:red, 126; green, 211; blue, 33 }  ,draw opacity=1 ][line width=0.75]    (10.93,-3.29) .. controls (6.95,-1.4) and (3.31,-0.3) .. (0,0) .. controls (3.31,0.3) and (6.95,1.4) .. (10.93,3.29)   ;

\draw (462,263) node   {$L'_{1} \ =\ \langle 1\rangle$};
\draw (470,228) node   {$L'_{2} \ =\ \langle 5,2\rangle$};
\draw (462,195) node   {$L'_{3} \ =\ \langle 4\rangle$};
\draw (162,81) node   {$\langle 1,\ 5,\ 2,\ \color{rgb, 255:red, 126; green, 211; blue, 33 }3\color{black}\ ,\ 4,\ 6,\ 7,\ 9,\ 10,\ 8\rangle$};
\draw (156,247) node   {$\langle 1,\ 5,\ 2\ ,\ 4,\ 6,\ 7,\ 9,\ 10,\ 8\rangle$};
\draw (461,155) node   {$L'_{4} \ =\ \langle 6\rangle$};
\draw (461,115) node   {$L'_{5} \ =\ \langle 7\rangle$};
\draw (470,75) node   {$L'_{6} \ =\ \langle 9,8\rangle$};
\draw (466,39) node   {$L'_{7} \ =\ \langle 10\rangle$};
\draw (604,241) node [color={rgb, 255:red, 126; green, 211; blue, 33 }  ,opacity=1 ]  {$3$};

\end{tikzpicture}

\caption{This example illustrates how an element addition is handled in the algorithm of Chen, Chu, and Pinsker~\cite{chen2013dynamic}. Inserting element $3$ to the array changes the levels of the elements. Upward arrows show that the level of the corresponding element increases after we add $3$ to the array.}\label{fig:chan}
\end{figure}

When $n$ elements are given, their runtime for constructing the data structure is $\tilde O(n)$ since patience sorting gives us all the levels in time $\tilde O(n)$ and the balanced trees can be constructed in time $\tilde O(n)$ for all $L'_k$.

	
\end{document}